\documentclass[11pt,a4paper]{article}

\title{Inverse Scattering and Locality in Integrable\\ Quantum Field Theories}
\author{Sabina Alazzawi$^*$ and Gandalf Lechner$^\dagger$\\\\
\footnotesize{ $^*$Zentrum Mathematik, Technische Universit\"at M\"unchen}\\
\footnotesize{Boltzmannstra\ss e 3, 85748 Garching, Germany}\\
\footnotesize{sabina.alazzawi@tum.de}\\\\
\footnotesize{ $^\dagger$School of Mathematics, Cardiff University}\\
\footnotesize{Senghennydd Road, Cardiff, UK, CF24 4AG}\\
\footnotesize{LechnerG@cardiff.ac.uk}}
\date{}
\usepackage{geometry}
\usepackage{mathrsfs}
\usepackage{stmaryrd}
\usepackage[pdfborder={0 0 0}]{hyperref}

\usepackage{undertilde}
\newcommand{\triplenorm}{\ensuremath{| \! | \! |}}
\usepackage{amsmath}
\usepackage{amsthm}
\usepackage{amssymb}
\usepackage{amsfonts}

\usepackage{enumitem}
\usepackage{tikz}

\usepackage{textcomp}
\usepackage{url}
\usepackage{bbold}
\usepackage[ansinew]{inputenc}
\newtheorem{theorem}{Theorem}[section]
\newtheorem{lemma}[theorem]{Lemma}

\newtheorem{proposition}[theorem]{Proposition}
\newtheorem{corollary}[theorem]{Corollary}
\newtheorem{definition}[theorem]{Definition}
\usepackage{fancyhdr}

\usepackage{wrapfig}
\numberwithin{equation}{section}

\newcommand{\id}{{\rm id}}
\newcommand{\spann}{{\rm span}}
\newcommand{\im}{{\rm Im\,}}
\newcommand{\dom}{{\rm dom}}
\newcommand{\Tr}{{\rm Tr}}
\newcommand{\K}{\mathcal K}
\newcommand{\X}{\mathcal X}
\newcommand{\Y}{\mathcal Y}
\newcommand{\I}{\mathcal I}
\newcommand{\F}{\mathcal F}
\newcommand{\B}{\mathcal B}
\newcommand{\W}{\mathcal W}
\newcommand{\Tu}{\mathcal T}
\newcommand{\balpha}{{\boldsymbol{\alpha}}}
\newcommand{\bbeta}{{\boldsymbol{\beta}}}
\newcommand{\bgamma}{{\boldsymbol{\gamma}}}
\newcommand{\bla}{{\boldsymbol{\lambda}}}
\newcommand{\bte}{{\boldsymbol{\theta}}}
\newcommand{\bzeta}{{\boldsymbol{\zeta}}}
\newcommand{\te}{\theta}
\newcommand{\Rl}{\mathbb{R}}
\newcommand{\Cl}{\mathbb{C}}
\newcommand{\C}{\mathcal{C}}
\newcommand{\Nl}{\mathbb{N}}
\newcommand{\Hil}{\mathscr{H}}

\newcommand{\CC}{\mathscr{C}}
\newcommand{\Ss}{\mathscr{S}}
\newcommand{\Cti}{\tilde{C}}
\newcommand{\CCh}{\hat{\mathscr{C}}}

\newcommand{\SF}{\mathcal{S}}
\renewcommand{\O}{\mathcal{O}}

\newcommand{\eps}{\varepsilon}
\newcommand{\con}{^{\rm con}}
\newcommand{\zd}{z^\dagger}
\newcommand{\Strip}{{\mathbb S}}
\newcommand{\tp}[1]{^{\otimes #1}}

\begin{document}

\maketitle

\begin{abstract}
	We present a solution method for the inverse scattering problem for integrable two-dimensional relativistic quantum field theories, specified in terms of a given massive single particle spectrum and a factorizing S-matrix. An arbitrary number of massive particles transforming under an arbitrary compact global gauge group is allowed, thereby generalizing previous constructions of scalar theories. The two-particle S-matrix $S$ is assumed to be an analytic solution of the Yang-Baxter equation with standard properties, including unitarity, TCP invariance, and crossing symmetry.
	
	Using methods from operator algebras and complex analysis, we identify sufficient criteria on $S$ that imply the solution of the inverse scattering problem. These conditions are shown to be satisfied in particular by so-called diagonal S-matrices, but presumably also in other cases such as the $O(N)$-invariant nonlinear $\sigma$-models.
\end{abstract}

\section{Introduction and Overview}\label{intro}

This paper is part of a research program on the non-perturbative construction and analysis of integrable relativistic quantum field theories in two dimensions, prominent examples being the Sinh-Gordon model, the Ising model, the Sine-Gordon model, the $O(N)$ $\sigma$-models, and many more. Such field theories are simple enough to be accessible by a range of different methods, and yet rich enough to sometimes resemble features of QFT in four dimensions (see, for example \cite{AAR,GrosseWulkenhaar:2014}). The literature on integrable quantum field theories in general is so voluminous that we have to restrict ourselves to mention the monographs \cite{AAR,Smir92,Ketov:2000,Baxter:1982} as just a few sample references. 

The main focus of the present article is the {\em construction} (in a sense to be made precise) of a large family of such models. In some cases, a construction with the tools of constructive quantum field theory in the Euclidean setting \cite{glimm1981quantum} has been accomplished a long time ago. In particular, the Sine-Gordon model was constructed from its classical Lagrangian by quantization and perturbative renormalization by Fröhlich  \cite{Frohlich:1975}. Most other models, however, have not been established in a non-perturbative manner yet. 

It has been known for a long time that due to the presence of infinitely many conservation laws, the dynamics (scattering) in integrable quantum field theories is severely restricted, to the extent that the particle number is conserved in collisions of arbitrary energy, and the full S-matrix is completely determined in terms of the two-particle S-matrix (``factorizing S-matrix'') \cite{AAR,iagolnitzer1993scattering}. Moreover, also the form of the two-particle S-matrix is subject to many constraints. This, on the one hand, often allows to determine the two-particle S-matrix from kinematical reasoning alone (up to certain mild ambiguities), and, on the other hand, suggests to use the (quantum) two-particle S-matrix instead of the (classical) Lagrangian as the datum defining the theory, and the starting point of its construction and analysis.

This inverse scattering point of view underlies the bootstrap form-factor program \cite{babujian2006form,Smir92}, where the aim is to calculate the Wightman $n$-point functions of local quantum fields associated with a given factorizing S-matrix. Using fundamental QFT properties like locality and covariance in conjunction with analytic and algebraic properties of the S-matrix, it is often possible to explicitly compute form factors of the theory (see, for example, \cite{babujian2011n}). The Wightman $n$-point functions are then given by infinite series of integrals over form factors, and for a non-perturbative construction, it would be necessary to control the convergence properties of these series. While this has been possible in a few special cases \cite{babujian2006form}, this problem remains largely open in general.

The present paper follows an alternative inverse scattering approach to the construction of integrable quantum field theories, see \cite{schroer1997modular,schroer2000modular,lechner2003polarization,L08} for the initial papers in this program, \cite{BischoffTanimoto:2013, bostelmann2014characterization, LS, CadamuroTanimoto:2015,Tanimoto:2016} for more recent developments, and \cite[Ch.~10]{Lechner:AQFT-book:2015} for a review. While the aim is, as in the form factor program, to construct an integrable quantum field theory from a given two-particle S-matrix, the tools that are used in the construction are different. In particular, the framework of algebraic quantum field theory \cite{Haag} is used instead of the framework of Wightman quantum field theory \cite{streater2000pct}.

\bigskip

In the following, we give a detailed overview of the different steps of the construction and the contents of the paper, and its relations to other developments within this program. In several aspects, our work is a generalization of \cite{L08}, where scalar theories with a single species of particles without internal degrees of freedom were considered, to general particle spectra.

\medskip

Our starting point is thus a single particle spectrum consisting of an arbitrary finite number of positive masses and a finite number of charges, corresponding to some arbitrary compact global gauge group, as well as a (two-body) S-matrix of arbitrary size. These data have to satisfy a number of conditions, out of which we mention here in particular a strict PCT symmetry. Another essential property is that the S-matrix has to be an analytic crossing-symmetric solution of the Yang-Baxter equation, a property that is trivially fulfilled in the scalar case. In Sections~\ref{section:1particle-spectrum} and \ref{smatrix}, we specify our assumptions in detail and recall from \cite{LS} how to build a suitable vacuum Hilbert space $\Hil$ from them. This space is a generalized Fock space, defined by an $S$-dependent representation of the permutation group, and carries a representation of the Zamolodchikov-Faddeev algebra \cite{zamolodchikov1979factorized} by ``$S$-twisted'' creation/annihilation operators (see Prop.~\ref{proposition:z-operators}~$iii)$).

The framework we are using is not the most general one. Two possible extensions are to consider only meromorphic (instead of analytic) S-matrices, with the poles of $S$ related to the bound states of the theory, or infinite-dimensional $S$, as is necessary for describing situations in which the gauge group is no longer compact. We refer to \cite{CadamuroTanimoto:2015,Tanimoto:2016} and \cite{HollandsLechner:2016} for new developments in these two directions. It should be noted, however, that our setup already allows for a huge family of models, including the $O(N)$ $\sigma$-models as particular examples.

\bigskip

Having fixed the single particle spectrum and S-matrix as the input data, the basic strategy of the construction proceeds, as in the scalar case, through two main steps. In the first step, carried out in \cite{LS} and reviewed here in Section~\ref{SecWedgeLocal}, one exploits the crossing symmetry of $S$ (see Def.~\ref{S-matrixDefinition}~$iv)$) to explicitly construct a pair of ``wedge-local'' quantum fields $\phi_\alpha(x),\phi'_\beta(y)$. 
\begin{wrapfigure}{r}{0.35\textwidth}
  \begin{center}
    \includegraphics[width=0.34\textwidth]{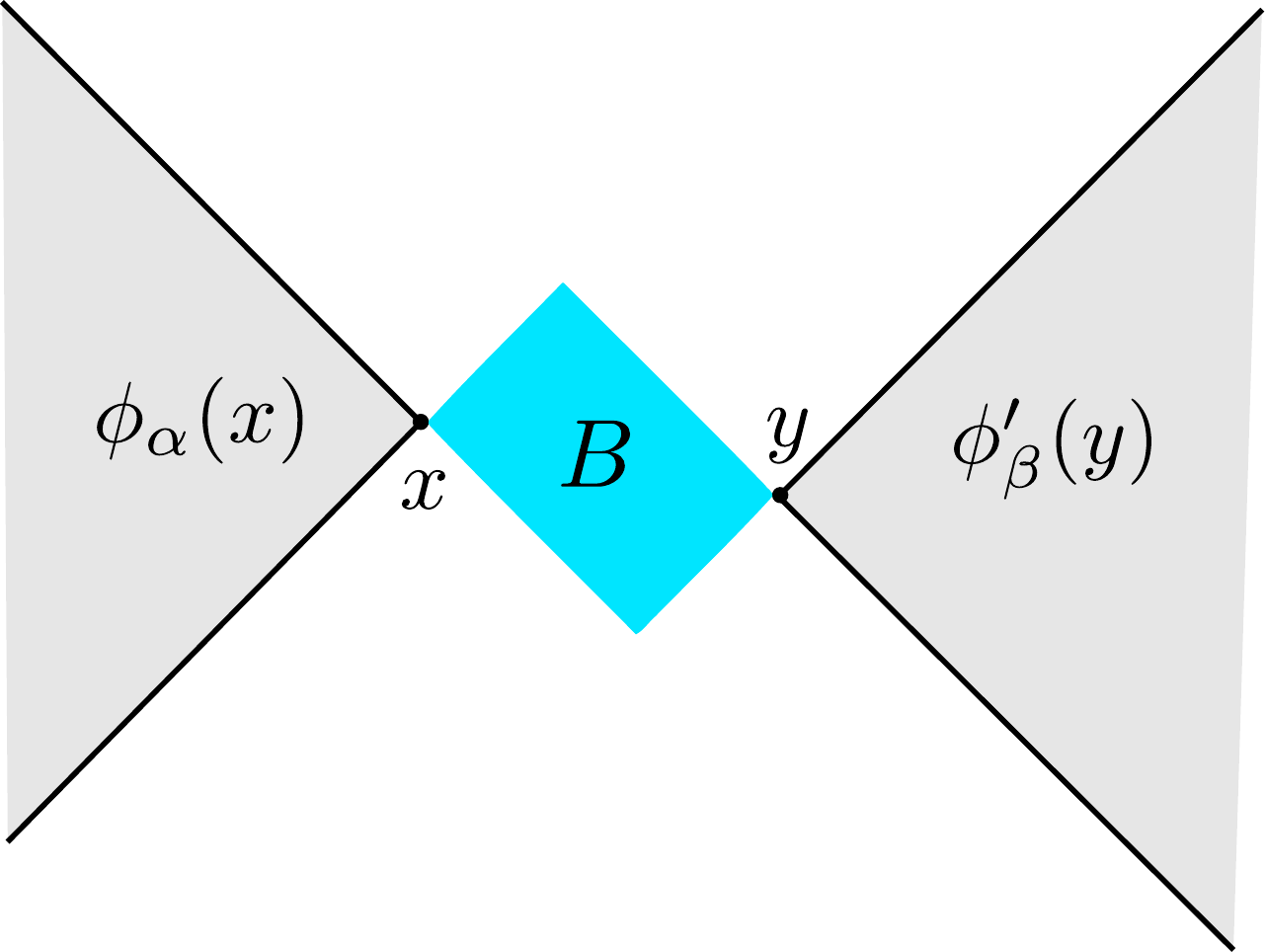}
 \end{center}
\end{wrapfigure}
These fields, labeled by spacetime points $x,y\in\Rl^2$ and indices $\alpha,\beta$ parameterizing the representation of the gauge group, commute with each other if $x$ lies spacelike to $y$ and {\em to the left} of $y$ (in a relativistic sense, see picture). This is tantamount to saying that these field operators are localized in (unbounded) ``wedge regions'' like the so-called right wedge $W_R\subset\Rl^2$, defined as the set of all points lying to the right of the origin. Both $\phi$ and $\phi'$ furthermore transform covariantly under the Poincar\'e group and the gauge group, and encode the two-body S-matrix in the vacuum expectation values of their products $\phi_\alpha(x)\cdot\phi_\beta(y)$ (Thm.~\ref{theorem:fields}).

\medskip

The second step of the construction consists in finding fields/observables $B$ of direct physical interest, which are localized in {\em finite} spacetime regions. By locality, this step involves in particular the solution of the commutator constraints $[B,\phi_\alpha(x)]=0$, $[B,\phi'_\beta(y)]=0$, if $B$ is to be localized as shown in the picture above. 

As the explicit solution of such commutator constraints is a complicated open problem even in the scalar case \cite{bostelmann2014characterization}, we rather use a more abstract approach. It is at this stage that we make use of the operator-algebraic description of quantum field theory \cite{Haag,Araki99}, which provides us with tools that are not readily available in other approaches. Our construction program will therefore result in an algebraic description of a quantum field theory (a ``net of local algebras'') instead of a sequence of $n$-point functions.

\medskip

Instead of the fields, we consider the von Neumann algebra $\F(W_R)$ generated by all fields $\phi_\alpha'(y)$, $y\in W_R$. As a consequence of the properties of the field operators $\phi,\phi'$, this algebra has the vacuum vector $\Omega$ as a cyclic and separating vector, and therefore defines a modular operator $\Delta$ \cite{BratteliRobinson:1987}. In this setting, the {\em modular nuclearity condition} of Buchholz, D'Antoni and Longo \cite{BDL90}, demanding that the map
\begin{align}\label{eq:modnucintro}
	\F(W_R) \ni A \longmapsto \Delta^{1/4}U(x)A\Omega\in\Hil\,,\qquad x\in W_R\,,
\end{align}
is nuclear (see Def.~\ref{DefNuclear}), is important. (Here $U(x)$ denotes the unitary representing translation by $x$ on $\Hil$.)

Namely, it is known from \cite{BL4,L08} that {\em if} this nuclearity condition holds, then the existence of ``large'' algebras (type III$_1$) of {\em local} fields/observables $B$ with cyclic vacuum vector follows. Moreover, in this case our construction automatically yields a solution of the inverse scattering problem, as has been shown in \cite{LS}. 

The main task in completing the second step of the construction program is therefore to decide for which $S$ the modular nuclearity condition holds. We recall the detailed definition of this condition and its implications in Section~\ref{Section:Operator-algebraic-Formulation}, where we also compute the relevant modular operators for the considered family of models (Bisognano-Wichmann property, Prop.~\ref{BisognWich}).

\medskip

As we shall explain in the body of the text, the Bisognano-Wichmann property entails that the abstract modular nuclearity condition takes a concrete form and can be investigated with tools from complex analysis of several variables. More specifically, we consider for any operator $A\in\F(W_R)$ the functions $(A\Omega)_n^\balpha(\bte):=\langle\zd_{\alpha_1}(\te_1)\cdots\zd_{\alpha_n}(\te_n)\Omega,\,A\Omega\rangle$, where the $\te_j$ are rapidity variables and the $\zd_\alpha(\te)$ the  corresponding Zamolodchikov creation operators. 

\medskip

It is a crucial aspect of our approach that the {\em existence} of local fields/observables is encoded in the complex analytic structure of the functions $(A\Omega)_n^\balpha$, and can be inferred without explicitly constructing such local operators. We identify two properties of particular significance in this context, ``property (H)'' (with ``H'' for ``Hardy'') and the ``intertwiner property''. 

\medskip

Property (H) demands that the $(A\Omega)_n^\balpha$ (as functions of $\te_1,...,\te_n\in\Rl$) have an analytic continuation into an $n$-dimensional tube domain of a particular form, and obey suitable Hardy-type bounds on this tube (Def.~\ref{definition:property-h}). This property was already known to hold in free field theories and scalar integrable models. We give here an abstract proof that property (H) implies an $n$-particle version of modular nuclearity (Thm.~\ref{thm:general-nuclearity}). 

\begin{wrapfigure}{r}{0.35\textwidth}
  \begin{center}
   \begin{tikzpicture}[baseline=(current bounding box.center),scale=1.0,line width=1.0pt,>=latex]
	       \draw[line width=2.0pt, gray] (0,0) to (2,0) {};
	       \draw[line width=2.0pt, gray] (2.5,0) to (4.5,0) {};
	       \draw[->] (1.3,0) to (1.3,1.5) {};
		   \draw[->] (0.1,0) to (0.1,1.5) {};
		   \draw[->] (0.5,0) to (0.5,1.5) {};
	       \node[above] at (0.9,0.1) {$\dots$};
	       \draw[->] (3.1,1.5) to [out=-90,in=-90,looseness=2.4] (1.6,1.5) {};
	       \draw[->] (4.3,1.5) to [out=-90,in=-90,looseness=0.7] (0.3,1.5) {};
	       \draw[<-] (3.3,0) to (3.3,1.5) {};
			\draw[<-] (2.6,0) to (2.6,1.5) {};
			\draw[<-] (4.1,0) to (4.1,1.5) {};
			\node[above] at (3.7,0.1) {$\dots$};
 \end{tikzpicture}
 \end{center}
\end{wrapfigure}

In Section~\ref{ChapterHardyS}, we then investigate the validity of property (H) concretely. It turns out to be most efficient to represent various analytic continuations of the  $(A\Omega)_n^\balpha$ as sums of ``contractions diagrams'' such as the one shown on the right, where each line carries an index $\alpha$ and a rapidity $\te$, the orientations of the lines distinguish between creation and annihilation operators, and crossings correspond to $S$-factors. Reading such diagrams analogously to knot partition functions (\cite{Kauffman:1993}, see Section~\ref{Section:Analyticity+Combinatorics} for precise definitions) then allows us to conveniently organize and extract the analytic and combinatorial properties of the functions $(A\Omega)_n^\balpha$.

This investigation works for general underlying S-matrix $S$, and in particular clarifies the role of the crossing symmetry of $S$, which acts on (partial) diagrams according to 	 
	 \begin{equation*}
          \begin{tikzpicture}[baseline=(current bounding box.center),scale=1.0,line width=1.0pt,>=latex]
	       \draw[line width=2.0pt, gray] (0,0) to (1,0) {};
	       \draw[line width=2.0pt, gray] (1.5,0) to (2.5,0) {};
	       \node at (0.5,0.4) {$\cdots$};
	       \draw[<-] (1.7,0) to (1.7,1.5) {};
	       \draw[->] (2.2,1.5) to [out=-90, in=-90,looseness=1.5] (.5,1.5) {};
	       \node[above] at (0.5,1.5) {\small$\mu$};
	       \node[above] at (1.7,1.5) {${\alpha}\atop{[\te]}$};
	        \node[above] at (2.2,1.5) {${\nu}\atop{[\te']}$};
	       \node[above] at (1.9,0) {\small$\lambda$};
	 \end{tikzpicture}
	 \xrightarrow{\quad\te\to\te-i\pi\quad}
	 \begin{tikzpicture}[baseline=(current bounding box.center),scale=1.0,line width=1.0pt,>=latex]
	       \draw[line width=2.0pt, gray] (0,0) to (1,0) {};
	       \draw[line width=2.0pt, gray] (1.5,0) to (2.5,0) {};
	       \node at (0.5,0.4) {$\cdots$};
	       \draw[->] (0.8,0) to (0.8,1.5) {};
	       \draw[->] (2.2,1.5) to [out=-90, in=-90,looseness=1.5] (.5,1.5) {};
	       \node[above] at (0.42,1.5) {\small$\mu$};
	       \node[above] at (0.8,1.5) {${\overline{\alpha}}\atop{[\te]}$};
	        \node[above] at (2.2,1.5) {${\nu}\atop{[\te']}$};
	       \node[above] at (0.95,0) {\small$\overline\lambda$};
	 \end{tikzpicture}
	 \quad,
	 \end{equation*}
	 (precisely explained in Section~\ref{Section:Analyticity+Combinatorics}). Using these methods, we give a proof that property (H) holds for a large class of ``regular'' S-matrices (Prop.~\ref{Corrhardystructure}).
	 
\bigskip

The second property relevant to the modular nuclearity condition is the so-called ``intertwiner property''. It is inspired by the fact that in many models (in particular, in the $O(N)$-models), the $S$-matrix reduces to a tensor flip with a {\em negative} sign at zero rapidity transfer. The intertwiner property demands that the ($S$-dependent) representation of the permutation groups $\mathfrak{S}_n$, which underlies the definition of the vacuum Hilbert space, can be intertwined in a suitable manner with the representation in which $S$ is replaced by the negative tensor flip (Def.~\ref{definition:property-I}). This intertwiner is required to preserve both, algebraic and analytic properties coming from property (H).

On a technical level, the intertwiner property serves as a tool to improve estimates (in dependence on the particle number $n$) on how well the $n$-particle projections of \eqref{eq:modnucintro} can be approximated by finite-dimensional maps. For all S-matrices satisfying the intertwiner property, we obtain ``full'' modular nuclearity (as opposed to just an $n$-particle version), and the inverse scattering problem for the underlying $S$ is solved (Thm.~\ref{mainTheorem}).

\medskip

The explicit characterization of all $S$ satisfying the intertwiner property amounts to characterizing a certain cohomology class of analytic cocycles of $\mathfrak{S}_n$-actions on tubes in $\Cl^n$. We do not completely solve this problem here, but rather construct a family of non-trivial examples of (diagonal) S-matrices with the intertwiner property (Prop.~\ref{lemDiag}). We furthermore provide partial evidence to the effect that also the $O(N)$-models have the intertwiner property. 

\bigskip

In summary, our results translate the inverse scattering problem for integrable models into a problem in complex analysis of several variables, related to the investigation of property (H) and the intertwiner property. Whereas property (H) is shown to hold for all regular S-matrices, the exact range of validity of the intertwiner property remains to be determined. However, for a large class of non-scalar S-matrices a solution of the inverse scattering problem is obtained. We expect that (slight variations of) our methods will also apply to models with bound states, as they are currently being analyzed by Cadamuro and Tanimoto \cite{CadamuroTanimoto:2015}.

This paper is based on the PhD thesis of one of us \cite{Alazzawi:2014}.

\section{Quantum field theories with factorizing S-matrices}

\subsection{Single particle spectrum}\label{section:1particle-spectrum}

The starting point of our construction is the specification of the single particle spectrum.
We allow for finitely many massive particle species carrying arbitrary charges. In this section, we recall the basic setup regarding the representations of the corresponding Poincar\'e and gauge groups, following \cite{LS}.

We consider a compact group $G$ as global gauge group, and a finite set $\mathcal{Q}$ of equivalence classes $q$ of unitary irreducible representations of $G$, interpreted as charge quantum numbers as usual. For simplicity, we restrict ourselves to the case that to each charge $q$, there exists exactly one mass $m(q)>0$. This will, in particular, ensure that the models we construct contain only massive, stable particles\footnote{Our results can be shown to hold also if finitely many isolated mass values are considered in each sector.}. The mass gap of the theory is, therefore positive, and will be denoted 
\begin{align}\label{eq:mass-gap}
	m_\circ
	:=
	\min\{m(q)\,:\,q\in\mathcal{Q}\}
	>0
	\,.
\end{align}
Since we are working on two-dimensional Minkowski space, we may parameterize the upper mass shell $H_{m(q)}^+=\{((p^2+m(q)^2)^{1/2},p):p\in\mathbb{R}\}$ by the rapidity $\theta$, that is
\begin{equation}\label{p}
p_{m(q)}(\theta):=m(q)
\begin{pmatrix}
\cosh\theta\\
\sinh\theta
\end{pmatrix}
,\qquad\theta\in\mathbb{R}.
\end{equation}
Choosing $L^2(\mathbb{R},\theta)$ as a natural representation space for the irreducible positive energy representation $U_{1,m}$ of the Poincar\'e group with mass $m$, we define the one-particle Hilbert space as $\mathscr{H}_{1}:=L^2(\mathbb{R},d\theta)\otimes\mathcal{K}$, where $\mathcal{K}$ is a finite-dimensional Hilbert space on which the gauge group~$G$ acts. This space decomposes into subspaces of fixed charge $q\in\mathcal{Q}$ and mass $m(q)$,
\begin{equation}\label{hilbert1}
\mathscr{H}_1=\bigoplus_{q\in\mathcal{Q}}L^2(\mathbb{R},d\theta)\otimes\mathcal{K}_q.
\end{equation}
The proper orthochronous Poincar\'e group $\mathcal{P}_+^\uparrow$ acts on $\mathscr{H}_1$ by means of the unitary, strongly continuous representation
\begin{equation}\label{U1}
U_1
:=
\bigoplus_{q\in\mathcal{Q}}\left(U_{1,m(q)}\otimes \id_{\mathcal{K}_q}\right),
\end{equation}
which satisfies the relativistic spectrum condition with ``mass gap'' $m_\circ>0$, i.e. the joint spectrum of the generators \mbox{$P=(P^0,P^1)$} of the translations is contained in $\{p\in\mathbb{R}^2:p_0\geq(p_1^2+m_\circ^2)^{1/2}\}$, and the mass operator $M:=((P^0)^2-(P^1)^2)^{1/2}$ has spectrum $\{m(q)\,:\,q\in\cal Q\}$.

The global gauge group $G$, on the other hand, is represented on $\mathscr{H}_1$ by the unitaries
\begin{equation}\label{V1}
	V_1(g)
	:=
	\bigoplus_{q\in\mathcal{Q}}
	\left(\id_{L^2(\mathbb{R},d\theta)}\otimes V_{1,q}(g)\right)
	,\qquad g\in G,
\end{equation}
where $V_{1,q}$ is an irreducible representation of $G$ in the class $q$. Clearly, $V_1$ and $U_1$ commute.\\

For several calculations, it will be useful to consider an orthonormal basis for each $\mathcal{K}_q$. Then their direct sum, denoted by $\{e^\alpha:\alpha=1,\dots,\dim\K\}$, constitutes an orthonormal basis of $\mathcal{K}$. Each index $\alpha$ thus corresponds to some charge $q_{[\alpha]}$ and mass $m_{[\alpha]}:=m(q_{[\alpha]})$, and $\theta\mapsto \psi^\alpha(\theta)$ denotes the respective component of a vector $\psi\in\mathscr{H}_1$. In this basis, the Poincar\'e representation reads
\begin{equation}
\left(U_{1}(a,t)\psi\right)^\alpha(\theta):=e^{ip_{m_{[\alpha]}}(\theta)\cdot a}\,\psi^\alpha(\theta-2\pi t)\,,
\end{equation}
where $a\in\mathbb{R}^2$ denotes the spacetime translation, and $t$ the parameter of the (rescaled) Lorentz boost.

For our subsequent analysis, a PCT operator will be essential, and we introduce it here on the one-particle level. For the existence of such an operator it is in particular necessary to assume that with each $q\in\cal Q$, also the complex conjugate representation $\bar q$ is contained in $\cal Q$, as we shall do from now on. Charge conjugation then exchanges $q$ and $\bar q$, and can be expressed in the basis $\{e^\alpha\}_\alpha$ by means of an involutive permutation, which we denote $\alpha\mapsto\bar\alpha$. That is, we have an antiunitary involution $\Gamma$ on $\K$, namely $(\Gamma v)^\alpha=\overline{v^{\bar\alpha}}$, which commutes with $V_1$. Together with the well-known TP operator for $U_{1,m}$ (complex conjugation), this defines our one-particle PCT operator $J_1$ as 
\begin{align}\label{eq:J1}
	(J_1\psi)^\alpha(\theta):=\overline{\psi^{\bar\alpha}(\theta)}\,.
\end{align}
The index notation referring to the basis $\{e^\alpha\}_\alpha$ is also used for tensor products: With $(\cdot,\cdot)$ the scalar product in $\mathcal{K}$, we define for vectors $v\in\mathcal{K}^{\otimes n}$ and tensors $M:\mathcal{K}^{\otimes m}\rightarrow\mathcal{K}^{\otimes n}$, $m,n\in\mathbb{N}$,
\begin{eqnarray}
	v^{\alpha_1\dots\alpha_n}
	&:=&
	(e^{\alpha_1}\otimes\cdots\otimes e^{\alpha_n},v)
	\\
	M^{\alpha_1\dots\alpha_n}_{\beta_1\dots\beta_m}
	&:=&
	(e^{\alpha_1}\otimes\cdots\otimes e^{\alpha_n},Me^{\beta_1}\otimes\cdots\otimes e^{\beta_m}).
\end{eqnarray}
When the length of a multi index is clear from the context, we also write $\balpha=(\alpha_1,...,\alpha_n)$, etc. 

Given $M\in\mathcal{B}(\mathcal{K}^{\otimes 2})$ and $n\geq 2$, another useful notation will be
\begin{equation}\label{kurznotation}
	M_{n,k}:=1_{k-1}\otimes M\otimes 1_{n-k-1},\qquad k=1,\dots,n-1,
\end{equation}
where $1_j$ denotes the identity on $\mathcal{K}^{\otimes j}$ and $M_{n,k}\in\mathcal{B}(\mathcal{K}^{\otimes n})$.

\bigskip

The Hilbert space $\mathscr{H}_1$ and the representations $U_1,V_1$ can be second quantized as usual: On the unsymmetrized Fock space
\begin{equation}\label{fockUnsy}
\widehat{\mathscr{H}}:=\bigoplus_{n=0}^\infty\mathscr{H}_1^{\otimes n}\simeq\bigoplus_{n=0}^\infty(L^2(\mathbb{R}^n,d^n\theta)\otimes\mathcal{K}^{\otimes n})
\,,
\end{equation}
we have the natural representations $\widehat{U}:=\bigoplus_{n=0}^\infty U_1^{\otimes n}$ of $\mathcal{P}_+^\uparrow$ and $\widehat{V}:=\bigoplus_{n=0}^\infty V_1^{\otimes n}$ of $G$. In our index notation, $\widehat{U}$ acts according to
\begin{equation}\label{actionPoincare}
     \big(\widehat{U}(a,t)\Psi\big)_n^{\boldsymbol{\alpha}}(\boldsymbol{\theta})
     =
     e^{i\sum_{k=1}^{n}p_{m_{[\alpha_k]}}(\theta_k)\cdot a}\Psi_n^{\boldsymbol{\alpha}}(\theta_1-2\pi t,\dots,\theta_n-2\pi t),\qquad(a,t)\in\mathcal{P}_+^\uparrow\,.
\end{equation}

Instead of passing to a symmetric or antisymmetric subspace of $\widehat{\mathscr{H}}$, we will work on a different subspace. To define it, we first need to introduce a suitable two-particle S-matrix.

\subsection{Two-particle S-Matrices and S-Symmetric Fock Space}\label{smatrix}

In our inverse scattering program, a unitary two-particle S-matrix $S$ is the most important input into the construction. In fact, we will use such an object in the very definition of the vacuum Hilbert space of our models. Below we give an abstract definition of a two-particle S-matrix. Its properties (Def.~\ref{S-matrixDefinition}) are clearly motivated by scattering theory \cite{iagolnitzer1993scattering}, but for the time being, $S$ will just serve as an algebraic object which induces a certain symmetrization procedure. We will comment on its significance as the $2\to2$ part of a factorizing S-matrix later on.

\bigskip

In order to formulate the properties of $S$ in a basis independent manner, we make use of the antiunitary involution $\Gamma$ on $\mathcal{K}$ introduced earlier, $(\Gamma v)^\alpha=\overline{v^{\bar\alpha}}$, and the flip operator $F$ on $\mathcal{K}\otimes\mathcal{K}$,
\begin{equation}\label{flip}
	F:\mathcal{K}\otimes\mathcal{K}\to\mathcal{K}\otimes\mathcal{K}
	,\qquad 
	F(u\otimes v):=v\otimes u.
\end{equation}

\begin{definition}\label{S-matrixDefinition}
	An $S$-matrix is a continuous bounded matrix-valued function $S:\{\zeta\in\mathbb{C}:0\leq\im \,\zeta\leq\pi\}\rightarrow\mathcal{B}(\mathcal{K}\otimes\mathcal{K} )$, which is analytic in the strip  $\Strip(0,\pi):=\{\zeta\in\mathbb{C}:0<\im\,\zeta<\pi\}$, and satisfies for arbitrary $\theta,\theta'\in\mathbb{R}$, the following properties:
	\begin{enumerate}
		\item Unitarity:\hspace{.7cm} $$S(\theta)^*=S(\theta)^{-1}.$$
		\item Hermitian analyticity:\hspace{.7cm} $$S(\theta)^{-1}=S(-\theta).$$
		\item Yang-Baxter equation:\\
		$$(S(\theta)\otimes 1_1)(1_1\otimes S(\theta+\theta'))(S(\theta')\otimes 1_1)=(1_1\otimes S(\theta'))(S(\theta+\theta')\otimes 1_1)(1_1\otimes S(\theta)).$$
		\item Crossing symmetry:\hspace{.7cm} $$\left(\Gamma u_1\otimes u_2 ,S(i\pi -\theta)\, v_1\otimes \Gamma v_2\right)=\left(u_2\otimes v_2,S(\theta)\,u_1\otimes v_1\right),\quad \forall\, u_1,\,u_2,\,v_1,\,v_2\in\mathcal{K}.$$
		\item PCT invariance:\hspace{.7cm} $$\left(\Gamma\otimes\Gamma\right)\, F\,S(\theta)\,F\,\left(\Gamma\otimes\Gamma\right)=S(-\theta).$$
		\item Translational invariance\footnote{This condition is equivalent to $S$ commuting with the translation unitaries $U_1(a,0)\otimes U_1(a,0)$, $a\in\Rl^2$.}: Let $M=\sum_m E_m$ denote the spectral decomposition of the mass operator $M$ of $U_1$, and $S$ the operator on ${\mathscr H}\otimes{\mathscr H}$ that multiplies with $\theta\mapsto S(\theta)$. Then, for $m\neq m'$,
		\begin{equation*}
			(E_m\otimes 1)\,S\,(1\otimes E_{m'})=0
			\,,\qquad
			(1\otimes E_m)\,S\,(E_{m'}\otimes 1)=0
			\,.
		\end{equation*}
		\item Gauge invariance:\hspace{.7cm} $$[S(\theta),V_1(g)\otimes V_1(g)]=0\qquad g\in G.$$
	\end{enumerate}
	The set of all such S-matrices will be denoted $\SF$.
\end{definition}

Although the above definition is manifestly basis-independent, we shall mostly work in the basis $\{e^\alpha\}_\alpha$ of $\cal K$ introduced earlier. In particular, we note that in view of $i),ii),iv),v)$, the components of $S$ in this basis satisfy
\begin{align}\label{eq:S-Properties-in-Basis}
	S^{\alpha\beta}_{\gamma\delta}(\te)
	=
	\overline{S^{\bar\beta\bar\alpha}_{\bar\delta\bar\gamma}(-\te)}	
	=
	S^{\bar\delta\bar\gamma}_{\bar\beta\bar\alpha}(\te)
	\,,\qquad
	S^{\alpha\beta}_{\gamma\delta}(i\pi-\te)
	=
	S^{\bar\gamma\alpha}_{\delta\bar\beta}(-\te)\,,
\end{align}
for any index $\alpha,\beta,\gamma,\delta$, and any rapidity $\te\in\Rl$. 

\medskip

The set $\SF$ can be completely determined in the scalar case, consisting of a single species of \textit{neutral} particles, i.e. $\K=\Cl$, $G=\{e\}$ \cite{DocL}. In this case the Yang-Baxter equation, translational invariance and gauge invariance are trivially fulfilled. 

As trivial examples of $S$ that work for any gauge group $G$ (and any $V,\K$) we have $S=\pm F$ \eqref{flip}, but depending on $G,V,\K$, several other elements of $\SF$ are known. Some of these are of particular physical interest, such as the S-matrices of the $O(N)$ $\sigma$-models, corresponding to $G=O(N)$ on $\K=\Cl^N$ in its fundamental representation \cite{AAR,Zam78}. 

We will discuss some examples in Section~\ref{section:intertwiners}. For the time being, the only properties of $S\in\SF$ that are relevant for our construction are those summarized in Def.~\ref{S-matrixDefinition}. At a later stage, we will however need to impose an additional regularity condition on $S$, which we already introduce here. 

\begin{definition}\label{regularS}
	An S-matrix $S\in\SF$ is called regular if there exists $0<\kappa<\frac{\pi}{2}$ such that $S$ has a bounded analytic continuation to the enlarged strip $\Strip(-\kappa,\pi+\kappa)\supset\Strip(0,\pi)$. The family of all regular S-matrices is denoted $\SF_0\subset\SF$, and for $S\in\SF_0$ and $\kappa$ as above, we write
	\begin{equation}\label{Skappabounded}
		\|S\|_\kappa:=\sup\{\|S(\zeta)\|:\zeta\in\overline{\Strip(-\kappa,\pi+\kappa)}\}<\infty
		\,.
	\end{equation}
\end{definition}

For a discussion of the (thermodynamical) motivation for this regularity property, we refer to \cite[p.~833]{L08}.

\medskip

We now come to the definition of a particular ``S-symmetric'' subspace of the unsymmetrized Fock space $\widehat{\mathscr H}$ over $\mathscr{H}_1$. This Hilbert space is constructed by introducing an $S$-dependent action $D_n$ of the permutation group $\mathfrak{S}_n$ of $n$ elements on $\mathscr{H}_1^{\otimes n}$ \cite{DocL}. We put
\begin{equation}\label{reprpermugroup}
\left(D_n(\tau_k)\Psi_n\right)(\theta_1,\dots,\theta_n):=S(\theta_{k+1}-\theta_k)_{n,k}\Psi_n(\theta_1,\dots,\theta_{k+1},\theta_k,\dots,\theta_n),\quad\Psi_n\in\mathscr{H}_1^{\otimes n},
\end{equation}
where $\tau_k\in\mathfrak{S}_n$, $k=1,\dots,n-1$, is the transposition that exchanges $k$ and $k+1$. Setting $D_n(\tau_{k_1}\cdots\tau_{k_l}):=D_n(\tau_{k_1})\cdots D_n(\tau_{k_l})$, we obtain for any permutation $\pi\in\mathfrak{S}_n$ a unitary tensor $S^\pi_n:\mathbb{R}^n\rightarrow\mathcal{U}(\mathcal{K}^{\otimes n})$ such that
\begin{equation}\label{tensor}
	\left(D_n(\pi)\Psi_n\right)(\boldsymbol{\theta})
	=
	S^\pi_n(\boldsymbol{\theta})\Psi_n(\theta_{\pi(1)},\dots,\theta_{\pi(n)}),\qquad\Psi_n\in\mathscr{H}_1^{\otimes n}.
\end{equation}
For transpositions $\pi=\tau_k$, we have $S^{\tau_k}_n(\boldsymbol{\theta})=S(\theta_{k+1}-\theta_k)_{n,k}$ by definition. For general $\pi$, the tensor $S^\pi_n$ is a (tensor) product of several such factors, see also Section~\ref{Section:Analyticity+Combinatorics} for a graphical notation.

As a consequence of properties 1)--3) of Def. \ref{S-matrixDefinition}, we have
\begin{lemma}{\rm\bf\cite{LM}}\label{lemma:Dn}
     $D_n$ is a unitary representation of the permutation group $\mathfrak{S}_n$ on $\mathscr{H}_1^{\otimes n}$.
\end{lemma}
The mean over $D_n$,
\begin{equation}\label{Pn}
P_n:=\frac{1}{n!}\sum_{\pi\in\mathfrak{S}_n}D_n(\pi),
\end{equation}
is the orthogonal projection onto the $D_n$-invariant subspace of $\mathscr{H}_1^{\otimes n}$ \cite{DocL}. The S-symmetrized Fock space $\mathscr{H}$ over $\mathscr{H}_1$ is then defined as
\begin{equation}\label{s}
	\mathscr{H}
	:=
	\bigoplus_{n=0}^\infty\mathscr{H}_n,\qquad\mathscr{H}_n:=P_n\mathscr{H}_1^{\otimes n},\;\,n\ge1,\qquad\Hil_0:=\Cl\,,
\end{equation}
and its Fock vacuum as $\Omega:=1\oplus0\oplus0...$ . General vectors $\Psi=(\Psi_0,\Psi_1,\Psi_2, ...\,)$, $\Psi_n\in\Hil_n$, are characterized by the symmetry property, $\bte\in\Rl^n$, $1\leq k<n$,
\begin{equation}\label{sum}
	\Psi_n(\boldsymbol{\theta})=S(\theta_{k+1}-\theta_k)_{n,k}\Psi_n(\theta_1,\dots,\theta_{k+1},\theta_k,\dots,\theta_n),
\end{equation}
and finite Fock norm $\|\Psi\|^2=\sum_{n=0}^{\infty}\int d^n\boldsymbol{\theta}\overline{\Psi_n^{\boldsymbol{\alpha}}(\boldsymbol{\theta})}\Psi_n^{\boldsymbol{\alpha}}(\boldsymbol{\theta})<\infty$. Here and in the following we used the Einstein summation convention (sum over $\alpha_1,...,\alpha_n=1,...,\dim\K$).

For later reference, we also introduce the particle number operator $N$ on $\mathscr{H}$,
\begin{equation}\label{opN}
(N\Psi)_n:=n\Psi_n,
\end{equation}
for vectors with $\sum_n n^2\|\Psi_n\|^2<\infty$, and refer to the dense subspace $\mathcal{D}\subset\mathscr{H}$, consisting of terminating sequences $(\Psi_0,\Psi_1,\dots,\Psi_n,0,\dots)$, as the subspace of finite particle number.

\medskip

Thanks to the properties of $S\in\cal S$, the representations $\widehat{U}=\bigoplus_{n=0}^\infty U_1^{\otimes n}$ and $\widehat{V}=\bigoplus_{n=0}^\infty V_1^{\otimes n}$ can be restricted to $\mathscr{H}$, and we denote these restrictions by 
\begin{equation}\label{rep}
	U:=\widehat{U}\big|_{\mathscr{H}},\qquad V:=\widehat{V}\big|_{\mathscr{H}}.
\end{equation}
Clearly, $U$ is a strongly continuous positive energy representation of $\mathcal{P}_+^\uparrow$, with up to a phase unique invariant unit vector $\Omega$, legitimizing thereby the interpretation of the latter as the physical vacuum state. The PCT operator $J$ on $\mathscr{H}$ is defined as
\begin{equation}
\left(J\Psi\right)_n^{\boldsymbol{\alpha}}(\boldsymbol{\theta}):=\overline{\Psi_n^{\overline{\alpha}_n\dots\overline{\alpha}_1}(\theta_n,\dots,\theta_1)}, \qquad \Psi\in\mathscr{H},
\end{equation}
which is well-defined because of property $v)$ in Def.~\ref{S-matrixDefinition}. This operator is an antiunitary involution which extends $U$ to a representation of the proper Poincaré group $\mathcal{P}_+$ on $\mathscr{H}$ by assigning the space-time reflection $j(x):=-x$ to $U(j):=J$ \cite{schutzenhofer2011multi}. The PCT operator $J$ commutes with the representation $V$ as shown in \cite[Lemma 2.3.]{LS}.\par
\bigskip
The S-symmetric Fock space $\mathscr H$ carries natural creation/annihilation operators, which arise as the compressions of the canonical unsymmetrized creation/annihilation operators on $\widehat{\mathscr{H}}$ to the $S$-symmetric subspace $\mathscr{H}$. Given $\varphi\in\mathscr{H}_1$, we define a creation operator as
\begin{align}
	z^\dagger(\varphi)\Psi_n:=\sqrt{n+1}\,P_{n+1}(\varphi\otimes\Psi_n)\,,\qquad \Psi_n\in\mathscr{H}_n\,,
\end{align}
and an annihilation operator as $z(\varphi):=z^\dagger(\varphi)^*$. We relate to these operators the distributions $z^{\dagger}_\alpha(\theta)$ and $z_\alpha(\theta)$ by
\begin{equation}
z^\dagger(\varphi)=\int d\theta\, z^\dagger_\alpha(\theta)\varphi^\alpha(\theta),\qquad z(\varphi)=\int d\theta \,z_\alpha(\theta)\overline{\varphi^\alpha(\theta)}.
\end{equation}

\newpage\begin{proposition}{\rm\bf\cite{LS}} \label{proposition:z-operators}
	Let $\varphi\in\mathscr{H}_1$ and $\Psi\in\mathcal{D}$.
	\begin{enumerate}
	\item The operators $z^\dagger(\varphi)$ and $z(\varphi)$ act explicitly according to
	\begin{subequations}\label{eq:zzd}
	\begin{align}
	\left(z(\varphi)\Psi\right)_n^{\boldsymbol{\alpha}}(\boldsymbol{\theta})
	&=
	\sqrt{n+1}\int d\theta' \overline{\varphi^\beta(\theta')}\Psi^{\beta\boldsymbol{\alpha}}_{n+1}(\theta',\boldsymbol{\theta})
	,\\
	\big(z^\dagger(\varphi)\Psi\big)_n(\boldsymbol{\theta})
	&=
	\frac{1}{\sqrt{n}}\sum_{k=1}^{n}S_n^{\sigma_k}(\boldsymbol{\theta})\left(\varphi(\theta_k)\otimes\Psi_{n-1}(\theta_1,\dots,\hat{\theta}_k,\dots,\theta_n)\right),\quad n\geq1
	,\\
	\big(z^\dagger(\varphi)\Psi\big)_0
	&=
	0,
	\end{align}
	\end{subequations}
	where $\sigma_k:=\tau_{k-1}\tau_{k-2}\cdots\tau_1\in\mathfrak{S}_n$ with $\sigma_1:=\id$ and $\hat{\theta}_k$ denotes omission of $\theta_k$.
	\item With respect to the particle number operator $N$ (\ref{opN}), there hold the bounds
	\begin{equation}\label{numberBounds}
		\|z(\varphi)\Psi\|\leq \|\varphi\|\,\|N^{1/2}\Psi\|,\qquad\|z^\dagger(\varphi)\Psi\|\leq \|\varphi\|\,\|(N+1)^{1/2}\Psi\|\,.
	\end{equation}
	\item The distributional kernels $z^{\dagger}_\alpha(\theta)$ and $z_\alpha(\theta)$ satisfy\footnote{The exact positions of the indices in the relations \eqref{exchange} are best memorized via a diagrammatic notation, which we introduce in Section~\ref{Section:Analyticity+Combinatorics}.}
	\begin{subequations}\label{exchange}
	\begin{eqnarray}
		z_\alpha(\theta)z_\beta(\theta')&=& S_{\delta\gamma}^{\beta\alpha}(\theta-\theta')z_\gamma(\theta')z_\delta(\theta)
		,\\
		z^\dagger_\alpha(\theta)z^\dagger_\beta(\theta')&=& S^{\gamma\delta}_{\alpha\beta}(\theta-\theta')z^\dagger_\gamma(\theta')z^\dagger_\delta(\theta),\\
		z_\alpha(\theta)z^\dagger_\beta(\theta')&=& S^{\alpha\gamma}_{\beta\delta}(\theta'-\theta)z_\gamma^\dagger(\theta')z_\delta(\theta)+\delta^{\alpha\beta}\delta(\theta-\theta')\cdot 1.
		\label{eq:zz-mixed}
	\end{eqnarray}
	\end{subequations}
	\end{enumerate}
\end{proposition}

In the special case that $S$ is the flip operator $F$ ($S^{\alpha\beta}_{\gamma\eta}(\theta)=\delta^\alpha_\eta\delta^\beta_\gamma$ in components), this construction yields the Bose Fock space with its canonical CCR operators. Similarly, $S=-F$ yields the CAR operators on the Fermi Fock space over $\Hil_1$. For generic S-matrices the operators $z^\dagger(\varphi)$, $z(\varphi)$ form a representation of the so-called Zamolodchikov-Faddeev algebra \cite{zamolodchikov1979factorized}, commonly used in the context of integrable quantum field theories, see e.g. \cite{Smir92}.

\subsection{Wedge-Local Fields and Field Algebras}\label{SecWedgeLocal}

The preparations made in the previous sections allow for the explicit construction of wedge local fields as shown in \cite{LS}. These auxiliary operators play an important role in our analysis of the existence of \textit{local} fields/observables. We shall therefore review the relevant results of \cite{LS} in this section.

To introduce these fields, we fix some arbitrary S-matrix $S\in\cal S$ and define, in analogy to a free field, a field operator on $\mathcal{D}$ as
\begin{equation}\label{field1}
	\phi(f):=z^\dagger(f^+)+z(Jf^-),\qquad f\in\mathscr{S}(\mathbb{R}^2)\otimes\mathcal{K},
\end{equation}
where $z,z^\dagger$ are the creation and annihilation operators \eqref{eq:zzd}, and
\begin{equation}
	f^{\pm,\alpha}(\theta):=\widetilde{f}^\alpha(\pm p_{m_{[\alpha]}}(\theta))=\frac{1}{2\pi}\int d^2x\,e^{\pm ip_{m_{[\alpha]}}(\theta)\cdot x}f^\alpha(x),\qquad \theta\in\mathbb{R}.
\end{equation}
Since $f^{\pm,\alpha}\in L^2(\mathbb{R},d\theta)$ for $f^\alpha\in\mathscr{S}(\mathbb{R}^2)$, the functions $f^\pm$ may be considered as vectors in $\mathscr{H}_1$. The operators (\ref{field1}) are related to the distributions
\begin{eqnarray}\label{fieldkernel}
	\phi_\alpha(x)=
	\frac{1}{2\pi}
	\int d\theta\,\left(z^\dagger_\alpha(\theta)\,e^{ip_{m_{[\alpha]}}(\theta)\cdot x}+z_{\overline{\alpha}}(\theta)\,e^{-ip_{m_{[\alpha]}}(\theta)\cdot x}\right)
\end{eqnarray}
by
\begin{equation}
	\phi(f)=\int d^2x\, \phi_\alpha(x) f^\alpha(x),\qquad f\in\mathscr{S}(\mathbb{R}^2)\otimes\mathcal{K}.
\end{equation}

In addition to $\phi$, it is useful to introduce a second auxiliary field $\phi'$, given by
\begin{equation}\label{field2}
	\phi'(f):=Jz^\dagger(Jf^+)J+Jz(f^-)J,\qquad f\in\mathscr{S}(\mathbb{R}^2)\otimes\mathcal{K}.
\end{equation}

As shown in Thm.~\ref{theo} below, the fields $\phi$ and $\phi'$ have some of the usual Wightman type properties \cite{streater2000pct}, such as covariance and cyclicity of the vacuum. However, the fields $\phi_\alpha(x)$, $\phi'_\alpha(x)$ are {\em not} localized at the space-time point $x\in\Rl^2$, but rather in unbounded wedge regions in $\Rl^2$. Recall that the {\em right wedge} is $W_R:=\{x\in\Rl^2\,:\,x_1>|x_0|\}$, and the {\em left wedge} is $W_L:=W_R'=-W_R$, the causal complement of $W_R\subset\Rl^2$. 

\begin{theorem}{\bf \cite{LS}}\label{theo}\label{theorem:fields}
Let $f\in\mathscr{S}(\mathbb{R}^2)\otimes\mathcal{K}$ and $\Psi\in\mathcal{D}$.
\begin{itemize}
\item[i)] The map $f\mapsto\phi(f)\Psi$ is linear and continuous.
\item[ii)] Define $(f^*)^\alpha(x):=\overline{f^{\overline{\alpha}}(x)}$. Then $\phi(f)^*\supset\phi(f^*)$.
\item[iii)] Each vector in $\mathcal{D}$ is entire analytic for $\phi(f)$. If $f=f^*$, then $\phi(f)$ is essentially self-adjoint on $\mathcal{D}$.
\item[iv)] $\phi(f)$ transforms covariantly under $\mathcal{P}_+^\uparrow$ and $G$, that is,
\begin{equation}
\begin{aligned}
U(a,t)\phi(f)U(a,t)^{-1}=\phi(f_{(a,t)}),&\qquad f_{(a,t)}(x):=f(\Lambda(t)^{-1}(x-a)),\,\, (a,t)\in\mathcal{P}_+^\uparrow,\\
V(g)\phi(f)V(g^{-1})=\phi(V_1(g)f),&\qquad (V_1(g)f)(x):=V_1(g)f(x),\qquad g\in G,
\end{aligned}
\end{equation}
where $\Lambda(t)$ is the Lorentz boost matrix with rapidity parameter $2\pi t$.
\item[v)] Let $f_{(j)}^\alpha(x):=\overline{f^{\overline{\alpha}}(-x)}$. Then
\begin{align}
	J\phi (f)J=\phi'(f_{(j)})
	\,,\qquad 
	J\phi' (f)J=\phi(f_{(j)})
	\,.
\end{align}
\item[vi)] For any open set $\mathcal{O}\subset\mathbb{R}^2$, the subspace
\begin{equation}
\mathcal{D}_{\mathcal{O}}:=\spann\,\{\phi(f_1)\cdots\phi(f_n)\Omega:f_1,\dots,f_n\in\mathscr{S}(\mathcal{O})\otimes\mathcal{K},n\in\mathbb{N}_0\}
\end{equation}
is dense in $\mathscr{H}$. That is, $\Omega$ is cyclic for the field $\phi$.
\item[vii)] The field $\phi$ is local if and only if $S=F$ (the flip).
\end{itemize}
The statements $i)$--$vii)$ also hold when $\phi$ and $\phi'$ are exchanged. 
\begin{itemize}
\item[viii)] Let $f\in\mathscr{S}(W_R+a)\otimes\mathcal{K}$ and $g\in\mathscr{S}(W_L+a)\otimes\mathcal{K}$ for some $a\in\mathbb{R}^2$. Then
\begin{equation}\label{commutatorwedge}
[\phi'(f),\phi(g)]\Psi=0,\qquad\Psi\in\mathcal{D}\,,
\end{equation}
that is, the fields $\phi$ and $\phi'$ are relatively wedge-local.
\end{itemize}
\end{theorem}

These results can be interpreted as follows: The fields $\phi,\phi'$ are not local in the usual sense,
\begin{wrapfigure}{r}{0.31\textwidth}
  \begin{center}
    \includegraphics[width=0.28\textwidth]{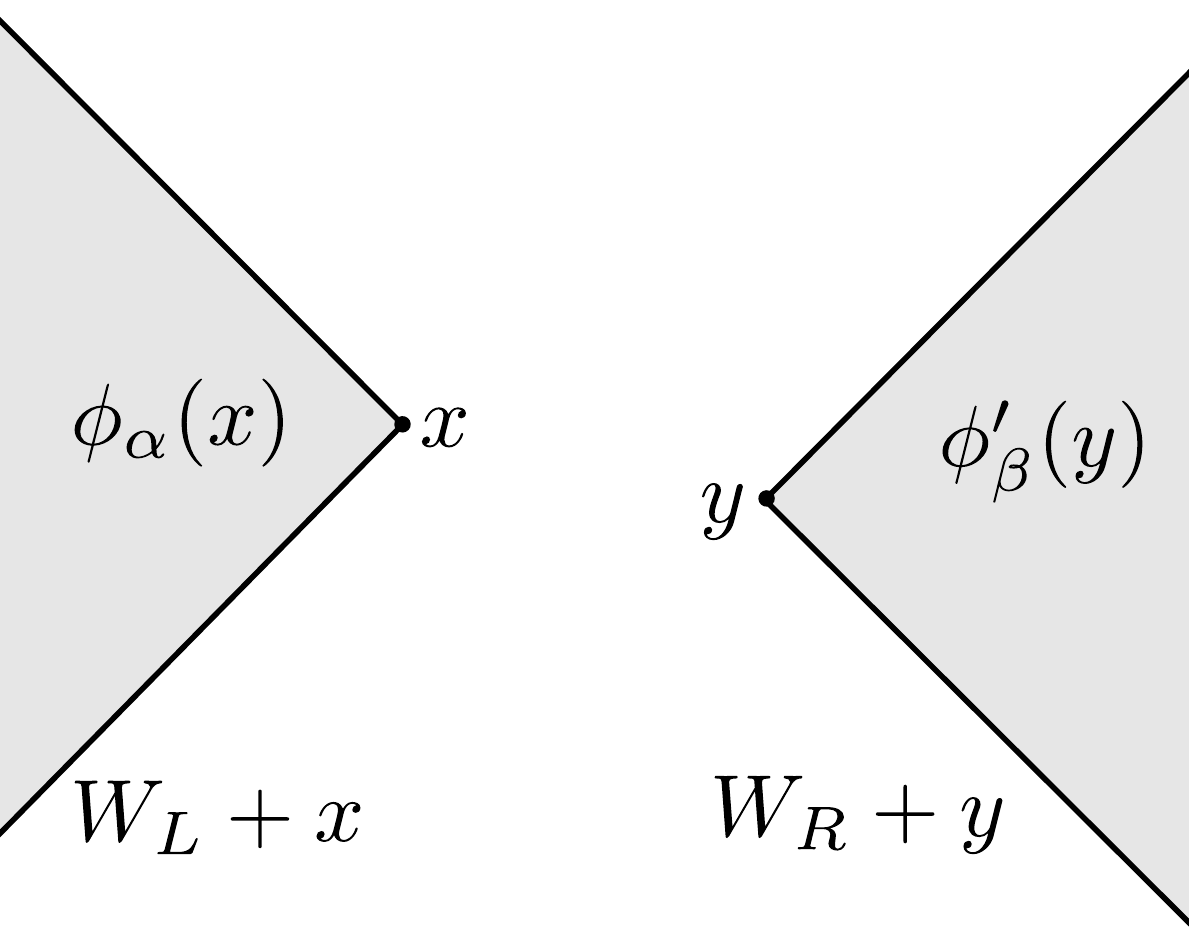}
 \end{center}
\end{wrapfigure}
but rather localized in wedges: $\phi_\alpha(x)$ is localized in the (shifted) left wedge $W_L+x$, and $\phi_\beta'(y)$ is localized in the (shifted) right wedge $W_R+y$. This assignment of localization regions to field operators is consistent with causal commutation relations and covariance, including the PCT symmetry. 

As the fields are localized in infinitely extended regions, they are however not the fields of direct physical interest, but rather auxiliary objects (``polarization-free generators'' \cite{BBS01}). The physical observables localized in a bounded region $\O$ can be characterized by demanding that they ought to commute with all $\phi_\alpha(x)$, $\phi_\beta'(y)$ such that $W_L+x$ and $W_R+y$ are spacelike to $\O$.

\section{Local Observables and the Modular Nuclearity Condition}\label{ChapterModNuc}

\subsection{Operator-algebraic formulation of the models}\label{Section:Operator-algebraic-Formulation}

The commutation constraints on local observables are best expressed in an operator-algebraic formulation, trading the fields $\phi,\phi'$ for the von Neumann algebras they generate. Due to the essential self-adjointness of the field operators, this can be easily done as follows: We define, $x\in\mathbb{R}^2$,
\begin{subequations}\label{algebras}
\begin{eqnarray}
	\mathcal{F}(W_L+x)&:=&\{e^{i\overline{\phi (f)}}:f=f^*\in\mathscr{S}(W_L+x)\otimes \mathcal{K}\}'',\\
	\mathcal{F}(W_R+x)&:=&\{e^{i\overline{\phi'(f)}}:f=f^*\in\mathscr{S}(W_R+x)\otimes \mathcal{K}\}''.
\end{eqnarray}
\end{subequations}
Here the overline indicates the selfadjoint closure, and the exponentials are defined by the functional calculus.

To translate the properties listed in Thm.~\ref{theo} into von Neumann algebraic terminology, we define the {\em family of all wedges} $\cal W$ as the collection of all $W_{L/R}+x$, $x\in\Rl^2$. Then \eqref{algebras} defines a map from $\cal W$ to von Neumann algebras in $\B(\Hil)$ with the following properties.

\begin{proposition}{\rm\bf\cite{LS}}\label{PropWedgeAlgebra}
	Let $S\in\mathcal{S}$ and $\mathcal{F}(W)$, \mbox{$W\in\mathcal{W}$}, be defined as in (\ref{algebras}). Then, for any $W,W_1,W_2\in\W$, the following holds.
	\begin{enumerate}
		\item Isotony: $\qquad\qquad\quad\quad\,\,\quad\,\,\mathcal{F}(W_1)\subset\mathcal{F}(W_2)\quad$ if $\quad W_1\subset W_2$,
		\item Covariance: $\qquad\quad\qquad\quad U(\lambda)\mathcal{F}(W)U(\lambda)^{-1}=\mathcal{F}(\lambda\,W),\quad \lambda\in\mathcal{P}_+$,
		\item Inner Symmetry: $\,\,\,\,\qquad\quad V(g)\mathcal{F}(W)V(g)^{-1}=\mathcal{F}(W),\quad g\in G,$
		\item Locality: $\qquad\quad\,\,\quad\qquad\quad\mathcal{F}(W_1)\subset\mathcal{F}(W_2)'\quad$ if $\quad W_1\subset W_2'$, 
		\item Cyclicity of the vacuum: $\,\,\,\,\mathcal{F}(W)\Omega\,\,\,$ is dense in $\,\,\,\mathscr{H}$.
	\end{enumerate}
	Here $\mathcal{F}(W)'$ denotes the commutant of $\mathcal{F}(W)$ in ${\cal B}({\mathscr H})$, and $W'\in\W$ is the spacelike (causal) complement of $W\in\W$.
\end{proposition}

We now turn to the question of extracting operators localized in {\em bounded} regions in Minkowski space from the wedge algebras $\F(W)$, $W\in\W$. 

\begin{wrapfigure}{r}{0.31\textwidth}
  \begin{center}
    \includegraphics[width=0.28\textwidth]{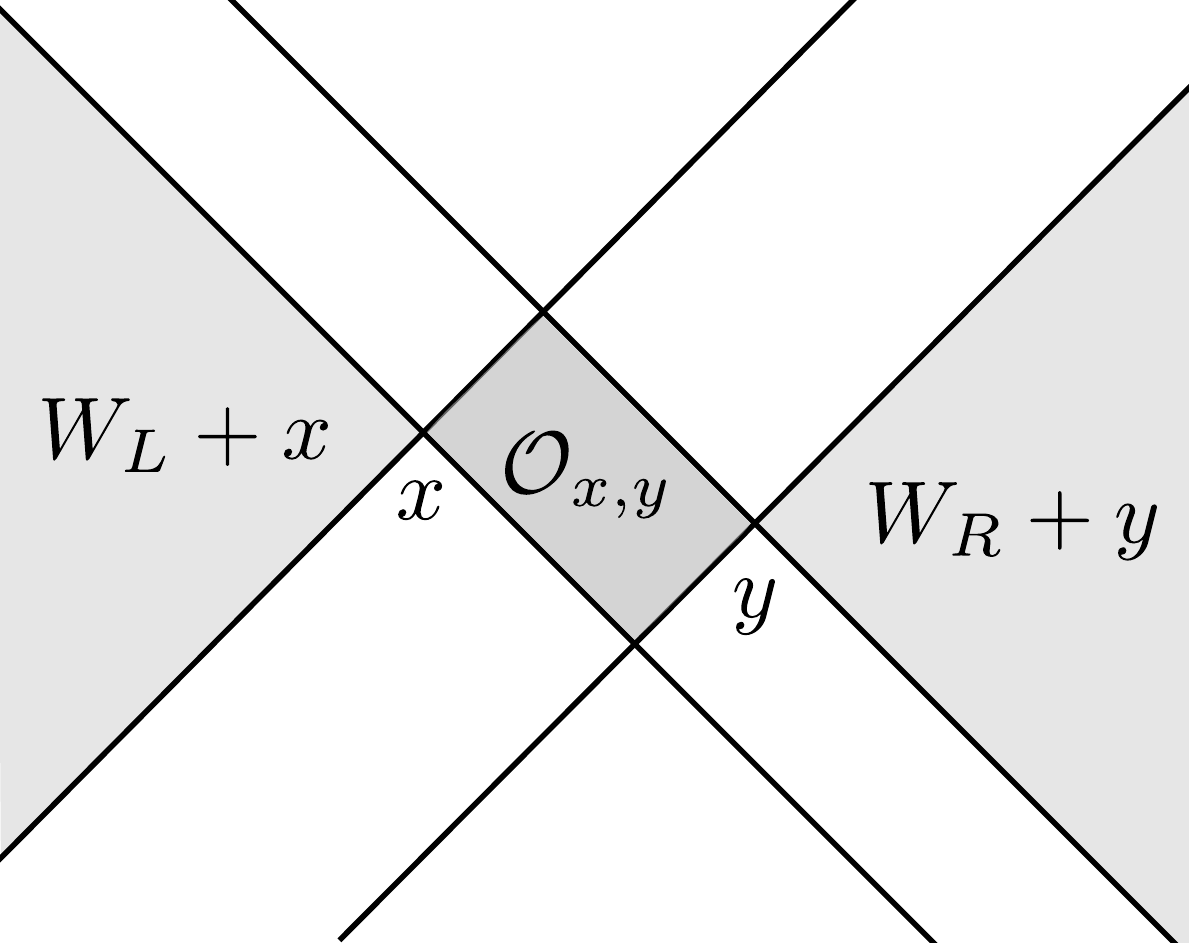}
  \end{center}
\end{wrapfigure}

The prototype of such a bounded localization region is a {\em double cone},
\begin{equation}
	\mathcal{O}_{x,y}
	:=
	(W_R+x)\cap(W_L+y),\qquad y-x\in W_R\,.
\end{equation}
The causal complement $\O_{x,y}'$ of $\O_{x,y}$ consists of two disjoint wedges, $\O_{x,y}'=(W_L+x)\cup(W_R+y)$. Thus any operator $A\in\mathcal{B}(\mathscr{H})$ which models an observable localized in $\mathcal{O}_{x,y}$ must commute with all observables localized in $W_L+x$ and $W_R+y$ by causality. In other words, the maximal von Neumann algebra that can be associated with $\O_{x,y}$ is
\begin{equation}\label{vNDouble}
	\mathcal{F}(\mathcal{O}_{x,y}):=\left(\mathcal{F}(W_L+x)\vee\mathcal{F}(W_R+y)\right)'=\mathcal{F}(W_L+x)'\cap\mathcal{F}(W_R+y)'\,.
\end{equation}
This yields our definition of von Neumann algebras associated with double cones. By additivity, we can extend this definition to arbitrary regions: If $\mathfrak{R}\subset\mathbb{R}^2$ is bounded and open, we define $\mathcal{F}(\mathfrak{R})\subset\B(\Hil)$ as the smallest von Neumann algebra containing all $\F(\O_{x,y})$, $\O_{x,y}\subset\mathfrak{R}$. This construction indeed yields a local net $\{\mathcal{F}(\mathcal{O})\}_{\mathcal{O}\subset\mathbb{R}^2}$ which inherits its basic properties from those of the wedge algebras (\ref{algebras}).

\begin{proposition}\label{PropDoubleCones}
	The von Neumann algebras $\mathcal{F}(\mathcal{O})$ \eqref{vNDouble} satisfy the properties $i)$--$iv)$ of Prop.~\ref{PropWedgeAlgebra} with double cones instead of wedges.
\end{proposition}

The simple proof of this proposition can be found in several places \cite{Bor1,BaumgrtelWollenberg:1992, DocL}. It has to be noted, however, that Prop.~\ref{PropDoubleCones} does not include the cyclicity result $v)$ of Prop.~\ref{PropWedgeAlgebra}. In fact, it seems possible that the double cone algebras \eqref{vNDouble} are trivial in the sense that $\F(\O)=\Cl\cdot 1$, which would mean that the model under consideration does not contain any (non-trivial) local observables. To rule out this pathological situation, and to identify those S-matrices for which local observables exist, the algebras \eqref{vNDouble} have to be analyzed in more detail. 

\medskip

A direct construction of elements of the double cone algebras $\F(\O)$ is difficult, as can be expected from the complicated form that interacting local quantum fields always have -- see, however,  \cite{bostelmann2013operator,bostelmann2014characterization} for results in this direction (for scalar $S$). We will focus here rather on an abstract analysis of the double cone algebras \eqref{vNDouble}, which circumvents the explicit construction of local fields and relies on operator-algebraic methods. In particular, we will make use of the modular data $\tilde J,\Delta$ \cite{BratteliRobinson:1987} of the pair $\F(W_R),\Omega$ (which exist because $\Omega$ is cyclic and separating for $\F(W_R)$, see Prop.~\ref{PropWedgeAlgebra}~$v),iv)$).

As we will show in Prop.~\ref{BisognWich} below, the operators $\tilde J$ and $\Delta$ act ``geometrically correct'', i.e. as expected from the Bisognano-Wichmann Theorem \cite{BisWich, bisognano1976duality} and Borchers' Theorem \cite{Bor1}. That is, the modular unitaries $\Delta^{it}$ and modular conjugation $\tilde J$ coincide with the Lorentz boosts and PCT operator $J$ from the representation $U$, respectively.

To motivate these results, we recall \cite{BisWich} that if $f\in\mathscr{S}(\mathbb{R}^2)\otimes\mathcal{K}$ has compact support in $W_R$, then $f^+\in\mathscr{H}_1$ lies in the domain of the positive self-adjoint operator $U(0,\tfrac{i}{2})$ and
\begin{equation}
(U(0,\tfrac{i}{2})f^+)(\theta)=f^+(\theta-i\pi)=f^-(\theta).
\end{equation}
Moreover, since $(f^{*})^\alpha(x)=\overline{f^{\overline{\alpha}}(x)}$, we have
\begin{equation}
(Jf^{\pm})^\alpha(\theta)=(f^{*})^{\mp,\alpha}(\theta).
\end{equation}
In view of the definitions \eqref{field2} of $\phi'$ and \eqref{eq:zzd} of the creation/annihilation operators, this yields, for compactly supported $f\in\mathscr{S}(W_R)\otimes\mathcal{K}$,
\begin{equation}\label{oneparticletomita}
JU(0,\tfrac{i}{2})\phi'(f)\Omega=JU(0,\tfrac{i}{2})f^+=(f^*)^+=\phi'(f^*)\Omega=\phi'(f)^*\Omega=\tilde J \Delta^{1/2}\phi'(f)\Omega\,.
\end{equation}
In the last step, we have here used the defining property of the Tomita operator $S=\tilde{J}\Delta^{1/2}$ (not to be confused with the S-matrix $S$). This calculation already suggests a one particle version of the Bisognano-Wichmann property. The full Bisognano-Wichmann property is the statement of the next proposition.

\begin{proposition}\label{BisognWich}
	Let $S\in\SF$. 
	\begin{enumerate}
		\item The modular operator $\Delta$ and conjugation $\tilde{J}$ of $(\mathcal{F}(W_R),\Omega)$ are
		\begin{eqnarray}\label{eq:CGMA}
			\Delta^{it}=U(0,-t),\quad t\in\mathbb{R},\qquad 		\tilde{J}=J\,.
		\end{eqnarray}
		\item Haag duality for wedges holds:
		\begin{equation}
		\mathcal{F}(W)'=\mathcal{F}(W'),\qquad W\in\mathcal{W}\,.
		\end{equation}
		\end{enumerate}
\end{proposition}
\begin{proof}
	The proof consists of several steps, see also \cite{Alazzawi:2014} for further details. To begin with, we note that the selfadjoint operators $A:=\overline{\phi'(f)}$, with $f\in\mathscr{S}(W_R)\otimes \mathcal{K}$ of compact support, are affiliated with $\mathcal{F}(W_R)$. This follows by a calculation on analytic vectors, see for example \cite{GL}. It then follows by standard arguments that $A\Omega$ lies in the dom$S={\rm dom}\Delta^{1/2}$ and $SA\Omega=A^*\Omega$. By the same arguments, also $AF\Omega$, where $F\in\mathcal{F}(W_R)$ is arbitrary, lies in the domain of $S$.
	
	By modular theory, the Tomita operator of the pair $(\mathcal{F}(W_R)',\Omega)$ is $S^*$. We can then repeat the same arguments as above to show that the selfadjoint closure of the other field operator,  $\overline{\phi(f')}$, with $f'\in\mathscr{S}(W_R')\otimes\mathcal{K}$ of compact support, is affiliated with $\mathcal{F}(W_R)'$, and satisfies $S^*\phi(f')\Omega=\phi(f')^*\Omega$.
	
	Next, we want to prove 
	\begin{equation}\label{inclusion}
	S^1_{\text{geo}}\subset SE^{(1)},
	\end{equation}
	where $S^1_{\text{geo}}:=JU(0,\tfrac{i}{2})E^{(1)}$ and $E^{(1)}$ is the projection onto $\mathscr{H}_1$. It follows from (\ref{oneparticletomita}) that \eqref{inclusion} holds on the space  $\mathcal{D}_0:=\left\lbrace \Psi\in\mathscr{H}_1: \Psi=\phi'(f)\Omega,\, \text{supp}\,f\in W_R\,\,\text{compact} \right\rbrace$, which is dense in $\mathscr{H}_1$. To show that this extends to dom$S_{\rm geo}^1$, recall that $S_{\rm geo}^1$ is a closed antilinear involution \cite{BGL}. Its domain is therefore of the form dom$\,S^1_{\text{geo}}=\text{dom}\,U(0,\tfrac{i}{2})=K+iK$, where $K$ is a real subspace of $\mathscr{H}_1$. An arbitrary vector $h+ik\in{\rm dom}S_{\rm geo}^1$, $h,k\in K$, can be approximated by sequences $h_n,k_n\in{\cal D}_0\cap K$, corresponding to real functions $f=f^*$, such that $h_n+ik_n\overset{n}{\longrightarrow} h+ik\in\text{dom}\,S^1_{\text{geo}}$. In particular, $S^1_{\text{geo}}(h_n+ik_n)=h_n-ik_n\overset{n}{\longrightarrow}h-ik$, which by the closedness of $S^1_{\text{geo}}$ gives $S^1_{
\text{geo}}(h+ik)=h-ik$. But $Sf^+=f^+$ \eqref{oneparticletomita}, and we conclude $S(h+ik)=h-ik$, and thus \eqref{inclusion}.
	
	To show the opposite inclusion, namely $SE^{(1)}\subset S^1_{\text{geo}}$, note that $E^{(1)}=\sum_{q\in\mathcal{Q}}E_{m_q}$, where the $E_{m_q}$ are the spectral projections of the mass operator $\sqrt{P^2}$. By a theorem of Borchers \cite{Bor1}, the Tomita operator $S$ commutes with $\sqrt{P^2}$ and consequently with $E^{(1)}$ \cite{Mund}. We proceed by defining $S(W_R'):=JSJ$ and $S^1_{\text{geo}}(W_R'):=JS^1_{\text{geo}}J$, corresponding to the opposite wedge. It follows from locality and modular theory that $S\subset S(W_R')^*$. Since $S(W_R')$ commutes with $E^{(1)}$, we have together with (\ref{inclusion})
	\begin{equation}
	SE^{(1)}\subset (S(W_R')E^{(1)})^*\subset S^1_{\text{geo}}(W_R')^*=S^1_{\text{geo}},
	\end{equation}
	where the last equality follows from $JU(0,\tfrac{i}{2})=U(0,-\tfrac{i}{2})J$ by the anti-unitarity of $J$. Thus, we have shown
	\begin{eqnarray}\label{polar}
	SE^{(1)}=JU(0,\tfrac{i}{2})E^{(1)}\,.
	\end{eqnarray}
	As the polar decomposition of a closed operator is unique and $S$ is closed, we obtain from (\ref{polar}) a one particle version of the Bisognano-Wichmann theorem, namely $\Delta^{1/2}E^{(1)}=U(0,\tfrac{i}{2})E^{(1)}$ and $\tilde{J}E^{(1)}=JE^{(1)}$. In particular, we have
	\begin{equation}\label{oneparticleBisoWich}
	\Delta^{it}E^{(1)}=U(0,-t)E^{(1)},\qquad\tilde{J}E^{(1)}=JE^{(1)}.
	\end{equation}
	This result can now be used to prove the equality of the modular operator $\Delta^{it}$ with $U(0,-t)$ along the same lines as in \cite{BL4}. To this end, define $L(t):=U(0,-t)\Delta^{-it}$, $t\in\mathbb{R}$, and $\overline{\phi'_t(f)}:=L(t)\overline{\phi'(f)}L(t)^{-1}$ with $f\in\mathscr{S}(W_R)\otimes\mathcal{K}$ of compact support. Since $L(t)\mathcal{F}(W_R)L(t)^{-1}\subset\mathcal{F}(W_R)$, $\overline{\phi'_t(f)}$ is affiliated with $\mathcal{F}(W_R)$. Making use of this, (\ref{oneparticleBisoWich}), and $L(t)^{-1}\Omega=\Omega$, we have with $A'\in\mathcal{F}(W_R)'$
	$$\phi'_t(f)A'\Omega=A'\phi'_t(f)\Omega=A'L(t)\phi'(f)\Omega=A'L(t)E^{(1)}\phi'(f)\Omega=A'\phi'(f)\Omega=\phi'(f)A'\Omega\,.$$
	That is,
	\begin{equation}
	(\phi'_t(f)-\phi'(f))A'\Omega=0, \qquad A'\in\mathcal{F}(W_R)',
	\end{equation}
	for any $f$ as above. But $\mathcal{F}(W_R)'\Omega$ is a core for $\phi'(f)$ and $\phi'_{t}(f)$ (this can be shown as in \cite{BL4}), which implies $\phi'(f)=\phi'_t(f)$. Hence $U(0,-t)\Delta^{-it}$ acts trivially on $\mathcal{F}(W_R)$, and since $\Omega$ is cyclic for $\mathcal{F}(W_R)$, it follows that $U(0,-t)\Delta^{-it}=1$ as claimed.
	
	Similarly, one proves the equality of the modular conjugation $\tilde{J}$ with the PCT operator $J$. Firstly, we have $\tilde{J}\mathcal{F}(W_R)\tilde{J}=\mathcal{F}(W_R)'$ by modular theory and secondly $J\mathcal{F}(W_R)J=\mathcal{F}(W_L)$ by definition. Moreover, for an arbitrary wedge $W$ we have $\mathcal{F}(W')\subset\mathcal{F}(W)'$ (Prop~\ref{PropWedgeAlgebra}). Therefore, it follows that $\tilde{J}J\mathcal{F}(W_R)(\tilde{J}J)^{-1}\subset\mathcal{F}(W_R)$. So with $I:= \tilde{J}J$, the operator $\overline{\phi'_I(f)}:=I\overline{\phi'(f)}I^{-1}$, $f\in\mathscr{S}(W_R)\otimes\mathcal{K}$, is affiliated with $\mathcal{F}(W_R)$, and by the same arguments as above we find $(\phi'_I(f)-\phi'(f))A'\Omega=0$, for any $A'\in\mathcal{F}(W_R)'$. Since $\mathcal{F}(W_R)'\Omega$ is also a core for $\phi'_I(f)$, we find $\phi'_I(f)=\phi'(f)$ and consequently $I=1$.
	
	Haag-duality, statement $ii)$, then follows easily from $\mathcal{F}(W_R)'=\tilde{J}\mathcal{F}(W_R)\tilde{J}=J\mathcal{F}(W_R)J=\mathcal{F}(W_L)$ and covariance.
\end{proof}

Our subsequent analysis of the double cone algebras will heavily rely on this explicit form of the modular data of $(\F(W_R),\Omega)$. The reason for this is that we will use the so-called modular nuclearity condition of Buchholz, D'Antoni, and Longo \cite{BDL90}, formulated in terms of $\Delta$.

\bigskip

\begin{wrapfigure}{r}{0.32\textwidth}
 \includegraphics[width=0.3\textwidth]{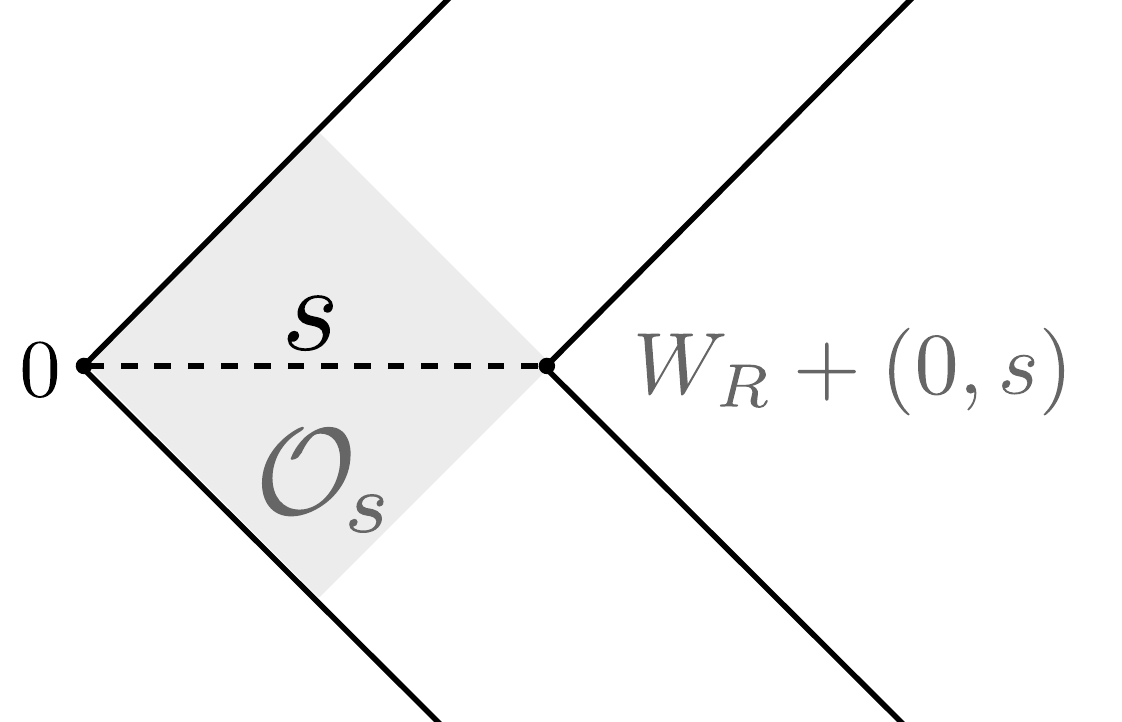}
\end{wrapfigure}
Fixing some ``splitting distance'' $s>0$, we consider the double cone $\O_s:=\O_{0,(0,s)}\subset W_R$, and the corresponding von Neumann algebra $$\F(\O_s)=\F(W_L)'\cap\F(W_R+(0,s))'.$$ In view of Prop.~\ref{BisognWich}~$ii)$, this algebra coincides with the relative commutant of the inclusion $\F(W_R+(0,s))\subset\F(W_R)$. 

Corresponding to this inclusion, we consider the linear maps
\begin{equation}\label{modNucCon}
	\Xi(s):\mathcal{F}(W_R)\rightarrow\mathscr{H},\qquad \Xi(s)A:=\Delta^{1/4}U_sA\Omega,\qquad s>0\,,
\end{equation}
where $U_s:=U(0,(0,s))$ is a shorthand for the purely spatial translation by $(0,s)\in\Rl^2$. 

The maps $\Xi(s)$ are bounded with norm at most one by modular theory. We will refer to the (much stronger) condition that $\Xi(s)$ is {\em nuclear}\footnote{Recall that a map between Banach spaces is called nuclear if it can be decomposed into a series of rank one maps with summable norms, see Def.~\ref{DefNuclear} below.} as the {\em modular nuclearity condition} \cite{BDL90} (for the inclusion $\F(W_R+(0,s))\subset\F(W_R)$).  

\begin{theorem}{\rm\bf\cite{BL4,L08}}\label{NuclearityCondition}
	Let the modular nuclearity condition be satisfied for some $s>0$, and let $\O_{a,b}$ be a double cone with $\sqrt{-(a-b)^2}>s$. Then
	\begin{enumerate}
		\item $\mathcal{F}(\mathcal{O}_{a,b})$ is isomorphic to the hyperfinite type III$_1$ factor.
		\item  $\mathcal{F}(\mathcal{O}_{a,b})$ has the vacuum $\Omega$ as a cyclic vector.
	\end{enumerate}
\end{theorem}

The significance of the inequality $\sqrt{-(a-b)^2}>s$ is to guarantee that the double cone under consideration has ``relativistic diameter'' larger than the double cone $\O_{0,(0,s)}$, for which modular nuclearity was assumed.

Nuclearity of the maps $\Xi(s)$ has several further consequences, in particular, the split property of $\F(W_R+(0,s))\subset\F(W_R)$. We refer to \cite{BDL90,buchholz1990nuclear,DL84,L08} for a discussion of these properties, and to \cite{LechnerSanders:2016} for a recent strengthening of the concept of modular nuclearity. In the present article, we will treat Thm.~\ref{NuclearityCondition} as a condition which implies the existence of infinite-dimensional algebras of observables localized in double cones.\\

Once the existence of non-trivial local observables is settled, one can use the results of \cite[Sect.~5]{LS} to conclude that the model described by $S\in\SF$ solves the inverse scattering problem: Based on the assumption that $\Omega$ is cyclic for $\F(\O)$ for some double cone $\O$, one can compute all scattering states and the S-matrix, which turns out to be the factorizing S-matrix with $2\to 2$ body operator $S$. Furthermore, the model is even asymptotically complete.\par
In light of these results, the main question to be addressed is the existence of local observables. We will investigate this question via the modular nuclearity condition, which is linked to Hardy space properties in the next section.

\subsection{Modular nuclearity and Hardy spaces}

Having identified the modular nuclearity condition as a sufficient condition for the existence of local observables, we now turn to the question of checking it in the models described in Section~\ref{SecWedgeLocal}. We will follow the same basic strategy that was applied in the case of a particle spectrum consisting of only a single species of neutral massive particles \cite{DocL, L08}, appropriately extended and generalized to the present setting of a richer particle spectrum. 

\medskip

Let us first recall a few facts about nuclear maps.
\begin{definition}\label{DefNuclear}
	Let $\X$ and $\Y$ be Banach spaces. A mapping $T\in\B(\X,\Y)$ is called nuclear if there exists a sequence of continuous linear functionals $\{\varphi_n\}_{n\in\mathbb{N}}\subset\X^*$ and a sequence of vectors \mbox{$\{y_n\}_{n\in\mathbb{N}}\subset\Y$} such that 
	\begin{equation}
		\sum_{n=1}^{\infty}\|\varphi_n\|_{\X^*}\,\|y_n\|_\Y<\infty,\qquad
		\label{nucleareAbb}
		T(x)=\sum_{n=1}^{\infty}\varphi_n(x)\,y_n,\qquad x\in\X.
	\end{equation}
	The nuclear norm of such a linear map is defined by
	\begin{equation}\label{normnuc}
		\|T\|_1:=\inf\sum_{n=1}^{\infty}\|\varphi_n\|_{\X^*}\,\|y_n\|_\Y,
	\end{equation}
	where the infimum is taken over all possible representations (\ref{nucleareAbb}) of $T$.
\end{definition}

The sets of nuclear maps between two Banach spaces $\X$ and $\Y$ are denoted by $\mathcal{N}(\X,\Y)$. We will rely on the following facts \cite{jarchow,pietsch}.

\begin{lemma}\label{propertiesNuclearMaps}
	Let $\X,\Y,{\cal V},{\cal Z}$ be Banach spaces. Then, we have
	\begin{enumerate}
		\item $\left(\mathcal{N}(\X,\Y),\|\cdot\|_1\right)$ is a Banach space.
		\item $\|T\|\leq\|T\|_1$ for any $T\in\mathcal{N}(\X,\Y)$.
		\item Let $T\in\mathcal{N}(\X,\Y)$, $B_1\in\mathcal{B}(\Y,{\cal Z})$ and $B_2\in\mathcal{B}({\cal V},\X)$. Then $B_1TB_2\in\mathcal{N}({\cal V},{\cal Z})$, and
		\begin{equation}
		\|B_1TB_2\|_1\leq\|B_1\|\,\|T\|_1\,\|B_2\|.
		\end{equation}
		\item Let $\Hil$ be a separable Hilbert space. Then, $\mathcal{N}(\mathscr{H},\mathscr{H})$ coincides with the set of trace class operators on $\mathscr{H}$, and
		\begin{equation*}
		\|T\|_1={\rm Tr}\,|T|,\qquad T\in\mathcal{N}(\mathscr{H},\mathscr{H}).
		\end{equation*}
	\end{enumerate}
\end{lemma}

Our aim is to show that for $s>0$ (at least for sufficiently large~$s$), the maps $\Xi(s)$ are nuclear maps, from the Banach space $(\F(W_R),\|\cdot\|_{\B(\Hil)})$ to the Hilbert space $\Hil$. In a first step towards this aim, we introduce their $n$-particle contributions, $n\in\Nl_0$, $s>0$, 
\begin{equation}
	\Xi_n(s):\mathcal{F}(W_R)\rightarrow\mathscr{H}_n,\qquad \Xi_n(s)A:=P_n\Xi(s) A=\Delta^{1/4}U_s(A\Omega)_n\,.
\end{equation}
In view of the second quantized nature of $U$ (and thus of $\Delta$), these operators sum to $\Xi(s)$, i.e. 
\begin{equation}\label{xi}
	\Xi(s)=\sum_{n=0}^{\infty}\Xi_n(s)\,.
\end{equation}
To show that $\Xi(s)$ is nuclear, we have to prove that all $\Xi_n(s)$ are nuclear, and that the series \eqref{xi} converges in the nuclear norm $\|\cdot\|_1$ (cf. Lemma~\ref{propertiesNuclearMaps}~$i)$). 

\medskip

These questions will be addressed with tools from complex analysis. To see how this connection comes about, we use our explicit knowledge of the translation unitaries \eqref{actionPoincare} and modular operator \eqref{eq:CGMA} to write down $\Xi_n(s)A$ explicitly. More precisely, the operator $\Delta^{1/4}$ in the definition of $\Xi$, (\ref{modNucCon}), coincides with a boost of imaginary parameter $\tfrac{i\pi}{2}$, cf. (\ref{eq:CGMA}). We, therefore, find with Formula \eqref{actionPoincare}
\begin{equation}\label{concrete}
	(\Xi_n(s)A)^{\boldsymbol{\alpha}}(\boldsymbol{\theta})
	=
	\prod_{k=1}^{n}e^{-m_{[\alpha_k]}s\cosh\theta_k}
	\cdot
	(A\Omega)_n^{\boldsymbol{\alpha}}(\theta_1-\tfrac{i\pi}{2},\dots,\theta_n-\tfrac{i\pi}{2})\,,
\end{equation}
to be understood in terms of analytic continuation. 

\bigskip

Our strategy for establishing nuclearity properties of $\Xi_n(s)$ relies on Hardy space properties of the functions $(A\Omega)_n$. As the derivation of these properties requires an intricate procedure of iterated analytic continuations, we first explain the general strategy, and begin by introducing some terminology. 

For an open convex domain $\mathcal{C}_n\subset\mathbb{R}^n$, we will consider the tube $\mathcal{T}_{\mathcal{C}_n}:=\mathbb{R}^n+i\,\mathcal{C}_n\subset\mathbb{C}^n$ based on $\C_n$. The (vector-valued) Hardy space $H^2(\mathcal{T}_{\mathcal{C}_n},\K^{\otimes n}):=H^2(\mathcal{T}_{\mathcal{C}_n})\otimes\K^{\otimes n}$ is defined as the space of all analytic functions $h:\Tu_{\C_n}\to\K^{\otimes n}$ such that for any $\bla\in\C_n$, the function $h_{\boldsymbol{\lambda}}:\boldsymbol{\theta}\mapsto h(\boldsymbol{\theta}+i\boldsymbol{\lambda})$ is an element of $L^2(\mathbb{R}^n,\K^{\otimes n}):=L^2(\mathbb{R}^n)\otimes\mathcal{K}^{\otimes n}$, with $L^2(\mathbb{R}^n,\mathcal{K}^{\otimes n})$-norms $\|h_\bla\|_2$ uniformly bounded in $\bla$. The Hardy space is a Banach space w.r.t. the norm \cite{stein1971introduction}
\begin{equation}\label{eq:HardyNorm}
	\triplenorm h\triplenorm
	:=
	\sup_{\boldsymbol{\lambda}\in\mathcal{C}_n}
	\|h_{\boldsymbol{\lambda}}\|_2
	=
	\sup_{\boldsymbol{\lambda}\in\mathcal{C}_n}
	\left(
	       \sum_{\boldsymbol{\alpha}}\int\limits_{\mathbb{R}^n}d^n\boldsymbol{\theta}\left|h^{\boldsymbol{\alpha}}(\boldsymbol{\theta}+i\boldsymbol{\lambda})\right|^2
	\right)^{1/2}
	<\infty.
\end{equation}
Let us recall the following two facts about Hardy spaces on tubes  \cite{stein1971introduction}: 
\begin{itemize}
     \item For $h\in H^2(\mathcal{T}_{\mathcal{C}_n},\mathcal{K}^{\otimes n})$, $K\subset \mathcal{C}_n$ compact and $k=1,\dots,n$, there hold the uniform limits
     \begin{equation}
	  \lim\limits_{|\theta_k|\rightarrow\infty}\sup\limits_{\boldsymbol{\lambda}\in K}|h^{\boldsymbol{\alpha}}(\boldsymbol{\theta}+i\boldsymbol{\lambda})|=0,
     \end{equation}
     with $\theta_1,\dots,\theta_{k-1},\theta_{k+1},\dots,\theta_n\in\mathbb{R}$.
     \item If $\C_n$ is an open polyhedron, (the interior of the convex hull of a finite subset of $\mathbb{R}^n$), any $h\in H^2(\mathcal{T}_{\C_n},\mathcal{K}^{\otimes n})$ has $L^2$-boundary values, i.e. can be extended to $\mathcal{T}_{\overline{\C_n}}$ such that the mapping $\overline{\C_n}\ni \boldsymbol{\lambda}\mapsto h_{\boldsymbol{\lambda}}\in L^2(\mathbb{R}^n,\mathcal{K}^{\otimes n})$ is continuous.
\end{itemize}

The desired connection between Hardy spaces and the quantum field theory models build from an S-matrix is expressed in the following definition. The main idea is to study analytic continuations of the functions $(A\Omega)_n\in L^2(\Rl^n,\K\tp{n})$ to certain tube domains. 

\begin{definition}\label{definition:property-h}
     An S-matrix $S\in\SF$ is said to have property (H) (for ``Hardy'') if for any $n\in\Nl$, there exists an open polyhedron $\C_n\subset\Rl^n$ such that
     \begin{enumerate}
		\item $\bla_{\pi/2}:=-(\frac{\pi}{2},...,\frac{\pi}{2})\in\C_n$, i.e.,
		\begin{align}\label{eq:c_n}
			c_n
			:=
			\|\bla_{\pi/2}-\partial\C_n\|_\infty
			>0
			\,,
		\end{align}
		and $\C_n\subset(-\pi,0)^{\times n}$.
		\item For any $A\in\F(W_R)$, the function $(A\Omega)_n$ can be analytically continued to $\Tu_n:=\Rl^n+i\,\C_n$.
		\item For any $A\in\F(W_R)$ and any $s>0$, the analytic continuation of $(U_sA\Omega)_n$ to $\Tu_n$ lies in the Hardy space $H^2(\Tu_n,\K^{\otimes n})$, and there exists a constant $\upsilon(2s,n)>0$ such that.
		\begin{align}\label{eq:general-hardy-bound}
			\triplenorm (U_sA\Omega)_n\triplenorm_{H^2(\Tu_n,\K^{\otimes n})}
			\leq
			\upsilon(2s,n)\cdot
			\|A\|
			\,,\qquad
			A\in\F(W_R)\,.
		\end{align}
     \end{enumerate}
\end{definition}

We will show later that this property holds in general cases by exploiting the localization of $A$ in $W_R$ (which corresponds to analyticity in rapidity space), and analyticity and boundedness properties of $S$. To explain how these properties imply nuclearity of $\Xi_n(s)$, let us first state the following nuclearity result for Hardy spaces on tubes.

\begin{proposition}\label{lemma:Hardy->Nuclearity}
     Let $\C_n\subset\Rl^n$ be an open polyhedron as in Def.~\ref{definition:property-h}, with distances $c_n>0$ \eqref{eq:c_n}. Given $s>0$, define the map
     \begin{align}\label{eq:def-X}
	  X_{\C_n,s}: H^2(\Tu_{\C_n},\K^{\otimes n})\to L^2(\Rl^n,\K^{\otimes n})
	  \,,\quad
	  (X_{\C_n,s}h)^\balpha(\bte)
	  :=
	  \prod_{k=1}^{n}e^{-\frac{s}{2}m_{[\alpha_k]}\cosh(\theta_k)}
	  \cdot
	  h_{\bla_{\pi/2}}^\balpha(\bte)
	  \,.
     \end{align}
     Then $X_{\C_n,s}$ is nuclear, with nuclear norm 
     \begin{align}\label{eq:nuclear-norm-X}
		\|X_{\C_n,s}\|_1
		\leq
		\left(
			\frac{\dim\K}{(\frac{1}{2}\pi m_\circ)^{1/2}}
			\cdot \frac{e^{-\frac{sm_\circ}{2}\cos c_n}}{c_n(s\cos c_n)^{1/2}}\right)^n
			\,,
     \end{align}
     where $m_\circ=\min\{m_{[\alpha]}\,:\,\alpha=1,\dots,\dim\K\}\in(0,\infty)$ is the mass gap \eqref{eq:mass-gap}.
\end{proposition}

The proof of Prop.~\ref{lemma:Hardy->Nuclearity} makes use of the following lemma.

\begin{lemma}\label{traceclassneu}
     Let $g\in L^2(\Rl)$, $b\in\Rl\backslash\{0\}$, and define an integral operator $R_{g,b}$ on $L^2(\mathbb{R})$ in terms of its integral kernel
     \begin{equation}\label{kernels1}
	  R_{g,b}(\theta,\theta')
	  :=
	  \frac{-{\rm sign}\,(b)}{2\pi i}\frac{\overline{g(\theta)}g(\theta')}{\theta'-\theta+ib}
	  \,.
     \end{equation}
     Then, $R_{g,b}$ is a positive trace class operator with trace norm
     \begin{equation}\label{tracenormneu}
	  \|R_{g,b}\|_1=\frac{\|g\|_2^2}{2\pi |b|}.
     \end{equation}
\end{lemma}
\begin{proof}
     Note first that $R_{g,b}=UR_{g_-,-b}U^*$, with the unitary $(Uf)(\theta):=i\cdot f(-\theta)$, $f\in L^2(\mathbb{R})$, and $g_-(\theta):=g(-\theta)$. Due to this unitary equivalence it suffices to consider $b>0$.
     
     To show that $R_{g,b}$ is positive, define $K_b(\theta):=-(2\pi i)^{-1}(\theta+ib)^{-1}$, which has positive Fourier transform $\widetilde{K}_b(\eta)=\Theta(\eta)e^{-b\eta}$. Then we have, $f\in L^2(\mathbb{R})$,
     \begin{align*}
	  \langle f,R_{g,b}f\rangle
	  =
	  \langle(g\cdot f),K_b\ast (g\cdot f)\rangle%_{L^2(\mathbb{R})}
	  =
	  \sqrt{2\pi}\langle\widetilde{(g\cdot f)},\widetilde{K_b}\cdot \widetilde{(g\cdot fa)}\rangle
	  \geq 0,
     \end{align*}
     yielding $R_{g,b}>0$. Since
     \begin{align*}
	  \int_\Rl d\te\,R_{g,b}(\te,\te)
	  =
	  -\frac{1}{2\pi i}\int d\theta\frac{\overline{g(\theta)}g(\theta)}{\theta-\theta+ib}=\frac{\|g\|_2^2}{2\pi b}
	  \,,
     \end{align*}
     it follows by \cite[Lemma on p.65]{RS3} that $R_{g,b}$ is trace class with trace $\frac{\|g\|_2^2}{2\pi b}$, proving the claim.
\end{proof}

\noindent{\em Proof of Prop.~\ref{lemma:Hardy->Nuclearity}.}
	As shorthand notations, we write $\C$ for the polyhedron $\C_n$ and $\bla$ for $\bla_{\pi/2}$ in this proof, and also observe that we may replace $\frac{sm_{[\alpha_k]}}{2}$ by $\delta:=\frac{sm_\circ}{2}$ in the exponentials in \eqref{eq:def-X} and estimate this larger operator instead.

     Let $h\in H^2(\Tu_{\C},\K^{\otimes n})$, and consider a closed polydisc $\overline{D_n(\boldsymbol{\theta}+i\boldsymbol{\lambda})}\subset\Tu_{\C_n}$ with center $\boldsymbol{\theta}+i\boldsymbol{\lambda}$, $\boldsymbol{\theta}\in\mathbb{R}^n$ and radius $r<d(\bla,\partial\C)$. By Cauchy's integral formula,
     \begin{equation*}
	  h^{\boldsymbol{\alpha}}(\boldsymbol{\theta}+i\boldsymbol{\lambda})
	  =
	  \frac{1}{(2\pi i)^n}
	  \oint\limits_{T_n(\boldsymbol{\theta}+i\boldsymbol{\lambda})}d^n\boldsymbol{\zeta}'
	  \frac{h^{\boldsymbol{\alpha}}(\boldsymbol{\zeta}')}{\prod_{k=1}^{n}\left(\zeta'_k-\theta_k+i\tfrac{\pi}{2}\right)},
     \end{equation*}
     where $T_n(\boldsymbol{\theta}+i\boldsymbol{\lambda})$ denotes the distinguished boundary of $D_n(\boldsymbol{\theta}+i\boldsymbol{\lambda})$. Since $\C$ is a polyhedron, the Hardy space properties recalled on p.~\pageref{eq:HardyNorm} can be used to the effect of deforming the contour of integration to the boundary of the tube based on the cube $\bla+(-r,r)^{\times n}$,
     \begin{equation}\label{eq:hardy-cauchy}
		h^{\boldsymbol{\alpha}}(\boldsymbol{\theta}+i\boldsymbol{\lambda})
		=
		\frac{1}{(2\pi i)^n}
		\sum_{\boldsymbol{\varepsilon}}\int_{\mathbb{R}^n}d^n\boldsymbol{\theta}'
		\left(
	       \prod_{k=1}^{n}\frac{\varepsilon_k}{\theta_k'-\theta_k-i\varepsilon_k r}
		\right)
		h^{\boldsymbol{\alpha}}(\boldsymbol{\theta}'+i(\boldsymbol{\lambda}-r\boldsymbol{\varepsilon})),
     \end{equation}
     where $\boldsymbol{\varepsilon}=(\varepsilon_1,\dots,\varepsilon_n)$, with $\varepsilon_k=\pm 1$, $k=1,\dots,n$.
     
     Since $\C\subset(-\pi,0)^{\times n}$ (see Def.~\ref{definition:property-h}), the function $u_\delta(\bzeta):=\prod_{k}e^{-i\delta\sinh\zeta_k}$ is analytic and rapidly decaying in the real directions in $\Tu_{\C}$. Hence \eqref{eq:hardy-cauchy} also holds for $u_\delta\cdot h$ instead of $h$. We may therefore split $u_\delta=u_{\delta/2}\cdot u_{\delta/2}$, and write the operator $X_{\C,s}$ as
     \begin{align*}
	  \left(X_{\C,s}h\right)^{\boldsymbol{\alpha}}(\boldsymbol{\theta})
	  &=
	  u_{\delta/2}(\bte+i\bla)
	  \cdot\left( u_{\delta/2}\cdot h\right)^{\boldsymbol{\alpha}}(\boldsymbol{\theta}+i\boldsymbol{\lambda})
	  \\
	  &=
	  \tfrac{u_{\delta/2}(\bte+i\bla)}{(2\pi i)^n}
	  \sum_{\boldsymbol{\varepsilon}}\int\limits_{\mathbb{R}^n}d^n\boldsymbol{\theta}'
	  \left(\prod_{k=1}^{n}\tfrac{\varepsilon_k}{\theta_k'-\theta_k-i\varepsilon_k r}\right)
	  \left( u_{\delta/2}\cdot 
	  h\right)^{\boldsymbol{\alpha}}(\boldsymbol{\theta}'+i(\boldsymbol{\lambda}-r\boldsymbol{\varepsilon}))
	  \\
	  &=
	  a_\delta(\bte)
	  \sum_{\boldsymbol{\varepsilon}}\int\limits_{\mathbb{R}^n}d^n\boldsymbol{\theta}'
	  \left(\prod_{k=1}^{n}\tfrac{\varepsilon_k\,e^{-\frac{\delta}{2}\cos(r)(\cosh\theta_k+\cosh\theta_k')}}{2\pi i(\theta_k'-\theta_k-ir\varepsilon_k)}\right)
	  b_{\delta,r,\boldsymbol{\varepsilon}}(\bte')
	  h^{\boldsymbol{\alpha}}(\boldsymbol{\theta}'+i(\boldsymbol{\lambda}-r\boldsymbol{\varepsilon})),
     \end{align*}
     where $a_\delta(\bte):=\prod_{j=1}^n e^{-\frac{\delta}{2}(1-\cos r)\cosh\te_j}$ and $b_{\delta,r,\boldsymbol{\varepsilon}}(\bte'):=\prod_{j=1}^n e^{\frac{i\,\delta}{2}\varepsilon_j\sin r\sinh\te_j'}$.
     
     We split this operation into four parts. On the very right, we have the evaluation operators 
     $E_{\boldsymbol{\lambda}-r\boldsymbol{\varepsilon}}:h\mapsto h_{\boldsymbol{\lambda}-r\boldsymbol{\varepsilon}}$ from $H^2(\Tu_n,\mathcal{K}^{\otimes n})$ to $L^2(\mathbb{R}^n,\mathcal{K}^{\otimes n})$, which are bounded with operator norm at most~1 for any $\boldsymbol{\varepsilon}$, by definition of the norm \eqref{eq:HardyNorm}.
     
     Next, there acts the unitary operator $B_{\delta,r,\boldsymbol{\varepsilon}}$ (on $L^2(\Rl^n,\K^{\otimes n})$) multiplying with $b_{\delta,r,\boldsymbol{\varepsilon}}$, and on the very left, we have the operator $A_\delta$ multiplying by $a_\delta$. Also $A_\delta$ is bounded with norm at most $1$ on $L^2(\Rl^n,\K^{\otimes n})$ because $1-\cos r>0$ as a consequence of $D_n(\bte+i\bla)\subset\Tu_\C\subset\Rl^n+i(-\pi,0)^{\times n}$. 
     
     Finally, the remaining integral kernel between $A_\delta$ and $B_{\delta,r,\boldsymbol{\varepsilon}}$ can be expressed in terms of the integral operators $R_{g,b}$ from Lemma~\ref{traceclassneu}. Namely, defining $g_{\delta,r}(\theta):=\exp[-\frac{\delta}{2}\cos r\cosh\theta]$ and $\hat{R}_{\delta,r,\varepsilon}:=R_{g_{\delta,r},-\varepsilon r}\otimes1$ on $L^2(\Rl)\otimes\K$, we have 
     \begin{eqnarray}\label{x}
	  X_{\C,s}
	  &=& 
	  A_\delta\sum_{\boldsymbol{\varepsilon}} 
	  \left(
	       \hat R_{\delta,r,\varepsilon_1}\otimes\cdots\otimes \hat R_{\delta,r,\varepsilon_n}
	  \right)
	  \,B_{\delta,r,\boldsymbol{\varepsilon}}\, E_{\boldsymbol{\lambda}-r\boldsymbol{\varepsilon}}
	  \,.
     \end{eqnarray}
     Taking into account the norm bounds $\|A_\delta\|\leq1$, $\|B_{\delta,r,\boldsymbol{\varepsilon}}\|\leq1$, $\|E_{\boldsymbol{\lambda}-r\boldsymbol{\varepsilon}}\|\leq1$ as well as $\|\hat R_{\delta,r,\varepsilon}\|_1=\dim\K\cdot\frac{\|g_{\delta,r}\|_2^2}{2\pi r}$ (Lemma~\ref{traceclassneu}), we may use Lemma~\ref{propertiesNuclearMaps} to conclude
     \begin{align*}
	  \|X_{\C,s}\|_1
	  &\leq
	  \sum_{\boldsymbol{\varepsilon}} 
	  \|\hat R_{\delta,r,\varepsilon_1}\otimes\cdots\otimes \hat R_{\delta,r,\varepsilon_n}\|_1
	  =
	  \left(\dim\K\cdot\frac{\|g_{\delta,r}\|_2^2}{\pi r}\right)^n
	  <
	  \infty\,.
     \end{align*}
     This proves that  $X_{\C,s}$ is nuclear. To also establish the claimed bound on $\|X_{\C,s}\|_1$, we estimate
     \begin{align*}
	  \|g_{\delta,r}\|_2^2
	  =
	  \int_\Rl d\te\,e^{-\delta\cos r\,\cosh\te}
	  \leq
	  \int_\Rl d\te\,e^{-\delta\cos r\,(1+\frac{\te^2}{2}})
	  =
	  e^{-\delta\cos r}\,\sqrt{\frac{\pi}{\delta\cos r}}
	  \,.
     \end{align*}
     Now letting $r\to \|\bla-\partial\C\|_\infty$ finishes the proof.\hfill$\square$
     
\medskip

Clearly, the Hardy space operators $X_{\C_ns,}$ \eqref{eq:def-X} resemble the action of the maps $\Xi_n(s)$ \eqref{concrete}. This connection is exploited to prove the following theorem.

\begin{theorem}\label{thm:general-nuclearity}
	Assume that $S\in\SF$ has property (H), with polyhedra $\C_n$ and constants $c_n,\upsilon(s,n)$ (Def.~\ref{definition:property-h}).
	\begin{enumerate}
		\item For any $n\in\Nl_0,s>0$, the map $\Xi_n(s)$ is nuclear, with nuclear norm bounded by
			\begin{align}\label{eq:general-Xi_n-nuclear-norm-bound}
				\|\Xi_n(s)\|_1
				\leq
				\upsilon(s,n)\cdot 
				\|X_{\C_n,s}\|_1
				<\infty
				\,.
			\end{align}
		\item $\Xi(s)$ is nuclear for all $s>0$ satisfying
		\begin{align}\label{eq:nuclear-norm-series}
			\sum_n
			\upsilon(s,n)\cdot\|X_{\C_n,s}\|_1
			<\infty
			\,.
		\end{align}
	\end{enumerate}
\end{theorem}
\noindent\textit{Remark.}
Part $ii)$ of this theorem gives an abstract sufficient condition for modular nuclearity to hold. To make use of it in concrete models, one has to find suitable bounds on $\upsilon(s,n)$, ensuring that the series converges for finite $s$.\par
Both factors $\upsilon(s,n)$ and $\|X_{\mathcal{C}_n,s}\|_1$ in the series (\ref{eq:nuclear-norm-series}) depend on the size of the polyhedra $\C_n$. If $S\in\mathcal{S}$ has property (H) for some $\C_n$, it clearly has this property on any smaller polyhedron as well. Later we will see that while $\|X_{\mathcal{C}_n,s}\|_1$ becomes larger with shrinking $\C_n$, cf. (\ref{eq:nuclear-norm-X}), the norm $\upsilon(s,n)$ becomes smaller. It will therefore later be important to find a fine balance between these two quantities.

\begin{proof}
	$i)$ We consider a decomposition $\Xi_n(s)$ into a product of a bounded and a nuclear map. According to property (H), the map
	\begin{align}
		\Upsilon_{\C_n,s} : \F(W_R)\to H^2(\Tu_{\C_n},\K^{\otimes n})
		\,,\qquad
		A\longmapsto (U_{s/2}A\Omega)_n
	\end{align}
	is a bounded linear map between the Banach spaces $(\F(W_R),\|\cdot\|_{\B(\Hil)})$ and $(H^2(\Tu_{\C_n},\K^{\otimes n}),\triplenorm\cdot\triplenorm)$, with operator norm $\|\Upsilon_{\C_n,s}\|\leq\upsilon(s,n)<\infty$. By Prop.~\ref{lemma:Hardy->Nuclearity}, the map $X_{\C_n,s}:H^2(\Tu_{\C_n},\K^{\otimes n})\to L^2(\Rl^n,\K^{\otimes n})$ is nuclear. But by comparison of \eqref{concrete} and the maps $\Upsilon_{\C_n,s}$ and $X_{\C_n,s}$, it follows that 
	\begin{align}
		\Xi_n(s)
		=
		X_{\C_n,s}\,\Upsilon_{\C_n,s}
		\,.
	\end{align}
	By Lemma~\ref{propertiesNuclearMaps}~$iii)$, this implies that $\Xi_n(s)$ is nuclear, with $\|\Xi_n(s)\|_1\leq\|\Upsilon_{\C_n,s}\|\cdot\|X_{\C_n,s}\|_1$. 
	
	$ii)$ The series \eqref{eq:nuclear-norm-series} dominates $\sum_n\|\Xi_n(s)\|_1$ by part $i)$. Hence its convergence implies convergence of $\sum_n\Xi_n(s)$ in nuclear norm. But the set of all nuclear maps between two Banach spaces is closed in nuclear norm (Lemma~\ref{propertiesNuclearMaps}~$i)$). Hence the convergence of \eqref{eq:nuclear-norm-series} implies nuclearity of $\Xi(s)$.
\end{proof}

This theorem shows that nuclearity of $\Xi(s)$ follows if analyticity of $(A\Omega)_n$ in sufficiently large tubes around $\bla_{\pi/2}$ can be established (which gives lower bounds on the $c_n$), and if sharp enough bounds on the analytic continuations of $(U_{\frac{s}{2}}A\Omega)_n$ can be obtained (which imply upper bounds on $\upsilon(s,n)$). 

Thm.~\ref{thm:general-nuclearity} thus reduces the inverse scattering problem to questions in complex analysis. These aspects require a detailed investigation of the functions $(A\Omega)_n$, which is carried out in the next section.

\section{Hardy space properties of wedge-local wavefunctions}\label{ChapterHardyS}

In this section, we demonstrate that any regular S-matrix (Def.~\ref{regularS}) has property (H) (Def.~\ref{definition:property-h}), and estimate the constants $\upsilon(s,n)$, $c_n$. We will adopt and suitably generalize the strategy used in the scalar case \cite{DocL, L08}.

\subsection{Analytic and combinatorial structure of contracted matrix elements}\label{Section:Analyticity+Combinatorics}

To establish the Hardy properties of Def.~\ref{definition:property-h}, we start by deriving analyticity properties of the functions $(A\Omega)^{\boldsymbol{\alpha}}_n$, $A\in\mathcal{F}(W_R)$. The first step in revealing these properties is given by expressing $(A\Omega)^{\boldsymbol{\alpha}}_n$ as the matrix elements
\begin{equation}\label{eq:An}
	(A\Omega)_n^{\boldsymbol{\alpha}}(\boldsymbol{\theta})=\frac{1}{\sqrt{n!}}\langle z_{\alpha_1
	}^\dagger(\theta_1)\cdots z_{\alpha_n
	}^\dagger(\theta_n)\Omega,A\Omega\rangle\,,
\end{equation}
and relating the $z^\dagger_\beta(\te)$ to the time zero fields of the ``left-local field'' $\phi$ \eqref{field1}. These are
\begin{equation}
	\varphi_\alpha(x_1):=\sqrt{2\pi}\phi_\alpha(0,x_1),\qquad \pi_\alpha(x_1):=\sqrt{2\pi}(\partial_0\phi)_\alpha(0,x_1),\qquad x_1\in\mathbb{R},
\end{equation}
to be understood in the sense of operator-valued distributions. In view of \eqref{field1}, their smeared versions are, $f\in\mathscr{S}(\mathbb{R})\otimes\mathcal{K}$,
\begin{align}\label{zero}
	\varphi(f)&=z^\dagger(\hat{f})+z(J\hat{f}_-)
	\,,\qquad
	\pi(f)=i\left(z^\dagger(\omega\hat{f})-z(\omega J\hat{f}_-)\right),
\end{align}
where $\hat{f}^{\alpha}(\theta):=\widetilde{f}^\alpha(m_{[\alpha]} \sinh\theta)$, $\hat{f}_-^{\alpha}(\theta):=\widetilde{f}^\alpha(-m_{[\alpha]} \sinh\theta)$, and (no sum over $\alpha$)
\begin{equation}
	\left(\omega\,\Phi\right)_1^\alpha(\theta):=\omega_{[\alpha]}(\theta)\,\Phi_1^\alpha(\theta),\qquad \omega_{[\alpha]}(\theta):=m_{[\alpha]}\cosh\theta,\qquad\Phi\in\dom(\omega)\subset\mathscr{H}_1
\end{equation}
is the one particle Hamiltonian. 

The operators $\varphi(f)$ and $\pi(f)$ are well-defined on the space $\mathcal{D}$ of finite particle number. Moreover, they are real in the sense that $\varphi(f^*)\subset\varphi(f)^*$, $\pi(f^*)\subset\pi(f)^*$, and ``left-local'' in the sense that for $A\in\mathcal{F}(W_R)$ and supp$\,f \subset\mathbb{R}_-$
\begin{equation}
     [\varphi(f),A]\Psi=0,\qquad[\pi(f),A]\Psi=0,\qquad\Psi\in\mathcal{D}.
\end{equation}
These commutation relations can be proven by arguments analogous to those yielding (\ref{commutatorwedge}). 

\bigskip
\bigskip

We now derive first Hardy space properties. This initial step is concerned with functions of $n=1$ variable, in which case a tube based on a polyhedron is simply an open interval, and the tube based on it an open strip region in the complex plane. Based on the localization of $A$ in the right wedge and $\varphi_\alpha$, $\pi_\beta$ on the left, we obtain the following generalization of a result established in the scalar case \cite{L08}.

\begin{lemma}\label{basicLemma}
	Given $A\in\mathcal{F}(W_R)$, $n_1,n_2\in\mathbb{N}_0$, and $\Psi_i\in\mathcal{H}_{n_i}$, $i=1,2$, define functionals $K,K^\dagger:\mathscr{S}(\mathbb{R})\otimes\mathcal{K}\rightarrow\mathbb{C}$ by
	\begin{equation}\label{functionalC}
		K(\hat{f}):=\langle\Psi_1,[z(\hat{f}),A]\Psi_2\rangle,\qquad K^\dagger(\hat{f}):=\langle\Psi_1,[z^\dagger(\hat{f}),A]\Psi_2\rangle,
	\end{equation}
	where $\hat{f}^\alpha(\theta):=\widetilde{f}^\alpha(m_{[\alpha]}\sinh\theta)$. Then there exists a function $\hat{K}\in H^2(\Strip(-\pi,0))\otimes\K$ whose boundary values satisfy
	\begin{equation}
		K(\hat{f})=\sum_\alpha\int d\theta\, \hat{K}_{\alpha}(\theta) \overline{\hat{f}^\alpha(\theta)},\qquad
		K^\dagger(\hat{f})=-\sum_\alpha\int d\theta\, \hat{K}_{\overline{\alpha}}(\theta-i\pi) \hat{f}^\alpha(\theta),
	\end{equation}
	and whose Hardy norm is bounded by
	\begin{equation}\label{hardynorm}
		\triplenorm \hat{K}\triplenorm\leq \left((n_1+1)^{1/2}+(n_2+1)^{1/2} \right)\|\Psi_1\|\,\|\Psi_2\|\,\|A\|.
	\end{equation}
\end{lemma}
\begin{proof}
	We first derive a bound on the functionals $K,K^\dagger$ \eqref{functionalC}. By application of the particle number bounds (\ref{numberBounds}) and the Cauchy-Schwarz inequality, we have
	\begin{eqnarray*}
		|K(\hat{f})|&\leq& \|z^\dagger(\hat{f})\Psi_1\|\|A\Psi_2\|+\|A^*\Psi_1\|\|z(\hat{f})\Psi_2\|\\
		&\leq&\left(\sqrt{n_1+1}+\sqrt{n_2}\right)\|\Psi_1\|\|\Psi_2\|\|A\|\,\|\hat{f}\|,\\
		|K^\dagger(\hat{f})|&\leq&\left(\sqrt{n_1}+\sqrt{n_2+1}\right)\|\Psi_1\|\|\Psi_2\|\|A\|\,\|\hat{f}\|.
	\end{eqnarray*}
	It now follows from Riesz' Lemma that the distributions $K$ and $K^\dagger$ are given by integration against functions $\hat{K},\hat{K}^\dagger\in\mathscr{H}_1=L^2(\mathbb{R})\otimes\mathcal{K}$, with norms bounded by
	\begin{equation}\label{esti1}
		\|\hat{K}^{\#}\|_2\leq\left(\sqrt{n_1+1}+\sqrt{n_2+1}\right)\|\Psi_1\|\|\Psi_2\|\|A\|.
	\end{equation}
	Next, we use the time zero fields (\ref{zero}) to derive the claimed analytic structure. To this end, consider the functionals $K_\pm:\mathscr{S}(\mathbb{R})\otimes\mathcal{K}\rightarrow\mathbb{C}$,
	\begin{equation}
		K_-(f):=\langle\Psi_1,[\varphi(f),A]\Psi_2\rangle,\qquad K_+(f):=\langle\Psi_1,[\pi(f),A]\Psi_2\rangle\,.
	\end{equation}
	Taking into account $\widehat{f^*}^\alpha(\theta)=(J\hat{f}_-)^\alpha(\theta)$, we observe that solving \eqref{zero} for $z,z^\dagger$ gives
	\begin{eqnarray}
		z(\hat{f}) =\dfrac{1}{2}\left(\varphi(f^*)+i\pi(\omega^{-1}f^*)\right)
		\,,\qquad z^\dagger(\hat{f})&=&\dfrac{1}{2}\left(\varphi(f)-i\pi(\omega^{-1}f)\right)\,,
	\end{eqnarray}
	and thus the same relation between the distributions $K,K^\dagger, K_\pm$. 
	
	Since $\varphi$ and $\pi$ are localized on the left, and $A$ in the right wedge, we have  supp$\,K_\pm\subset\mathbb{R}_+$. Thus, there exist functions $p\mapsto\widetilde{K}_{\pm}(p)$ which are analytic in the lower half plane, satisfy polynomial bounds at the real boundary and at infinity, and reproduce the Fourier transforms of $K_\pm$ as their distributional boundary values \cite[Thm. IX.16]{RS2}. As $\sinh$ maps the strip $\Strip(-\pi,0)$ to the lower half plane, this also implies that 
	\begin{equation}\label{Cpm}
		\hat{K}_{+,\alpha}(\theta):=\widetilde{K}_{+,\alpha}(m_{[\alpha]} \sinh\theta)
		\,,\qquad 
		\hat{K}_{-,\alpha}(\theta)
		:=
		m_{[\alpha]} \cosh\theta\cdot\widetilde{K}_{-,\alpha}(m_{[\alpha]} \sinh\theta),
	\end{equation}
	are analytic in this strip. To relate these functions to $K$, $K^\dagger$, we compute
	\begin{eqnarray*}
		K(\hat{f})&=&\dfrac{1}{2}\left(K_-(f^*)+iK_+(\omega^{-1}f^*)\right)=\dfrac{1}{2}\int dp\left(\widetilde{K}_{-,\alpha}(p)+i\omega(p)^{-1}\widetilde{K}_{+,\alpha}(p)\right)\overline{\widetilde{f}^{\overline{\alpha}}(p)}\\
		&=&\dfrac{1}{2}\int d\theta \left(\hat{K}_{-,\overline{\alpha}}(\theta)+i\hat{K}_{+,\overline{\alpha}}(\theta)\right)\overline{\hat{f}^{\alpha}(\theta)}
		=\int d\theta\, \hat{K}_{\alpha}(\theta)\overline{\hat{f}^\alpha(\theta)}
	\end{eqnarray*}
	and
	\begin{eqnarray*}
		K^\dagger(\hat{f})=\dfrac{1}{2}\left(K_-(f)-iK_+(\omega^{-1}f)\right)
		&=&\dfrac{1}{2}\int d\theta \left(\hat{K}_{-,\alpha}(-\theta)-i\hat{K}_{+,\alpha}(-\theta)\right)\hat{f}^\alpha(\theta)\\
		&=&
		\int d\theta\, \hat{K}^\dagger_\alpha(\theta)\hat{f}^\alpha(\theta)
		\,.
	\end{eqnarray*}
	These equations imply
	\begin{equation}
		\hat{K}_{\overline{\alpha}}(\theta)=\dfrac{1}{2}\left(\hat{K}_{-,\alpha}(\theta)+i\hat{K}_{+,\alpha}(\theta)\right),\qquad \hat{K}_\alpha^\dagger(\theta)=\dfrac{1}{2}\left(\hat{K}_{-,\alpha}(-\theta)-i\hat{K}_{+,\alpha}(-\theta)\right),
	\end{equation}
	and, in particular, the analyticity of $\theta\mapsto \hat{K}_\alpha(\theta)$ in the strip $\Strip(-\pi,0)$.
	Furthermore, it follows that the boundary values of $\hat{K}_\pm$ also exist as functions in $L^2(\mathbb{R})\otimes\mathcal{K}$. Since $\hat{K}_{\pm,\alpha}(\theta-i\pi)=\pm\hat{K}_{\pm,\alpha}(-\theta)$ for $\theta\in\mathbb{R}$ by (\ref{Cpm}), we have $\hat{K}_\alpha^\dagger(\theta)=-\hat{K}_{\overline{\alpha}}(\theta-i\pi)$.
	
	It remains to prove that $\hat{K}$ is an element of the Hardy space $H^2(S(-\pi,0))\otimes\K$. For that purpose, we consider $\hat{K}_\alpha^{(s)}(\zeta):=e^{-im_{[\alpha]}s\sinh\zeta}\hat{K}_\alpha(\zeta)$, with $s>0$, which is clearly analytic in the strip $S(-\pi,0)$ as well. The identity
	\begin{equation}\label{esti2}
		\left|\hat{K}^{(s)}_{-\lambda,\alpha}(\theta)\right|=\frac{1}{2}e^{-m_{[\alpha]}s\sin\lambda\cosh\theta}\left|\hat{K}_{-,\overline{\alpha}}(\theta-i\lambda)+i\hat{K}_{+,\overline{\alpha}}(\theta-i\lambda)\right|
	\end{equation}
	yields that $\hat{K}^{(s)}_{-\lambda,\alpha}\in L^2(\mathbb{R})$ for all $\lambda\in[0,\pi]$ and $s>0$, since $\theta\mapsto\hat{K}_{\pm,\alpha}(\theta-i\lambda)$ is bounded by polynomials in $\cosh\theta$ for $\theta\rightarrow\infty$ and $0<\lambda<\pi$. Noting that $\|\hat{K}^{(0)}_{0/-\pi}\|_2=\|\hat{K}^{(s)}_{0/-\pi}\|_2$ and (\ref{esti1}), the three lines theorem may be applied and we arrive at
	\begin{equation}\label{esti3}
		\|\hat{K}^{(s)}_{-\lambda}\|_2\leq\left(\sqrt{n_1+1}+\sqrt{n_2+1}\right)\|\Psi_1\|\|\Psi_2\|\|A\|,\qquad 0\leq\lambda\leq\pi.
	\end{equation}
	Since (\ref{esti2}) is monotonically increasing for $s\rightarrow 0$, it follows that the uniform bound (\ref{esti3}) holds also for $\hat{K}_{-\lambda}=\hat{K}_{-\lambda}^{(0)}$, with $0\leq\lambda\leq \pi$. This finishes the proof.
\end{proof}

\bigskip

Lemma \ref{basicLemma} is our basic tool to derive analyticity properties of the rapidity functions $(A\Omega)_n^\balpha$ from the localization of $A$ in the right wedge $W_R$, as we shall explain now. In view of the properties of the creation operators $\zd_\alpha(\te)$ (see Prop.~\ref{proposition:z-operators}), we may write
\begin{align*}
     \sqrt{n!}\,(A\Omega)_n^{\boldsymbol{\alpha}}(\boldsymbol{\theta})
     &=
     \langle z_{\alpha_1}^\dagger(\theta_1)\cdots z_{\alpha_n}^\dagger(\theta_n)\Omega,A\Omega\rangle
     \\
     &=
     \langle z_{\alpha_2}^\dagger(\theta_2)\cdots z_{\alpha_n}^\dagger(\theta_n)\Omega,[z_{\alpha_1}(\theta_1),A]\Omega\rangle
     \,,
\end{align*}
and may therefore apply this Lemma \ref{basicLemma} to conclude that $(A\Omega)_n^\balpha$ is analytic in the variable $\theta_1$ in the strip $(-\pi,0)$. Its boundary value at Im$(\theta_1)=-\pi$ is given by
\begin{align}\label{Rechnung}
	(A\Omega)_n^\balpha&(\te_1-i\pi,\te_2,...,\te_n)
	=
	\langle z_{\alpha_2}^\dagger(\theta_2)\cdots z_{\alpha_n}^\dagger(\theta_n)\Omega,[A,z^\dagger_{\overline{\alpha}_1}
	(\theta_1)]\Omega\rangle
	\\
	&=
	\langle z_{\alpha_2
	}^\dagger(\theta_2)\cdots z_{\alpha_n
	}^\dagger(\theta_n)\Omega,Az^\dagger_{\overline{\alpha}_1}
	(\theta_1)\Omega\rangle-\langle z_{\overline{\alpha}_1}
	(\theta_1) z_{\alpha_2
	}^\dagger(\theta_2)\cdots z_{\alpha_n
	}^\dagger(\theta_n)\Omega,A\Omega\rangle
	\nonumber
	\\
	&=
	\langle z_{\alpha_2}^\dagger(\theta_{2})\cdots z_{\alpha_n}^\dagger(\theta_{n})\Omega,Az^\dagger_{\overline{\alpha}_1}(\theta_{1})\Omega\rangle
	-	\sum_{l=2}^{n}\delta(\theta_l-\theta_1)\delta^{\alpha_l\overline{\xi_{l-1}}}\,\delta^{\overline{\alpha}_1\overline{\xi}_1}\prod_{m=2}^{l-1}S^{\alpha_m\overline{\xi}_m}_{\overline{\xi}_{m-1}\beta_m}(\theta_1-\theta_m)
	\nonumber
	\\
	&\qquad\qquad\qquad\qquad\qquad\qquad\qquad\times
	\langle z_{\beta_2}^\dagger(\theta_{2})\cdots z_{\beta_{l-1}}^\dagger(\theta_{l-1}) z_{\alpha_{l+1}}^\dagger(\theta_{l+1}) \cdots z_{\alpha_n}^\dagger(\theta_{n})\Omega,A
	\Omega\rangle\nonumber
	\,,
\end{align}
where we used the Zamolodchikov exchange relations (\ref{exchange}) in the last step.

\medskip

To establish analyticity of $(A\Omega)_n^\balpha$ not just in the single variable $\te_1$, but in an $n$-dimensional tube in $\Cl^n$, we now move the leading creation operators $z^\dagger_{\alpha_2}(\te_2)$ from the left to the right hand side, and rewrite the above expression in terms of expectation values of the commutator $[z_{\alpha_2}(\te_2),A]$. In this form, Lemma~\ref{basicLemma} can be applied again, now yielding analyticity in the second variable $\te_2$. As we will show below, this type of argument results in an iterative procedure with which we can successively analytically continue in all variables $\te_1,...,\te_n$.

\medskip

As can be seen from \eqref{Rechnung}, this scheme will produce sums of products of delta distributions, $S$-factors and matrix elements of $A$. To organize these terms efficiently, we will now introduce a graphical notation\footnote{For an alternative algebraic description emphasizing the role of the representations $D_n$ from Lemma~\ref{lemma:Dn} and avoiding diagrammatic notation, see~\cite{Alazzawi:2014}.} for certain (contracted) matrix elements of $A$. 

All our diagrams will consist of a number of oriented lines, which start/end either at {\em external vertices} at the top of the diagram, or at the bottom. Each line carries an index $\alpha\in\{1,...,\dim\K\}$ and a rapidity $\te\in\Rl$, which we indicate by a label $\alpha, [\te]$ where necessary. The basic element of our graphical notation is
\begin{equation}\label{eq:def-A-matrix-element}
	 \begin{tikzpicture}[baseline=(current bounding box.center),scale=1.0,line width=1.0pt,>=latex]
	       \draw[line width=2.0pt, gray] (0,0) to (2,0) {};
	       \draw[line width=2.0pt, gray] (2.5,0) to (4.5,0) {};
	       \draw[->] (0.4,0) to (0.4,1.5) {};
	       \draw[->] (1.6,0) to (1.6,1.5) {};
	       \node at (1,0.7) {$\cdots$};
	       \draw[<-] (2.9,0) to (2.9,1.5) {};
	       \draw[<-] (4.1,0) to (4.1,1.5) {};
	       \node at (3.5,0.7) {$\cdots$};
	       \node[above] at (0.4,1.5) {${\lambda_1}\atop{[\eta_1]}$};
	       \node[above] at (1.6,1.5) {${\lambda_\ell}\atop{[\eta_\ell]}$};
	        \node[above] at (2.9,1.5) {${\rho_1}\atop{[\te_1]}$};
	       \node[above] at (4.1,1.5) {${\rho_r}\atop{[\te_r]}$};
	       \node at (9,1.2) {$:=\langle\Omega,\,z_{\rho_r}(\te_r)\cdots z_{\rho_1}(\te_1)\,A\,\zd_{\lambda_\ell}(\eta_\ell)\cdots\zd_{\lambda_1}(\eta_1)\Omega\rangle$};
	       \node at (9.5,0.5) {$=\langle\Omega,\, z_{\lambda_1}(\eta_1)\cdots z_{\lambda_\ell}(\eta_\ell)\,A^*\,\zd_{\rho_1}(\te_1)\cdots\zd_{\rho_r}(\te_r)\Omega\rangle^*$\,.};
	 \end{tikzpicture}
	 \end{equation}
	 \medskip
	 
We have included the second formula with the conjugate matrix element because in this form, the ordering of the operators matches the ordering of the lines in the diagram, incoming lines represent creation operators, and outgoing lines represent annihilation operators. In the following, by ``left'' and ``right'' we will always refer to the parts of the diagram as shown, or the order in the second (conjugate) matrix element.

\begin{wrapfigure}{r}{0.3\textwidth}
 \begin{center}
     \begin{tikzpicture}[scale=1.0,line width=1.0pt,>=latex]
	  \draw[<-] (0,0) to (1,1) {};
	  \draw[->] (0,1) to (1,0) {};
	  \node[below] at (0,0) {\small$\gamma$};
	  \node[above] at (0,1) {${\alpha}\atop{[\te]}$};
	  \node[above] at (1,1) {${\beta}\atop{[\te']}$};
	  \node[below] at (1,0) {\small$\delta$};
	  \node at (3,0.5) {$=S^{\alpha\beta}_{\gamma\delta}(\te'-\te)$};
	  \node at (2,-1) {$=$};
	  \draw[->] (3.5,-0.5) to (2.5,-1.5) {};
	  \draw[->] (3.5,-1.5) to (2.5,-.5) {};
	  \node[left] at (2.5,-0.5) {\small$\gamma$};
	  \node[right] at (3.5,-1.5) {${\beta}\atop{[\te']}$};
	  \node[right] at (3.5,-0.5) {${\alpha}\atop{[\te]}$};
	  \node[left] at (2.5,-1.5) {\small$\delta$};
     \end{tikzpicture}
 \end{center}
\end{wrapfigure}

\medskip

Besides these matrix elements of $A$, also $S$-factors and delta distributions will be represented in our graphical notation. This is done in close analogy to the conventions used in the context of knot diagrams \cite{Kauffman:1993}: A crossing\footnote{We do not have to distinguish between over- and undercrossings because $S$ induces a representation of the permutation group instead of the braid group, see Def.~\ref{smatrix}~$ii)$.} between two oriented lines corresponds to an $S$-factor as shown in the picture on the right --- the two upper indices of $S$ correspond to the indices of the two incoming lines (ordered left to right), the two lower indices  of $S$ correspond to the two outgoing lines (ordered left to right), and the argument of $S$ is the rapidity of the right incoming line minus the rapidity of the left incoming line.

As in this picture, also in the following the rapidities of lines are always taken to stay the same when crossing with other lines, but the index may change, i.e. we assign an individual index to each line segment between external vertices and/or crossings.

With these conventions, the first two Zamolodchikov exchange relations \eqref{exchange} imply 
\begin{equation}\label{eq:zz-graphical-1}
     \begin{tikzpicture}[baseline=(current bounding box.center),scale=1.0,line width=1.0pt,>=latex]
	       \draw[line width=2.0pt, gray] (0,0) to (2,0) {};
	       \draw[line width=2.0pt, gray] (2.5,0) to (4.5,0) {};
	       \node at (0.2,0.7) {$\cdots$};
	       \draw[->] (0.6,0) to (0.6,1.5) {};
	       \draw[->] (1,0) to (1,1.5) {};
	       \node at (1.6,0.7) {$\cdots$};
	        \node at (2.7,0.7) {$\cdots$};
	       \draw[<-] (3.1,0) to (3.1,1.5) {};
	       \draw[<-] (3.5,0) to (3.5,1.5) {};
	       \node at (4.1,0.7) {$\cdots$};
	       \node at (5.5,0.7) {$=$};
	       \draw[line width=2.0pt, gray] (6.5,0) to (8.5,0) {};
	       \draw[line width=2.0pt, gray] (9,0) to (11,0) {};
	       \node at (6.7,0.7) {$\cdots$};
	       \draw[->] (7.1,0) to (7.5,1.5) {};
	       \draw[->] (7.5,0) to (7.1,1.5) {};
	       \node at (8.1,0.7) {$\cdots$};
	        \node at (9.2,0.7) {$\cdots$};
	       \draw[<-] (9.6,0) to (10,1.5) {};
	       \draw[<-] (10,0) to (9.6,1.5) {};
	       \node at (10.6,0.7) {$\cdots$};
	 \end{tikzpicture}
\end{equation}

Here we have suppressed the indices/rapidities labeling the lines, and the equation holds when arbitrary indices/rapidities are inserted, identical and in the same order on the external vertices of the left and right hand sides. The proof of \eqref{eq:zz-graphical-1} amounts to inserting such labels according to \eqref{eq:def-A-matrix-element} and carefully observing the index positions in \eqref{exchange}.

\begin{wrapfigure}{r}{0.3\textwidth}
 \begin{center}
     \begin{tikzpicture}[scale=1.0,line width=1.0pt,>=latex]
	  \draw[<-] (0,0) to [out=-90, in=-90,looseness=2] (1.5,-.07) {};
	  \node[above] at (0,0) {${\alpha}\atop{[\eta]}$};
	  \node[above] at (1.5,0) {${\beta}\atop{[\te]}$};
	  \node at (3,-0.3) {\small$=\delta_{\alpha\beta}\,\delta(\te-\eta)$};
     \end{tikzpicture}     
\end{center}
\end{wrapfigure}
\medskip
For the mixed exchange relation \eqref{eq:zz-mixed}, we need to introduce ``contractions'' between rapidities and/or indices. As usual in the context of knot partition functions, a line between two external vertices with rapidities $\te,\te'$ represents a delta distribution $\delta(\te-\te')$, and in case this line does not cross any other lines, also a Kronecker delta $\delta_{\alpha\beta}$ between the indices $\alpha,\beta$ of the two external vertices is understood (see picture on the right).

Together with our convention that incoming lines represent creation operators, and outgoing ones annihilation operators, the mixed commutation Zamolodchikov relation \eqref{eq:zz-mixed} then reads for the right hand side of the diagrams \eqref{eq:def-A-matrix-element}
\begin{eqnarray}\label{eq:zzd-graphical-1}
     \begin{tikzpicture}[baseline=(current bounding box.center),scale=0.8,line width=1.0pt,>=latex]
	       \draw[line width=2.0pt, gray, dashed] (0,0) to (0.5,0) {};
	       \draw[line width=2.0pt, gray] (1,0) to (3,0) {};
	       \node at (0.35,0.7) {$\cdots$};
	        \draw[->] (1.2,0) to (1.2,1.5) {};
	       \draw[<-] (1.5,0) to (1.5,1.5) {};
	       \draw[<-] (1.8,0) to (1.8,1.5) {};
	       \draw[<-] (2.8,0) to (2.8,1.5) {};
	       \node at (2.3,0.7) {$\cdots$};
	     \end{tikzpicture}
	     &=&
	     \begin{tikzpicture}[baseline=(current bounding box.center),scale=0.8,line width=1.0pt,>=latex]
	       \draw[line width=2.0pt, gray, dashed] (0,0) to (0.5,0) {};
	       \draw[line width=2.0pt, gray] (1,0) to (3,0) {};
	       \node at (0.35,0.7) {$\cdots$};
	        \draw[<-] (1.2,1.5) to [out=-90, in=-90, looseness=5] (1.5,1.5) {};
	       \draw[<-] (1.8,0) to (1.8,1.5) {};
	       \draw[<-] (2.8,0) to (2.8,1.5) {};
	       \node at (2.3,0.7) {$\cdots$};
	     \end{tikzpicture}
	     +
	     \begin{tikzpicture}[baseline=(current bounding box.center),scale=0.8,line width=1.0pt,>=latex]
	       \draw[line width=2.0pt, gray, dashed] (0,0) to (0.5,0) {};
	       \draw[line width=2.0pt, gray] (1,0) to (3,0) {};
	       \node at (0.35,0.7) {$\cdots$};
	        \draw[->] (1.8,1.5) to [out=-90, in=-90, looseness=4] (1.2,1.5) {};
	       \draw[<-] (1.5,0) to (1.5,1.5) {};
	       \draw[<-] (2.8,0) to (2.8,1.5) {};
	       \node at (2.3,0.7) {$\cdots$};
	     \end{tikzpicture}
	     +...+
	     \begin{tikzpicture}[baseline=(current bounding box.center),scale=0.8,line width=1.0pt,>=latex]
	       \draw[line width=2.0pt, gray, dashed] (0,0) to (0.5,0) {};
	       \draw[line width=2.0pt, gray] (1,0) to (3,0) {};
	       \node at (0.35,0.7) {$\cdots$};
	        \draw[->] (2.8,1.5) to [out=-90, in=-90, looseness=1.4] (1.2,1.5) {};
	       \draw[<-] (1.5,0) to (1.5,1.5) {};
	       \draw[<-] (1.8,0) to (1.8,1.5) {};
	       \node at (2.3,0.7) {$\cdots$};
	     \end{tikzpicture}
	     \,,
\end{eqnarray}
and for the left hand side of \eqref{eq:def-A-matrix-element}
\begin{eqnarray}\label{eq:zzd-graphical-2}
  \begin{tikzpicture}[baseline=(current bounding box.center),scale=0.8,line width=1.0pt,>=latex]
	       \draw[line width=2.0pt, gray] (0,0) to (2,0) {};
	       \draw[line width=2.0pt, gray, dashed] (2.5,0) to (3,0) {};
	       \node at (0.75,0.7) {$\cdots$};
	        \draw[->] (.2,0) to (.2,1.5) {};
	       \draw[->] (1.2,0) to (1.2,1.5) {};
	       \draw[->] (1.5,0) to (1.5,1.5) {};
	       \draw[<-] (1.8,0) to (1.8,1.5) {};
	       \node at (2.7,0.7) {$\cdots$};
	     \end{tikzpicture}
	     &=&
	      \begin{tikzpicture}[baseline=(current bounding box.center),scale=0.8,line width=1.0pt,>=latex]
	       \draw[line width=2.0pt, gray] (0,0) to (2,0) {};
	       \draw[line width=2.0pt, gray, dashed] (2.5,0) to (3,0) {};
	       \node at (0.75,0.7) {$\cdots$};
	        \draw[->] (.2,0) to (.2,1.5) {};
	       \draw[->] (1.2,0) to (1.2,1.5) {};
	       \draw[->] (1.8,1.5) to [out=-90, in=-90, looseness=5](1.5,1.5) {};
	       \node at (2.7,0.7) {$\cdots$};
	     \end{tikzpicture}
	     +
	     \begin{tikzpicture}[baseline=(current bounding box.center),scale=0.8,line width=1.0pt,>=latex]
	       \draw[line width=2.0pt, gray] (0,0) to (2,0) {};
	       \draw[line width=2.0pt, gray, dashed] (2.5,0) to (3,0) {};
	       \node at (0.75,0.7) {$\cdots$};
	        \draw[->] (.2,0) to (.2,1.5) {};
	       \draw[->] (1.5,0) to (1.5,1.5) {};
	       \draw[->] (1.8,1.5) to [out=-90, in=-90, looseness=4](1.2,1.5) {};
	       \node at (2.7,0.7) {$\cdots$};
	     \end{tikzpicture}
	     +...+
	     \begin{tikzpicture}[baseline=(current bounding box.center),scale=0.8,line width=1.0pt,>=latex]
	       \draw[line width=2.0pt, gray] (0,0) to (2,0) {};
	       \draw[line width=2.0pt, gray, dashed] (2.5,0) to (3,0) {};
	       \node at (0.75,0.7) {$\cdots$};
	       \draw[->] (1.2,0) to (1.2,1.5) {};
	       \draw[->] (1.5,0) to (1.5,1.5) {};
	       \draw[->] (1.8,1.5) to [out=-90, in=-90, looseness=1.4](.2,1.5) {};
	       \node at (2.7,0.7) {$\cdots$};
	     \end{tikzpicture}
	     \,.
\end{eqnarray}
These equations can be proven by inserting indices/rapidities, repeatedly applying \eqref{eq:zz-mixed} and using $z_\alpha(\te)\Omega=0$ as well as $\overline{S^{\alpha\beta}_{\gamma\delta}(\te)}=S^{\gamma\delta}_{\alpha\beta}(-\te)$ (Def.~\ref{smatrix}~$i),ii)$) .

\medskip

We now use this graphical notation to define {\em contracted matrix elements} of $A$. Given two integers $0\leq k\leq n$, a {\em contraction of type $(n,k)$} is a diagram as in \eqref{eq:def-A-matrix-element}, with $k$ external vertices on the left and $n-k$ external vertices on the right, and an arbitrary number of contractions (pairings) between external vertices on the left and right hand side. A contraction between two external vertices, say $l$ on the left and $r$ on the right, is represented by a line from $r$ to $l$.

The set of all contractions of type $(n,k)$ will be denoted $\CC_{n,k}$; it contains contractions $C$ of {\em length} $|C|$ (defined as the number of pairs in $C$) up to $|C|\leq\min\{k,n-k\}$, and we also allow for the empty contraction $C=\{\,\}$ with $|C|=0$.

Each contraction $C$ corresponds to a tensor-valued distribution $\langle A\rangle_C^{\alpha_1...\alpha_n}(\te_1,...,\te_n)$ on $\Rl^n$, which is defined by attaching rapidities $\te_1,...,\te_n$ and indices $\alpha_1,...,\alpha_n$ to the external vertices of the diagram, ordered from left to right, taking the product of all $S$- and $\delta$-factors appearing in the diagram and the matrix element of $A$, and summing over all internal lines, i.e. all lines that are not connected to one of the $n$ external vertices. Symbolically, this means
\begin{align}
	\langle A\rangle_C=\sum_{\text{internal}\atop{\rm lines}}\prod_{\rm crossings\atop{}}S\prod_{ {\rm contracted}\atop{\rm lines}}\delta\,.
\end{align}
For example (with $\te_{ab}:=\te_a-\te_b$),
\begin{equation*}
      \begin{tikzpicture}[baseline=(current bounding box.center),scale=1.0,line width=1.0pt,>=latex]
	       \draw[line width=2.0pt, gray] (0,0) to (2,0) {};
	       \draw[line width=2.0pt, gray] (2.5,0) to (4.5,0) {};
	       \draw[->] (1,0) to (1,1.5) {};
	       \node[above] at (0.1,1.5) {${\alpha_1}\atop{[\te_1]}$};
	       \node[above] at (1,1.5) {${\alpha_2}\atop{[\te_2]}$};
	       \node[above] at (1.8,1.5) {${\alpha_3}\atop{[\te_3]}$};
	       
	       \draw[->] (2.9,1.5) to [out=-90,in=-90,looseness=2] (1.6,1.5) {};
	       \draw[->] (4.3,1.5) to [out=-90,in=-90,looseness=1] (0.2,1.5) {};
	       \draw[<-] (3.5,0) to (3.5,1.5) {};
	       \node[above] at (2.8,1.5) {${\alpha_4}\atop{[\te_4]}$};
	       \node[above] at (3.5,1.5) {${\alpha_5}\atop{[\te_5]}$};
	       \node[above] at (4.3,1.5) {${\alpha_6}\atop{[\te_6]}$};
	 \end{tikzpicture}
	 =
	 \delta(\te_{16})\delta(\te_{34})\delta^{\alpha_3}_{\alpha_4}\sum_{\beta,\gamma,\varepsilon}S^{\alpha_5\alpha_6}_{\beta\gamma}(\te_{65})S^{\beta\varepsilon}_{\alpha_2\alpha_1}(\te_{26})\,\langle\zd_\gamma(\te_5)\Omega,A\zd_\varepsilon(\te_2)\Omega\rangle\,.
\end{equation*}
For an unambiguous definition of $\langle A\rangle_C$ we exclude self intersecting lines \begin{tikzpicture}[scale=0.20,line width=0.7pt,>=latex]
	  \draw[->] (0,0) to [out=0, in=0,looseness=2] (1,-1) to [out=-180, in=-180,looseness=2] (3,0){};
     \end{tikzpicture}  (``type I Reidemeister moves''). Then the diagram of a contraction $C$ is uniquely defined by the pairings in $C$ up to the Reidemeister moves II and III \cite{Kauffman:1993}:
\begin{center}
\begin{minipage}{0.4\textwidth}
     \begin{tikzpicture}[scale=1.0,line width=1.0pt,>=latex]
	  \draw[<-<] (0,0) to [out=-50, in=-120] (2,0) {};
	  \draw[<-<] (0,-0.7) to [out=50, in=120] (2,-0.7) {};
	  \node[left] at (0,0) {$\alpha$};
	  \node[left] at (0,-0.7) {$\beta$};
	  \node[right] at (2,0) {$\gamma$};
	  \node[right] at (2,-0.7) {$\delta$};
	  \node at (3,-0.35) {$\longleftrightarrow$};
	   \draw[<-<] (3.8,0) to [out=-20, in=-160] (5.8,0) {};
	  \draw[<-<] (3.8,-0.7) to [out=20, in=160] (5.8,-0.7) {};
	  \node[left] at (3.8,0) {$\alpha$};
	  \node[left] at (3.8,-0.7) {$\beta$};
	  \node[right] at (5.8,0) {$\gamma$};
	  \node[right] at (5.8,-0.7) {$\delta$};
     \end{tikzpicture}
     \end{minipage}
     \qquad\qquad
     \begin{minipage}{0.4\textwidth}
               \begin{tikzpicture}[scale=1.0,line width=1.0pt,>=latex]
	  \draw[<-<] (0,0) to [out=-50, in=170] (2,-0.7) {};
	  \draw[<-<] (0,-0.7) to [out=50, in=-160] (2,0.4) {};
	  \draw[<-<] (1.6,-1.1) to [out=150, in=-90] (0.8,0.6) {};
	  \node[left] at (0,0) {$\alpha$};
	  \node[left] at (0,-0.7) {$\beta$};
	  \node[right] at (2,0.4) {$\gamma$};
	  \node[right] at (2,-0.7) {$\delta$};
	  \node[right] at (0.8,0.6) {$\eta$};
	  \node[below] at (1.6,-1.1) {$\xi$};
	  \node at (3,-0.35) {$\longleftrightarrow$};
	  \draw[<-<] (3.8,0) to [out=-50, in=170] (5.8,-0.7) {};
	  \draw[<-<] (3.8,-0.7) to [out=-50, in=-90] (5.8,0.4) {};
	  \draw[<-<] (5.4,-1.1) to [out=150, in=-90] (4.6,0.6) {};
	  \node[left] at (3.8,0) {$\alpha$};
	  \node[left] at (3.8,-0.7) {$\beta$};
	  \node[right] at (5.8,0.4) {$\gamma$};
	  \node[right] at (5.8,-0.7) {$\delta$};
	  \node[right] at (4.6,0.6) {$\eta$};
	  \node[below] at (5.4,-1.1) {$\xi$};
     \end{tikzpicture}
          \end{minipage}
\end{center}
But as a consequence of Hermitian analyticity and the Yang-Baxter equation (Def.~\ref{S-matrixDefinition}~$ii),iii)$), in both cases the left and right partial diagram give the same contribution to $\langle A\rangle_C$, as follows by straightforward calculation.

\medskip

With these conventions, we have defined $\langle A\rangle_C$ for each contraction $C\in\CC_{n,k}$ and each $A\in\B(\Hil)$, and now comment on the analytic properties of this distribution. To begin with, the matrix elements \eqref{eq:def-A-matrix-element} are tempered (vector-valued) distributions because of the particle number bounds \eqref{numberBounds} and the boundedness of $A$. Within $\langle A\rangle_C$, they only depend on those $\te$-variables that are not contracted, whereas the delta distributions depend only on the contracted variables. Hence their product, as it appears in $\langle A\rangle_C$, is well-defined. Also the product of these distributions with the $S$-factors is well defined: If we consider regular $S\in\SF_0$ (Def.~\ref{regularS}), the analyticity and boundedness of $S$ in a strip containing the real line implies that $S$ is smooth and has bounded derivatives on $\Rl$ via Cauchy's integral formula. Thus we conclude that $\langle A\rangle_C$ is well-defined as a tempered distribution 
on $\Rl^n$, taking values in $\K^{\otimes n}$, for any contraction $C\in\CC_{n,k}$.

The {\em completely contracted matrix elements} (of type $(n,k)$) of $A$ are defined as
\begin{align}\label{eq:def-Acon}
     \langle A\rangle^{\rm con}_{n,k}
     :=
     \sum_{C\in\CC_{n,k}}
     (-1)^{|C|}
     \langle A\rangle_C
     \,,
\end{align}
they are our main object of interest in the following.

\medskip

To explain the relation between the completely contracted matrix elements $\langle A\rangle^{\rm con}_{n,k}$ of some $A\in\F(W_R)$ and its wavefunctions $(A\Omega)_n$, it is instructive to consider the two special cases $k=0$ and $k=n$. As either $k=0$ (no left external vertices) or $n-k=0$ (no right external vertices), in both cases, $\CC_{n,0}=\CC_{n,n}=\{\{\,\}\}$ contains only the empty contraction $C=\{\,\}$. By observing the orderings in \eqref{eq:def-A-matrix-element}, one finds 
\begin{align}\label{zusammenhang}
     \left(\langle A\rangle^{\text{con}}_{n,0}\right)^{\boldsymbol{\alpha}}(\boldsymbol{\theta})
     &=
     \langle z_{\alpha_1
     }^\dagger(\theta_1)\cdots z_{\alpha_n
     }^\dagger(\theta_n)\Omega,A\Omega\rangle=\sqrt{n!}\,(A\Omega)_n^{\boldsymbol{\alpha}}(\bte)
     \\
     \left(\langle A\rangle^{\text{con}}_{n,n}\right)^{\boldsymbol{\alpha}}(\boldsymbol{\theta})
     &=
     \langle \Omega,Az_{{\alpha}_n
     }^\dagger(\theta_n)\cdots z_{{\alpha}_1
     }^\dagger(\theta_1)\Omega\rangle
     \\
     &=
     \sqrt{n!}\,\overline{(A^*\Omega)_n^{{\alpha_n}\dots{\alpha_1}}(\theta_n,\dots,\theta_1)}=\sqrt{n!}\,(JA^*\Omega)_n^{\overline{\alpha_n}...\overline{\alpha_1}}(\boldsymbol{\theta})
     \nonumber
     \,.
\end{align}

Because of this close connection between the completely contracted matrix elements of $A$ and the wavefunctions $(A\Omega)_n$, analyticity and boundedness properties of the former will imply corresponding properties of the latter.

Our next aim is to prove the following proposition on analytic continuations of the completely contracted matrix elements.

\begin{proposition}\label{Lemma}
	Let $0\leq k< n$, $S\in\mathcal{S}_0$ and $A\in\mathcal{F}(W_R)$. 
	\begin{enumerate}
		\item The distribution $(\langle A\rangle^{\rm{con}}_{n,k})^\balpha(\bte)$ has an analytic continuation in the variable $\theta_{k+1}$ to the strip $\Strip(-\pi,0)$, and its boundary value at $\im\theta_{k+1}=-\pi$ is
	       \begin{equation}
			 \left(\langle A\rangle^{\rm{con}}_{n,k}\right)^{\alpha_1...\alpha_n}(\theta_1,...\theta_{k+1}-i\pi,...\theta_n)
			 =
			 \left(\langle A\rangle^{\rm{con}}_{n,k+1}\right)^{\alpha_1...\overline{\alpha_{k+1}}...\alpha_n}(\theta_1,...\theta_{k+1},...\theta_n).
	       \end{equation}
	       \item Let $f_1,\dots,f_n\in\mathscr{S}(\mathbb{R})\otimes\mathcal{K}$ and $0\leq\lambda\leq\pi$. Then here holds the bound
	       \begin{equation}\label{eq:bound-on-continuation}
		    \left|\int d^n\boldsymbol{\theta}
		    \Big(\bigotimes_{j=1}^{n}f_j(\theta_j),\langle A\rangle^{\rm{con}}_{n,k}(\theta_1,\dots,\theta_{k+1}-i\lambda,\dots,\theta_n)
		    \Big)
		    \right|
		    \leq 2^{n}\sqrt{n!}\,\|A\|\cdot\prod_{j=1}^{n}\|f_j\|_2\,.
		\end{equation}
	\end{enumerate}
\end{proposition}
% \begin{proof}
     \noindent{\em Proof.} $i)$ We begin the proof by considering special terms in the sum \eqref{eq:def-Acon}, namely those which
     \begin{wrapfigure}[7]{r}{0.3\textwidth}
       \begin{tikzpicture}[baseline=(current bounding box.center),scale=0.9,line width=1.0pt,>=latex]
	       \draw[line width=2.0pt, gray] (0,0) to (2,0) {};
	       \draw[line width=2.0pt, gray] (2.5,0) to (4.5,0) {};
	       \draw[->] (0.4,0) to (0.4,1.5) {};
	       \draw[->] (1.6,0) to (1.6,1.5) {};
	       \node at (1,0.7) {$\cdots$};
	       \draw[<-] (2.9,0) to (2.9,1.5) {};
	       \draw[<-] (4.1,0) to (4.1,1.5) {};
	       \node at (3.5,0.7) {$\cdots$};
	        \node[above] at (2.9,1.5) {${\alpha}\atop{[\te]}$};
	 \end{tikzpicture}
	 $$\langle\Omega,\dots z_\alpha(\te)A\dots\Omega\rangle$$
     \end{wrapfigure}
      correspond to contraction diagrams $C\in\CC_{n,k}$ in which the leftmost external vertex in the right hand side of the diagram, i.e. line number $k+1$ in the full diagram, is {\em not} contracted, and we denote the subset of these contractions by $\CCh_{n,k}\subset\CC_{n,k}$. Throughout this proof, we will label this line with index $\alpha$ and rapidity $\te$, so that it corresponds to the annihilation operator $z_\alpha(\te)$ to the left of $A$, cf. the first line in \eqref{eq:def-A-matrix-element}. 
     
     If we switch the order of $z_{\alpha}(\te)$ and $A$ in $\langle A\rangle_C$, the incoming line labeled $(\alpha,\te)$ is switched to the rightmost position of the left half of the diagram, and we may use the mixed Zamolodchikov relation \eqref{eq:zzd-graphical-2} to see that $\langle A\rangle_C$ with $z_{\alpha}(\te)$ and $A$ interchanged coincides with the sum $\sum_{C'\in P(C)}\langle A\rangle_{C'}$, where $C'$ runs over the subset $P(C)\subset\CC_{n,k}$ of all contractions that differ from $C$ precisely by adding an additional contraction between line $k+1$ and an uncontracted line on the left (in particular, $|C'|=|C|+1$). Therefore, $\langle \widehat{A}\rangle_C:=\langle A\rangle_C-\sum_{C'\in P(C)}\langle A\rangle_{C'}$, $C\in\CCh_{n,k}$, can be written as the (contracted) expectation value of the commutator $[z_{\alpha}(\te),A]$. 
     
     We claim that $\langle \widehat{A}\rangle_C^\alpha(\te)$ (we suppress all dependence on the other rapidities and indices here) has an analytic continuation in $\te\in S(-\pi,0)$. Indeed, by Lemma~\ref{basicLemma}, expectation values of $[z_{\alpha}(\te),A]$ analytically continue to $S(-\pi,0)$, and the boundary value at $\im\te=-\pi$ is given by the same (contracted) expectation value, now taken of $[A,\zd_{\overline{\alpha}}(\te)]$.
     
     \begin{wrapfigure}{r}{0.3\textwidth}
	  \begin{equation*}
          \begin{tikzpicture}[baseline=(current bounding box.center),scale=1.0,line width=1.0pt,>=latex]
	       \draw[line width=2.0pt, gray] (0,0) to (1,0) {};
	       \draw[line width=2.0pt, gray] (1.5,0) to (2.5,0) {};
	       \node at (0.5,0.4) {$\cdots$};
	       \draw[<-] (1.7,0) to (1.7,1.5) {};
	       \draw[->] (2.2,1.5) to [out=-90, in=-90,looseness=2] (.5,1.5) {};
	       \node[above] at (0.5,1.5) {\small$\mu$};
	       \node[above] at (1.7,1.5) {${\alpha}\atop{[\te]}$};
	        \node[above] at (2.2,1.5) {${\nu}\atop{[\te']}$};
	       \node[above] at (1.9,0.1) {\small$\lambda$};
	 \end{tikzpicture}
	 = S^{\alpha\nu}_{\mu\lambda}(\te'-\te)
	 \end{equation*}
     \end{wrapfigure}
     But to discuss the full $\te$-dependence of $\langle\widehat{A}\rangle_C$, we also have to consider the dependence of the $S$-factors in $\langle A\rangle_C$ on $\te$ (the $\delta$ distributions do not depend on this variable because $\te$ is not contracted in $C$). Those $S$-factors that depend on $\te$ arise from contractions crossing the $\te$-line as shown on the right. Therefore $\te$ always appears with a minus sign in the argument of $S$, and the $S$-factors are analytic in $\te\in S(-\pi,0)$ as well.
     
     To compute the boundary value at $\im\te=-\pi$, we have to take into account both, the change from $[z_{\alpha}(\te),A]$ to $[A,\zd_{\overline{\alpha}}(\te)]$, and the crossing symmetry of $S$ (Def.~\ref{S-matrixDefinition}~$iv)$), namely $S^{\alpha\beta}_{\gamma\delta}(\te+i\pi)=S^{\overline{\gamma}\alpha}_{\delta\overline{\beta}}(-\te)$. Expanding the commutator $[A,\zd_{\overline{\alpha}}(\te)]=A\,\zd_{\overline{\alpha}}(\te)-\zd_{\overline{\alpha}}(\te)\,A$, the term with the creation operator to the right of $A$ corresponds to a contraction diagram of type $(n,k+1)$, where the incoming $(\alpha,\te)$-line has been transformed to an outgoing $(\overline{\alpha},\te)$-line in the rightmost position on the left hand side of the diagram. By observing the crossing relation
     \begin{equation*}
          \begin{tikzpicture}[baseline=(current bounding box.center),scale=1.0,line width=1.0pt,>=latex]
	       \draw[line width=2.0pt, gray] (0,0) to (1,0) {};
	       \draw[line width=2.0pt, gray] (1.5,0) to (2.5,0) {};
	       \node at (0.5,0.4) {$\cdots$};
	       \draw[<-] (1.7,0) to (1.7,1.5) {};
	       \draw[->] (2.2,1.5) to [out=-90, in=-90,looseness=1.5] (.5,1.5) {};
	       \node[above] at (0.5,1.5) {\small$\mu$};
	       \node[above] at (1.7,1.5) {${\alpha}\atop{[\te]}$};
	        \node[above] at (2.2,1.5) {${\nu}\atop{[\te']}$};
	       \node[above] at (1.9,0) {\small$\lambda$};
	 \end{tikzpicture}
	 \xrightarrow{\quad\te\to\te-i\pi\quad}
	 \begin{tikzpicture}[baseline=(current bounding box.center),scale=1.0,line width=1.0pt,>=latex]
	       \draw[line width=2.0pt, gray] (0,0) to (1,0) {};
	       \draw[line width=2.0pt, gray] (1.5,0) to (2.5,0) {};
	       \node at (0.5,0.4) {$\cdots$};
	       \draw[->] (0.8,0) to (0.8,1.5) {};
	       \draw[->] (2.2,1.5) to [out=-90, in=-90,looseness=1.5] (.5,1.5) {};
	       \node[above] at (0.42,1.5) {\small$\mu$};
	       \node[above] at (0.8,1.5) {${\overline{\alpha}}\atop{[\te]}$};
	        \node[above] at (2.2,1.5) {${\nu}\atop{[\te']}$};
	       \node[above] at (0.95,0) {\small$\overline\lambda$};
	 \end{tikzpicture}
	 \quad,
	 \end{equation*}
	 one sees that this term is precisely $\langle A\rangle_{\Cti}^{\overline{\alpha}}(\te)$, where $\Cti\in\CC_{n,k+1}$ consists of the same pairings as the original contraction $C\in\CC_{n,k}$, but the $(\alpha,\te)$-line has been ``crossed'' from the right to the left hand side of the diagram (in particular $|\Cti|=|C|$). The other term in the commutator, with the creation operator to the left of $A$, corresponds to switching the incoming $(\alpha,\te)$-line in the original diagram of $C$ to an outgoing $(\overline{\alpha},\te)$-line in the same position. Using the mixed Zamolodchikov relation \eqref{eq:zzd-graphical-1}, it follows that 
	  \begin{align}\label{eq:partial-crossing}
	       \langle \widehat{A}\rangle_C^\alpha(\te-i\pi)
	       =
	       \langle A\rangle_{\Cti}^{\overline{\alpha}}(\te)
	       -
	       \sum_{\Cti'\in\tilde{P}(\Cti)}	  \langle A\rangle_{\Cti'}^{\overline{\alpha}}(\te)
	       \,,
	  \end{align}
	  where $\tilde{P}(\Cti)\subset\CC_{n,k+1}$ denotes the set of all contractions which differ from $\Cti$ precisely by contracting $k+1$, the rightmost line of the left half of the diagram, with an uncontracted line on the right (in particular $|\Cti'|=|\Cti|$ for $\Cti'\in\tilde{P}(\Cti)$).
     
	  To conclude the proof, we note that since any contraction $C\in\CC_{n,k}$ either contracts $k+1$ or not, we have the disjoint unions $\CC_{n,k}=\bigsqcup_C\{C\sqcup P(C)\}$ and $\CC_{n,k+1}=\bigsqcup_{\Cti}\{C\sqcup \tilde{P}(\Cti)\}$, where $C$ runs over $\CCh_{n,k}\subset\CC_{n,k}$, and $\Cti$ runs over $\CCh_{n,k+1}\subset\CC_{n,k+1}$, the set of all contractions not contracting line $k+1$ as a line on the left. Taking into account $|C'|=|C|+1$ for $C'\in P(C)$, the completely contracted matrix elements may be written as 
	  \begin{align}
	       \langle A\rangle_{n,k}\con
	       &=
	       \sum_{C\in\CCh_{n,k}}
		    \bigg\{
			 (-1)^{|C|}\langle A\rangle_C+\sum_{C'\in P(C)}(-1)^{|C'|}\langle A\rangle_{C'}
		    \bigg\}
	       =
	       \sum_{C\in\CCh_{n,k}} (-1)^{|C|}\langle\widehat{A}\rangle_C
	       \,,
	  \end{align}
	  which implies that $(\langle A\rangle_{n,k}\con)^\alpha(\te)$ has an analytic continuation to $S(-\pi,0)$. At $\im\te=-\pi$, \eqref{eq:partial-crossing} gives
	  \begin{align*}
	       [\langle A\rangle_{n,k}\con]^\alpha(\te-i\pi)
	       &=
	       \sum_{C\in\CCh_{n,k}} (-1)^{|C|}\langle\widehat{A}\rangle_C^\alpha(\te-i\pi)
	       \\
	       &=
	       \sum_{\Cti\in\CCh_{n,k+1}} (-1)^{|\Cti|}
		    \bigg\{
			 \langle A\rangle_{\Cti}^{\overline{\alpha}}(\te)
			 -
			 \sum_{\Cti'\in\tilde{P}(\Cti)}  \langle A\rangle_{\Cti'}^{\overline{\alpha}}(\te)
		    \bigg\}
	       \\
	       &=
	       \sum_{C\in\CC_{n,k+1}}(-1)^{|C|}\langle A\rangle_C^{\overline{\alpha}}(\te)
	       \\
	       &=
	       (\langle A\rangle_{n,k+1}\con)^{\overline{\alpha}}(\te)\,.
	  \end{align*}
	  This concludes the proof of part $i)$.
     
     \bigskip
     $ii)$ Let $0\leq k\leq n$, $C\in\CC_{n,k}$ an arbitrary contraction, and $f_1,...,f_n\in \Ss^2(\mathbb{R})\otimes\mathcal{K}$ be testfunctions. To estimate $\langle A\rangle_C(f_1\otimes ...\otimes f_n)$, we split the integration variables into three parts: First $\bte\in\Rl^{|C|}$, those variables in $\te_1,...,\te_k$ that are contracted by $C$, second $\bte'\in\Rl^{n-k-|C|}$, those variables in $\te_{k+1},...,\te_n$ that are not contracted, and third $\bte''\in\Rl^{k-|C|}$, those variables in $\te_1,...,\te_k$ that are not contracted. An analogous split is applied to the sum over indices, resulting in indices $\bgamma,\balpha,\bbeta$ with $|\bgamma|=|C|$, $|\balpha|=n-k-|C|$, and $|\bbeta|=k-|C|$.
     
     Carrying out all integrations over contraction delta functions, we find
     \begin{align*}
	  \langle A\rangle_C(f_1\otimes ...\otimes f_n)
	  =
	  \int d\bte&\sum_\bgamma
	  \int d\bte'\int d\bte''\sum_{\balpha,\bbeta}
	  {F}_{\bte,\bgamma}^{\balpha}(\bte')\,{G}^{\bbeta}_{\bte,\bgamma}(\bte'')
	  \\
	  &\times\langle\zd_{\alpha_1}(\te_1')\cdots\zd_{\alpha_{n-k-|C|}}(\te_{n-k-|C|}')\Omega,\,
		    A\zd_{\beta_1}(\te_1'')&\cdots\zd_{\beta_{k-|C|}}(\te_{k-|C|}'')\Omega\rangle
		    \,.
     \end{align*}
     Here $G$ results from $f_1\otimes ... \otimes f_k$ by application of $S$-factors and reordering of indices, and analogously, $F$ results from $f_{k+1}\otimes...\otimes f_n$; the separation of variables expressed in $F,G$ is possible because no $S$-factors appear that depend on uncontracted rapidities on the left {\em and} right of the diagram of $C$. 
     
     We have written the rapidities $\bte$ and the $\bgamma$-indices at the bottom to indicate that we view $F, G$ as testfunctions in $n-k-|C|$ and $k-|C|$ rapidities/indices, that depend on $\bte,\bgamma$ as parameters. By application of the Cauchy-Schwarz inequality, the particle number bounds \eqref{numberBounds} and the boundedness of $A$, we get
     \begin{align}\label{estimate1}
	  \left| \langle A\rangle_C(f_1\otimes ...\otimes f_n)\right|
	  &\leq
	  \int d\bte\sum_\bgamma
	  \sqrt{(n-k-|C|)!}\sqrt{(k-|C|)!}\,\|{F}_{\bte,\bgamma}\|_2\,\|{G}_{\bte,\bgamma}\|_2\,\|A\|
	  \,,
	  \nonumber
     \end{align}
     where $\|\cdot\|_2$ denotes the norms on $L^2(\Rl^a)\otimes\K^{\otimes a}$, $a=k,n-k$. We now exploit the fact that the $S$-factors are unitary, and the underlying tensor structure of $F,G$. This allows us to proceed to
     \begin{align*}
	  \left| \langle A\rangle_C(f_1\otimes ...\otimes f_n)\right|
	  &\leq
	  \sqrt{(n-k-|C|)!}\sqrt{(k-|C|)!}\int d\bte\sum_\bgamma
	  \prod_{j=1}^{|C|}\left(f_{l_j}^{\gamma_j}(\te_j)f_{r_j}^{\gamma_j}(\te_j)\right)
	  \cdot{\prod_i}'\|f_i\|_2
	  \,\|A\|
	  \,,
	  \nonumber
     \end{align*}
     where $(l_j,r_j)$ denote the pairs that are contracted by $C$, and the second product $\Pi_i'$ runs over all uncontracted lines. Now another application of Cauchy-Schwarz yields
      \begin{align}
	  \left| \langle A\rangle_C(f_1\otimes ...\otimes f_n)\right|
	  &\leq
	  \sqrt{(n-k-|C|)!}\sqrt{(k-|C|)!}\prod_{i=1}^n\|f_i\|_2
	  \,\|A\|
	  \,,
	  \nonumber
     \end{align}
     and by using the estimate $\sum_{C\in\CC_{n,k}}\sqrt{(n-k-|C|)!}\sqrt{(k-|C|)!}\leq 2^n\sqrt{n!}$ \cite{L08}, we arrive at the claimed inequality \eqref{eq:bound-on-continuation} for the boundary values at $\lambda=0$ and $\lambda=\pi$. 
     
     The bound \eqref{eq:bound-on-continuation} implies, in particular, that upon integrating all variables but $\te_{k+1}$, $$h^{\alpha_{k+1}}(\theta_{k+1}):=\int d\te_j\,\left(\langle A\rangle^{\text{con}}_{n,k}\right)^{\boldsymbol{\alpha}}(\boldsymbol{\theta})\prod_{\stackrel{j=1}{j\neq k+1}}^{n}f_j^{\alpha_j}(\theta_j)$$ is square-integrable. From here, we can deduce \eqref{eq:bound-on-continuation} also for imaginary part $0<\lambda<\pi$ as in the scalar case: One uses the boundedness of $S$ on $\overline{S(0,\pi)}$ and the bounds found in Lemma \ref{basicLemma} to see that also $\theta_{k+1}\mapsto h_{-\lambda}^{\alpha_{k+1}}(\theta_{k+1})=h^{\alpha_{k+1}}(\theta_{k+1}-i\lambda)$ is in $L^2(\mathbb{R})$ for any $0\leq \lambda\leq \pi$ and, by the first part, also analytic on $S(-\pi,0)$. By application of the three lines theorem, it follows that \eqref{eq:bound-on-continuation} also holds for $|\int d\theta_{k+1}\,h^{\alpha_{k+1}}(\theta_{k+1}-i\lambda)f_{k+1}^{\alpha_{k+1}}(\theta_{k+1})|$. \hfill$\square$

\subsection{Property (H) holds for regular $\boldsymbol{S\in\SF_0}$}

The modular group of $(\F(W_R),\Omega)$ acts on the functions $(A\Omega)_n$, $A\in\F(W_R)$, according to Prop.~\ref{BisognWich}~$i)$ as $(\Delta^{it}A\Omega)^\balpha_n(\bte)=(A\Omega)_n^\balpha(\bte+4t\,\bla_{\pi/2})$. Since $A\Omega\in\dom\Delta^{1/2}$, this implies analyticity of $A\Omega_n$ in the ``center of mass rapidity'' $(\te_1+...+\te_n)/n$ in the strip $\Strip(-\pi,0)$, with boundary value at the lower boundary being $(JA^*\Omega)_n^\balpha$. In comparison to this general fact, we will now argue that Prop.~\ref{Lemma} implies much stronger analyticity properties of $(A\Omega)_n$, involving $n$ complex variables.

Starting at $\langle A\rangle_{n,0}\con=\sqrt{n!}(A\Omega)_n$ \eqref{zusammenhang}, we see that $\sqrt{n!}(A\Omega)_n^\balpha(\bte)$ has an analytic continuation to $\te_1\in\Strip(-\pi,0)$, with boundary value at $\im\te_1=-\pi$ given by $(\langle A\rangle_{n,1}\con)^{\overline{\alpha_1}\alpha_2...\alpha_n}(\bte)$. This distribution has, in turn, an analytic continuation in $\te_2\in\Strip(-\pi,0)$, with boundary value at $\im\te_2=-\pi$ given by $(\langle A\rangle_{n,2}\con)^{\overline{\alpha_1}\,\overline{\alpha_2}...\alpha_n}(\bte)$, etc. After $n$ successive steps of analytic continuation we arrive at $(\langle A\rangle_{n,n}\con)^{\overline{\alpha_1}...\overline{\alpha_n}}(\bte)=\sqrt{n!}\,(JA^*\Omega)_n^{\boldsymbol{\alpha}}$. 

To state the ensuing properties precisely, we define the $n$-dimensional tube
\begin{equation}\label{eq:Tun}
	\mathcal{T}_n:=\mathbb{R}^n-i\mathcal{G}_n,\qquad \mathcal{G}_n:=\{\boldsymbol{\lambda}\in\mathbb{R}^n:\pi>\lambda_1>\lambda_2>\dots>\lambda_n>0\}
	\,.
\end{equation}

\begin{corollary}{\bf (to Prop.~\ref{Lemma}~$\boldsymbol{i)}$)} \label{Analyticity}
     Let $A\in\mathcal{F}(W_R)$. Then the function $(A\Omega)_n\in L^2(\Rl^n,\K^{\otimes n})$ is the distributional boundary value of an analytic function (denoted by the same symbol) $(A\Omega)_n:\Tu_n\to\K^{\otimes n}$.
\end{corollary}
\begin{proof}
     The base $-{\cal G}_n$ is the convex closure of the $n$ line segments from $-\pi\sum_{j=1}^{k}e_j$ to $-\pi\sum_{j=1}^{k+1}e_j$, $k=0,...,n-1$, where $\{e_j\}$ denotes the standard basis of $\Rl^n$. By Prop.~\ref{Lemma}, $(A\Omega)_n$ is analytic on each of these line segments. The statement then follows from the Malgrange Zerner Theorem \cite{eps66} and can be proven along the same lines as in the scalar case, see \cite[Proof of Cor. 5.2.6. a)]{DocL} for details.
\end{proof}

In property (H) (Def.~\ref{definition:property-h}), it is required that $(A\Omega)_n$, $A\in\F(W_R)$, is analytic in a tube containing $\bla_{\pi/2}=-(\tfrac{\pi}{2},\dots,\tfrac{\pi}{2})$ in the interior of its base. As the point $\bla_{\pi/2}$ only lies at the boundary of $-{\mathcal G}_n$ \eqref{eq:Tun}, an extension of the domain of analyticity is necessary. 

Such an extension is only possible if $S$ is regular in the sense of Def.~\ref{regularS}, which provides the main motivation for this definition. We will therefore from now on assume that $S$ is analytic and bounded on an enlarged strip, i.e. on $\Strip(-\kappa,\pi+\kappa)$ for some $0<\kappa<\frac{\pi}{2}$. To describe the resulting domain of analyticity of $(A\Omega)_n$, we define
\begin{equation}
     \mathcal{B}_n(\kappa)
     :=
     \{\boldsymbol{\lambda}\in\mathbb{R}^n:0<\lambda_1,\dots,\lambda_n<\pi,\quad-\kappa<\lambda_r-\lambda_l<\kappa,\quad 1\leq l<r\leq n\},
\end{equation}
as well as
\begin{equation}\label{l}
     \mathcal{C}_n(\kappa):=(-\tfrac{\kappa}{2},\tfrac{\kappa}{2})^{\times n}
     \,,\qquad
     \mathcal{T}_n(\kappa):=\mathbb{R}^n+i\left(\boldsymbol{\lambda}_{\pi/2}+\mathcal{C}_n(\kappa)\right)
     \subset
     -\B_n(\kappa)\,.
\end{equation}
\begin{figure}[h]
\begin{minipage}{3cm}
\begin{tikzpicture}[scale=1]
               \begin{scope}%[transparent]    %current
              \fill[gray!30!white] (0,0)--(-2.5,0)--(-2.5,-2.5) -- (0,0);
              
               \end{scope} 
        \begin{scope}%[transparent]    %current
                \draw[fill] (1.2,.2) node{Im$\,\zeta_1$};

\draw[fill] (.5,1) node{Im$\,\zeta_2$};

\draw[fill] (-2.5,.2) node{\footnotesize{$(-\pi,0)$}};
\draw[fill] (-2.8,-2.8) node{\footnotesize{$(-\pi,-\pi)$}};
\draw[fill] (.2,-2.8) node{\footnotesize{$(0,-\pi)$}};
\draw[fill] (.2,.2) node{\footnotesize{$0$}};
\fill[black] (-2.5,-2.5) circle (.3ex);
\fill[black] (-2.5,0) circle (.3ex);
\fill[black] (0,-2.5) circle (.3ex);
\fill[black] (0,0) circle (.3ex);
        \begin{scope}[->]
            \draw (0,-2.5) -- (0,1) node[anchor=north] {};
            \draw (-2.5,0) -- (1,0) node[anchor=east] {};

        \end{scope}
    \begin{scope}  
                \draw[fill] (-.9,-1.6) node{\scriptsize {$\boldsymbol{\lambda}_{\pi/2}$}} ;
       \draw[style=densely dashed] (-2.5,-2.5)--(0,0);
       \draw (-2.5,0) -- (-2.5,-2.5);
      \draw (-2.5,-2.5) -- (0,-2.5);
      \fill[black] (-1.25,-1.25) circle (.3ex);
        \end{scope}
        \end{scope}
\end{tikzpicture}
\end{minipage}\hspace{2cm}
\begin{minipage}{3cm}
\begin{tikzpicture}[scale=1]
               \begin{scope}%[transparent]    %current
              \fill[gray!30!white] (0,0)--(-.85,0)--(-2.5,-2.5+.85) -- (-2.5,-2.5)--(0,0);
              \fill[gray!50!white] (0,0)--(-2.5,-2.5)--(-1.65,-2.5) -- (0,-.85);
               \end{scope} 
        \begin{scope}%[transparent]    %current
                \draw[fill] (1.2,.2) node{Im$\,\zeta_1$};
\draw[fill] (.5,1) node{Im$\,\zeta_2$};
\draw[fill] (.3,-.45) node{$\kappa$};
\draw[fill] (-2.5,.2) node{\footnotesize{$(-\pi,0)$}};
\draw[fill] (-2.8,-2.8) node{\footnotesize{$(-\pi,-\pi)$}};
\draw[fill] (.2,-2.8) node{\footnotesize{$(0,-\pi)$}};
\draw[fill] (.2,.2) node{\footnotesize{$0$}};
\fill[black] (-2.5,-2.5) circle (.3ex);
\fill[black] (-2.5,0) circle (.3ex);
\fill[black] (0,-2.5) circle (.3ex);
\fill[black] (0,0) circle (.3ex);
        \begin{scope}[->]
            \draw (0,-2.5) -- (0,1) node[anchor=north] {};
            \draw (-2.5,0) -- (1,0) node[anchor=east] {};
            \draw (0.1,-.05)--(.1,-.85);
               \draw (.1,-.85)--(0.1,-.05);
        \end{scope}
    \begin{scope}  
    \draw[style=densely dashed] (-0.85,0) --(-2.5,-1.65);
       
    \draw[style=densely dashed] (0,-0.85) --(-1.65,-2.5);
       \draw (-2.5,0) -- (-2.5,-2.5);
      \draw (-2.5,-2.5) -- (0,-2.5);
      \fill[black] (-1.25,-1.25) circle (.3ex);
        \end{scope}
        \end{scope}
\end{tikzpicture}
\end{minipage}\hspace{2cm}
\begin{minipage}{3cm}
\begin{tikzpicture}[scale=1]
               \begin{scope}%[transparent]    %current
              
               \end{scope} 
        \begin{scope}%[transparent]    %current
                \draw[fill] (1.2,.2) node{Im$\,\zeta_1$};
\draw[style=dashed] (0,0) -- (-2.5,-2.5);
  \fill[gray!80!white] (-.85,-.85)--(-1.65,-.85)--(-1.65,-.85) -- (-1.65,-1.65)--(-.85,-1.65);
\draw[fill] (.5,1) node{Im$\,\zeta_2$};
\draw[fill] (.3,-.45) node{$\kappa$};
\draw[fill] (-2.5,.2) node{\footnotesize{$(-\pi,0)$}};
\draw[fill] (-2.8,-2.8) node{\footnotesize{$(-\pi,-\pi)$}};
\draw[fill] (.2,-2.8) node{\footnotesize{$(0,-\pi)$}};
\draw[fill] (.2,.2) node{\footnotesize{$0$}};
\fill[black] (-2.5,-2.5) circle (.3ex);
\fill[black] (-2.5,0) circle (.3ex);
\fill[black] (0,-2.5) circle (.3ex);
\fill[black] (0,0) circle (.3ex);
        \begin{scope}[->]
            \draw (0,-2.5) -- (0,1) node[anchor=north] {};
            \draw (-2.5,0) -- (1,0) node[anchor=east] {};
            \draw (0.1,-.05)--(.1,-.85);
               \draw (.1,-.85)--(0.1,-.05);
        \end{scope}
    \begin{scope}  
       \draw[style=densely dashed] (-.85,-.85)--(-1.65,-.85);
       \draw[style=densely dashed] (-1.65,-.85) -- (-1.65,-1.65);
       \draw[style=densely dashed] (-1.65,-1.65) --(-.85,-1.65);
  \draw[style=densely dashed] (-.85,-1.65) --(-.85,-.85);
    \draw[style=densely dashed] (0,-0.85) --(-1.65,-2.5);
       \draw (-2.5,0) -- (-2.5,-2.5);
      \draw (-2.5,-2.5) -- (0,-2.5);
      \fill[black] (-1.25,-1.25) circle (.3ex);
        \end{scope}
        \end{scope}
\end{tikzpicture}
\end{minipage}
\caption{The two-dimensional bases $-\mathcal{G}_2$ (left), $-\mathcal{B}_2(\kappa)$ (middle) and \mbox{$\boldsymbol{\lambda}_{\pi/2}+\mathcal{C}_2(\kappa)$} (right).}
\end{figure}

\begin{proposition}\label{enlagement}
     Let $S\in\mathcal{S}_0$ be analytic in $\Strip(-\kappa,\pi+\kappa)$ and $A\in\mathcal{F}(W_R)$. Then $(A\Omega)_n$ is analytic in the tube $\mathbb{R}^n-i\mathcal{B}_n(\kappa)$.
\end{proposition}
\begin{proof}
     We follow the proof of the scalar case \cite[Proof of Prop. 5.2.7. a)]{DocL}. To this end, let $\sigma\in\mathfrak{S}_n$. By Corollary \ref{Analyticity},
     \begin{equation*}
     \left(A\Omega\right)_n(\boldsymbol{\theta}^\sigma):=\left(A\Omega\right)_n(\theta_{\sigma(1)},\dots,\theta_{\sigma(n)})
     \end{equation*}
     is analytic in the permuted tube $\mathbb{R}^n-i\mathcal{G}_n^\sigma$, where
     \begin{equation*}
	  \mathcal{G}_n^\sigma:=\sigma\mathcal{G}_n=\{\boldsymbol{\lambda}\in\mathbb{R}^n:\pi>\lambda_{\sigma(1)}>\cdots>\lambda_{\sigma(n)}>0\}.
     \end{equation*}
     Since $\left(A\Omega\right)_n\in\mathscr{H}_n$, this vector is invariant under the representation $D_n$ of $\mathfrak{S}_n$ (\ref{reprpermugroup}), i.e.
     \begin{equation}\label{permversion}
	  \left(A\Omega\right)_n(\boldsymbol{\theta})=\left(D_n(\sigma)\left(A\Omega\right)_n\right)(\boldsymbol{\theta})=S_n^\sigma(\boldsymbol{\theta})\cdot\left(A\Omega\right)_n(\boldsymbol{\theta}^\sigma).
     \end{equation}
     The tensor $S_n^\sigma(\boldsymbol{\theta})$ consists of factors of the form $S(\theta_{r}-\theta_l)$, $1\leq l<r\leq n$, acting on various tensor factors of $\K^{\otimes n}$. As all such factors are analytic in the tube $\mathbb{R}^n+i\mathcal{B}_n'(\kappa)$ with
     \begin{equation*}
	  \mathcal{B}_n'(\kappa)
	  :=
	  \{\boldsymbol{\lambda}\in\mathbb{R}^n:-\kappa<\lambda_r-\lambda_l<\kappa,\quad\,1\leq l<r\leq n\}
	  \,,
     \end{equation*}
     also all $S_n^\sigma$, $\sigma\in\mathfrak{S}_n$, are analytic in $\mathbb{R}^n+i\mathcal{B}_n'(\kappa)$.
     
     Therefore, both sides of (\ref{permversion}) admit an analytic continuation: The left hand side to $\mathcal{T}_n=\mathbb{R}^n-i\mathcal{G}_n$ as before and the right hand side to the tube based on $\mathcal{B}_n'(\kappa)\cap (-\mathcal{G}_n^\sigma)$. Since convergence to the boundary in the sense of distributions gives the same value on $\mathbb{R}^n$, Epstein's generalization of the Edge of the Wedge Theorem \cite{eps60} can be applied, yielding that $(A\Omega)_n$ has an analytic continuation to the tube based on the convex closure of
     \begin{equation*}
	  \bigcup\limits_{\sigma\in\mathfrak{S}_n}\mathcal{B}_n'(\kappa)\cap(-\mathcal{G}_n^\sigma).
     \end{equation*}
     As the convex closure of $\bigcup_{\sigma}(-\mathcal{G}_n^\sigma)$ is the cube $(-\pi,0)^{\times n}$, $(A\Omega)_n$ is analytic in the tube based on $\mathcal{B}_n'(\kappa)\cap(-\pi,0)^{\times n}=-\mathcal{B}_n(\kappa)$, and the proof is complete.
\end{proof}
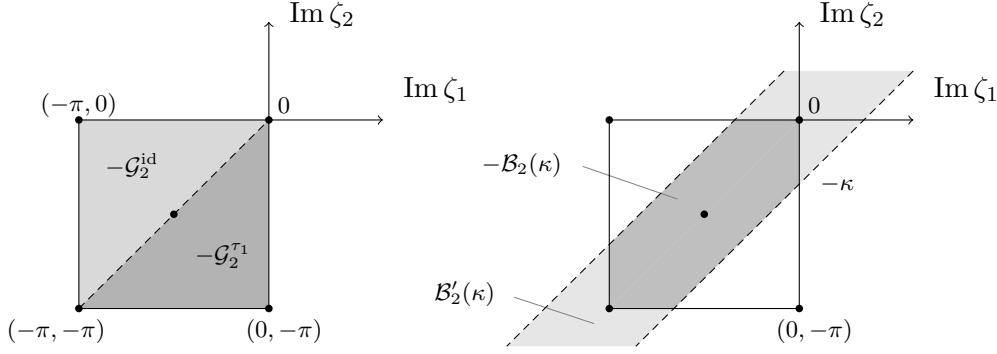
\begin{figure}[h]
\begin{minipage}{3cm}
\begin{tikzpicture}[scale=1]
               \begin{scope}%[transparent]    %current
              \fill[gray!30!white] (0,0)--(-2.5,0)--(-2.5,-2.5) -- (0,0);
   \fill[gray!60!white] (0,0)--(-2.5,-2.5)-- (0,-2.5) -- (0,0);
              
               \end{scope} 
        \begin{scope}%[transparent]    %current
                \draw[fill] (2.2,.4) node{Im$\,\zeta_1$};
\draw[fill] (.7,1.4) node{Im$\,\zeta_2$};

\draw[fill] (-2.5,.2) node{\footnotesize{$(-\pi,0)$}};
\draw[fill] (-2.8,-2.8) node{\footnotesize{$(-\pi,-\pi)$}};
\draw[fill] (.2,-2.8) node{\footnotesize{$(0,-\pi)$}};
\draw[fill] (.2,.2) node{\footnotesize{$0$}};
\fill[black] (-2.5,-2.5) circle (.3ex);
\fill[black] (-2.5,0) circle (.3ex);
\fill[black] (0,-2.5) circle (.3ex);
\fill[black] (0,0) circle (.3ex);
        \begin{scope}[->]
            \draw (0,-2.5) -- (0,1.3) node[anchor=north] {};
            \draw (-2.5,0) -- (1.5,0) node[anchor=east] {};
        \end{scope}
    \begin{scope}
       \draw[style=densely dashed] (-2.5,-2.5)--(0,0);
       \draw[fill] (-1.8,-.6) node{\footnotesize{$-\mathcal{G}^{\rm id}_2$}};
 \draw[fill] (-.6,-1.8) node{\footnotesize{$-\mathcal{G}^{\tau_1}_2$}};
       \draw (-2.5,0) -- (-2.5,-2.5);
      \draw (-2.5,-2.5) -- (0,-2.5);
      \fill[black] (-1.25,-1.25) circle (.3ex);
        \end{scope}
        \end{scope}
\end{tikzpicture}
\end{minipage}\hspace{2.5cm}
\begin{minipage}{3cm}
\begin{tikzpicture}[scale=1]
               \begin{scope}%[transparent]    %current
\fill[gray!20!white] (1.5,0.65)--(-.2,.65)--(-3.85,-3)--(-2.15,-3);

              \fill[gray!50!white] (0,0)--(-.85,0)--(-2.5,-2.5+.85) -- (-2.5,-2.5)--(0,0);
              \fill[gray!50!white] (0,0)--(-2.5,-2.5)--(-1.65,-2.5) -- (0,-.85);

               \end{scope} 
        \begin{scope}%[transparent]    %current
                \draw[fill] (2.2,.4) node{Im$\,\zeta_1$};
\draw[fill] (.7,1.4) node{Im$\,\zeta_2$};
\draw[fill] (.5,-.85) node{\footnotesize{$-\kappa$}};
\draw[fill] (-4.4,-2.3) node{\footnotesize{$\mathcal{B}_2'(\kappa)$}};
\draw[fill] (-3.65,-.6) node{\footnotesize{$-\mathcal{B}_2(\kappa)$}};
\draw[fill] (.2,-2.8) node{\footnotesize{$(0,-\pi)$}};
\draw[fill] (.2,.2) node{\footnotesize{$0$}};
\fill[black] (-2.5,-2.5) circle (.3ex);
\fill[black] (-2.5,0) circle (.3ex);
\fill[black] (0,-2.5) circle (.3ex);
\fill[black] (0,0) circle (.3ex);
        \begin{scope}[->]
            \draw (0,-2.5) -- (0,1.3) node[anchor=north] {};
            \draw (-2.5,0) -- (1.5,0) node[anchor=east] {};
        \end{scope}
    \begin{scope}
    \draw[help lines] (-2.85,-.615)--(-1.6,-1.05);
    \draw[help lines] (-3.75,-2.35)--(-2.7,-2.7);
    \draw[style=densely dashed] (1.5,0.65) --(-2.15,-3);
        \draw[style=densely dashed] (-.2,0.65) --(-3.85,-3);
       \draw (-2.5,0) -- (-2.5,-2.5);
      \draw (-2.5,-2.5) -- (0,-2.5);
      \fill[black] (-1.25,-1.25) circle (.3ex);
        \end{scope}
        \end{scope}
\end{tikzpicture}
\end{minipage}
\caption{Regions appearing in the proof of Proposition \ref{enlagement} for the case $n=2$.}
\label{fig1}
\end{figure}

\medskip
The bases $-\B_n(\kappa)$ are of the form required in Def.~\ref{definition:property-h}: They contain $\bla_{\pi/2}$ in their interior, with $n$-independent distance constants $c_n$ \eqref{eq:c_n} given by $c_n=\frac{m_\circ}{2}\cos\frac{\kappa}{2}$. This implies that the nuclear norm of the map $X_{-\B_n(\kappa),s}$ can be estimated as $\|X_{-\B_n(\kappa),s}\|_1\leq C(s,\kappa,m_\circ)^n$ \eqref{eq:nuclear-norm-X}. Furthermore, also the Hardy norm bounds required in Def.~\ref{definition:property-h} hold on these tubes: By application of the mean value property, the bounds of Prop.~\ref{Lemma}~$ii)$ can be converted into pointwise bounds on $(A\Omega)_n$ in the tube based on $-{\cal G}_n$. To estimate this function on the permuted tubes based on $-{\cal G}^\sigma_n$, $\sigma\in\mathfrak{S}_n$, one needs bounds on the tensor $S^\sigma_n(\bte)$ \eqref{permversion}. This bound is, however, only of the form $\sup_{\sigma\in{\mathfrak S}_n}\sup_{\bzeta\in\B_n(\kappa)}\|S_n^\sigma(\bzeta)\|\leq C^{n^2}$ with some $C>1$, in contrast to an incorrect estimate in \cite{L08}.

We therefore make an alternative choice of the tubes $\C_n$, which results in $\upsilon(s,n)\sim C^n$, and $\|X_{\C_n,s}\|_1\sim c_s\,n^n$.

\begin{proposition}\label{Corrhardystructure}
     Let $S\in\mathcal{S}_0$ be analytic and bounded by $\|S\|_\kappa<\infty$ on $\Strip(-\kappa,\pi+\kappa)$ for some $\kappa>0$. Then property (H) holds (Def.~\ref{definition:property-h}). The bases $\C_n$ can be chosen as
     \begin{align}
	  \C_n
	  =
	  \bla_{\pi/2}+\big(-\tfrac{\kappa}{2n},\tfrac{\kappa}{2n})^{\times n}
	  \,,
     \end{align}
     i.e., $c_n=\frac{\kappa}{2n}$, and the Hardy bound constants \eqref{eq:general-hardy-bound} can be estimated as
     \begin{align}\label{eq:upsilon-small-base}
		\upsilon(s,n)
		&\leq
		\max\left\{1,\;\sqrt{\frac{2}{\kappa}}\,\frac{e^{-sm_\circ \sin\kappa}}{(\pi m_\circ s\sin\kappa)^{1/4}}\,\|S\|_\kappa^n\,(\dim\K)^{n/2}\right\}
		\,.
     \end{align}
\end{proposition}
\begin{proof}
     We already know that for $A\in\F(W_R)$, the function $(A\Omega)_n$ has an analytic continuation to the tube based on $-\B_n(\kappa)$, which contains $\C_n$. Clearly $c_n=\|\bla_{\pi/2}-\partial\C_n\|_\infty=\frac{\kappa}{2n}$. But we have to give a proof of the Hardy properties of $(A\Omega)_n$, and the bound \eqref{eq:upsilon-small-base}.
    
     Similarly to the strategy used in \cite{Lech05}, we first work on the tube based on the simplex $\mathcal{C}_n^0$ defined as the convex closure of $0,-\kappa e_1,...,-\kappa e_n$ (with $\{e_j\}_j$ the standard basis of $\Rl^n$). To obtain Hardy estimates on the tube with base $\C_n^0$, we need to control the $L^2(\Rl^n,\K^{\otimes n})$-norms of $(A\Omega)_{n,\bla}$, for any $\bla\in\{0,-\kappa e_1,...,-\kappa e_n\}$.
  
	For $\bla=0$, we clearly have $\|(A\Omega)_n\|\leq\|A\|$. For $\bla=-\kappa e_1$, we recall the correspondence (\ref{zusammenhang}), that is, $(A\Omega)_n^{\boldsymbol{\alpha}}(\boldsymbol{\theta})=\tfrac{1}{\sqrt{n!}}\langle z_{\alpha_2
     }^\dagger(\theta_2)\cdots z_{\alpha_n
     }^\dagger(\theta_n)\Omega,[z_{\alpha_1
     }(\theta_1),A]\Omega\rangle$. Hence, by application of Lemma \ref{basicLemma} it follows that
     $$h^{\alpha_1}_{-\lambda}:\theta_1\mapsto\int d\theta_2\cdots d\theta_n (A\Omega)_n^{\alpha_1\alpha_2\dots\alpha_n}(\theta_1-i\lambda,\theta_2,\dots,\theta_n)f^{\alpha_2\dots\alpha_n}(\theta_2,\dots,\theta_n)$$ is in $L^2(\mathbb{R},d\theta)$ for any $f\in L^2(\mathbb{R}^{n-1})\otimes\mathcal{K}^{\otimes n-1}$ and any $0\leq\lambda\leq \pi$. Moreover, we have $\|h_{-\lambda}\|_2\leq  \|f\|_2\|A\|$. By the mean value property,
	\begin{equation*}
		h^{\alpha_1}(\zeta)=\frac{1}{\pi r^2}\int\limits_{D(\zeta,r)}d\theta\, d\lambda\, h^{\alpha_1}(\theta+i\lambda),
	\end{equation*}
	with $D(\zeta,r)\subset\mathbb{R}-i[0,\pi]$ a disc of radius $r$ and center $\zeta$. Straightforward estimates (see~\cite[p.~843]{L08}) then give the pointwise bound 
	\begin{align*}
		\left|h^{\alpha_1}(\theta_1-i\lambda)\right|\leq \left(\tfrac{2}{\pi \min\{\lambda,\pi-\lambda\}}\right)^{1/2} \|f\|_2\|A\|
		\,,\qquad \theta_1\in\mathbb{R}\,,\quad 0<\lambda<\pi\,. 
	\end{align*}
	This implies that for fixed $\balpha,\theta_1,\lambda$, the function $(\theta_2,\dots,\theta_n)\mapsto(A\Omega)_n^{\boldsymbol{\alpha}}(\theta_1-i\lambda,\theta_2,\dots,\theta_n)$ is in $L^2(\mathbb{R}^{n-1})$ with norm bounded by $\left(\tfrac{2}{\pi \min\{\lambda,\pi-\lambda\}}\right)^{1/2} \|A\|$. 
	
	To improve the falloff behavior in the first variable, we consider the shifted operator $U_sAU_s^{-1}$, $s>0$, which has wavefunction $(U_sA\Omega)_n^{\boldsymbol{\alpha}}(\boldsymbol{\zeta})=u_{n,s}^{[\boldsymbol{\alpha}]}(\boldsymbol{\zeta})\cdot(A\Omega)_n^{\boldsymbol{\alpha}}(\boldsymbol{\zeta})$, with $u_{n,s}^{[\boldsymbol{\alpha}]}(\boldsymbol{\zeta}):=\prod_{k=1}^{n}e^{-ism_{[\alpha_k]}\sinh \zeta_k}$. Using the previous pointwise bound on $(A\Omega)_n$, and  $0<\kappa<\frac{\pi}{2}$ as well as $\cosh\te_1\geq1+\frac{1}{2}\te_1^2$, we find 
	\begin{align*}
		\|(U_sA\Omega)_{n,-\kappa\,e_1}\|^2
		&=
		\sum_{\boldsymbol{\alpha}}\int d\theta_1 e^{-2sm_{[\alpha_1]}\sin\kappa\cosh\theta_1}
		\int d\theta_2\cdots d\theta_n\left|(A\Omega)_n^{\boldsymbol{\alpha}}(\theta_1-i\kappa,\theta_2,\dots,\theta_n)\right|^2
		\\
		&\leq 
		\frac{2\|A\|^2}{\pi\,\kappa}
		\sum_{\boldsymbol{\alpha}}\int d\theta_1 e^{-2sm_{\circ}\sin\kappa\,\cosh\theta_1}
		\\
		&\leq
		\frac{2(\dim\K)^n\|A\|^2}{\pi\,\kappa}
		e^{-2sm_{\circ}\sin\kappa}
		\int d\theta_1 e^{-sm_{\circ}(\sin\kappa)\,\te_1^2}
		\\
		&\leq
		\frac{2(\dim\K)^n\|A\|^2}{\sqrt{\pi}\,\kappa}
		\frac{e^{-2sm_{\circ}\sin\kappa}}{\sqrt{sm_\circ\,\sin\kappa}}
		\;=:\;
		a(s,\kappa)^2 (\dim\K)^n\|A\|^2
		\,.
	\end{align*}

	To obtain bounds on $\|(U_sA\Omega)_{n,-\kappa e_j}\|$ for $j=2,...,n$, we need to take the S-symmetry of $(U_sA\Omega)_n$ into account. To this end, recall from Def.~\ref{regularS} that for $0<\lambda\leq\kappa$
	\begin{equation*}
		\sup\limits_{\theta\in\mathbb{R}}\|S(\theta)\|=1,\qquad \sup\limits_{\theta\in\mathbb{R}}\|S(\theta+i\lambda)\|\leq 1,\qquad\sup\limits_{\theta\in\mathbb{R}}\|S(\theta-i\lambda)\|\leq\|S\|_{\kappa}<\infty.
	\end{equation*}
	For any $\sigma\in\mathfrak{S}_n$, the tensor
	$$\zeta_k\mapsto S_n^\sigma(\theta_1,\dots,\theta_{k-1},\zeta_k,\theta_{k+1},\dots,\theta_n),\qquad\theta_k\in\mathbb{R},\qquad k=1,\dots,n,$$
	is, as a (tensor-) product of several factors $S(\te_l-\te_r)$, analytic (at least) in the strip $\Strip(-\kappa,\kappa)$. An estimate on this function is obtained by determining the number of $\zeta_k$-dependent factors $S$ in the above tensor. To this end, recall the fact that any $\sigma\in\mathfrak{S}_n$ can be decomposed (non-uniquely) into a product of inv$(\sigma)$ transpositions $\tau_j\in\mathfrak{S}_n$, where inv$(\sigma)$ is the number of pairs $(i,j)$, $i,j=1,\dots,n$, with $i<j$ and $\sigma(i)>\sigma(j)$. Therefore, we count that the maximal possible number of transpositions which involve the element $k$ is $n-1$. Hence, the representing tensor $S_n^\sigma(\theta_1,\dots,\theta_{k-1},\zeta_k,\theta_{k+1},\dots,\theta_n)$ contains at most $n-1$ factors depending on the variable $\zeta_k$. Moreover, each of those factors is bounded by $\|S\|_{\kappa}$, and all other factors are bounded by 1. In view of of the $S$-symmetry \eqref{permversion} of $(A\Omega)_n$, we therefore have, $j=2,\dots,n$,
	\begin{eqnarray*}
		\|(U_sA\Omega)_{n,-\kappa e_j}\|
		\leq 
		\|S\|_{\kappa}^{n-1}\,\|(U_sA\Omega)_{n,-\kappa e_1}\|
		\leq
		\|S\|_{\kappa}^{n}\,a(s,\kappa)\,(\dim\K)^{n/2}\,\|A\|
		\,.
	\end{eqnarray*}
	It is furthermore clear from our proof that also for imaginary part $\bla$ on the line segments connecting $0$ and $-\kappa e_j$, we have finite $L^2$-norms $\|(U_sA\Omega)_{n,\bla}\|<\infty$. By Malgrange-Zerner type estimates we therefore obtain
	\begin{align}\label{eq:corner-bound}
		\|(U_sA\Omega)_{n,\bla}\|
		&\leq
		\max\big\{\|(U_sA\Omega)_{n},\,\|(U_sA\Omega)_{n,-\kappa e_j}\|\,:\,j=1,...,n\big\}
		\\
		&\leq
		\max\big\{1,a(s,\kappa)\,(\dim\K)^{n/2}\,\|S\|_\kappa^n\big\}\cdot\|A\|
	\end{align}
	for any $\bla$ in the simplex $\C_n^0$.
	
	It is easily checked that $\C_n^0$ contains the cube $(-\frac{\kappa}{n},0)^{\times n}$, i.e., \eqref{eq:corner-bound} holds, in particular, for $\bla\in(-\frac{\kappa}{n},0)^{\times n}$.

	\medskip

	We proceed by considering the tube based on $\mathcal{C}_n^{-\pi}:=(-\pi,\dots,-\pi)+(0,\tfrac{\kappa}{n})^{\times n}$. Since $(U_sA\Omega)^{\boldsymbol{\alpha}}_{n,(-\pi,\dots,-\pi)}=(JA^*(s)\Omega)_n^{\boldsymbol{\alpha}}$, cf. (\ref{zusammenhang}), it follows, due to $\|A^*\|=\|A\|$, immediately that the bound (\ref{eq:corner-bound}) also holds for $\boldsymbol{\lambda}\in\mathcal{C}_n^{-\pi}$.
	
	To show the validity of (\ref{eq:corner-bound}) in the tube with base $\bla_{\pi/2}+(-\frac{\kappa}{n},0)^{\times n}$, we interpolate between $\C_n^0$ and $\C_n^{-\pi}$. Both cubes, $\mathcal{C}_n^0$ and $\mathcal{C}_n^{-\pi}$, are contained in the analyticity domain of $(U_sA\Omega)_n^{\boldsymbol{\alpha}}$, and connected by the line segment $l$ from the point $(0,\dots,0)$ to $(-\pi,\dots,-\pi)$, cf. Figure \ref{fig}.
\begin{figure}[h]
\begin{center}
\begin{tikzpicture}[scale=1.2]
               \begin{scope}%[transparent]    %current
              \fill[gray!30!white] (0,0)--(-.85,0)--(-2.5,-2.5+.85) -- (-2.5,-2.5)--(0,0);
              \fill[gray!30!white] (0,0)--(-2.5,-2.5)--(-1.65,-2.5) -- (0,-.85);
               \end{scope} 
        \begin{scope}%[transparent]    %current
                \draw[fill] (1.2,.2) node{Im$\,\zeta_1$};
\draw[fill] (.5,1) node{Im$\,\zeta_2$};
\draw[fill] (.4,-.45) node{\footnotesize{$\kappa(S)$}};
\draw[fill] (-2.5,.2) node{\footnotesize{$(-\pi,0)$}};
\draw[fill] (-2.8,-2.8) node{\footnotesize{$(-\pi,-\pi)$}};
\draw[fill] (.2,-2.8) node{\footnotesize{$(0,-\pi)$}};
\draw[fill] (.2,.2) node{\footnotesize{$0$}};
\fill[black] (-2.5,-2.5) circle (.3ex);
\fill[black] (-2.5,0) circle (.3ex);
\fill[black] (0,-2.5) circle (.3ex);
\fill[black] (0,0) circle (.3ex);
        \begin{scope}[->]
            \draw (0,-2.5) -- (0,1) node[anchor=north] {};
            \draw (-2.5,0) -- (1,0) node[anchor=east] {};
            \draw (0.1,-.05)--(.1,-.85);
               \draw (.1,-.85)--(0.1,-.05);
        \end{scope}
    \begin{scope}  
    \draw[style=densely dashed] (-0.85,0) --(-2.5,-1.65);
       
    \draw[style=densely dashed] (0,-0.85) --(-1.65,-2.5);
       \draw (-2.5,0) -- (-2.5,-2.5);
      \draw (-2.5,-2.5) -- (0,-2.5);
        \draw[style=densely dashed](0,0)--(-2.5,-2.5);
      \fill[black] (-1.25,-1.25) circle (.3ex);
             \draw[style=densely dashed] (-1.25+.2,-1.25+.2)--(-1.25+.2,-1.25-.2);
             \draw[style=densely dashed] (-1.25-.2,-1.25+.2) -- (-1.25+.2,-1.25+.2);
             \draw[style=densely dashed] (-1.25-.2,-1.25+.2) --(-1.25-.2,-1.25-.2);
               \draw[style=densely dashed] (-1.25-.2,-1.25-.2) --(-1.25+.2,-1.25-.2);
       \draw[style=densely dashed] (0,0)--(-.2,-.2);
   \draw[style=densely dashed] (-.4,0)--(-.4,-.4);
    \draw[style=densely dashed] (-.4,0)--(-.4,-.4);
     \draw[style=densely dashed] (0,-.4)--(-.4,-.4);
 \draw[style=densely dashed] (-2.5+.4,-2.5+.4)--(-2.5,-2.5+.4);
  \draw[style=densely dashed] (-2.5+.4,-2.5)--(-2.5+.4,-2.5+.4);
        \end{scope}
               \end{scope}
\end{tikzpicture}
\end{center}
\caption{Cubes appearing in the proof of Proposition \ref{Corrhardystructure} for the case $n=2$.}\label{fig}
\end{figure}
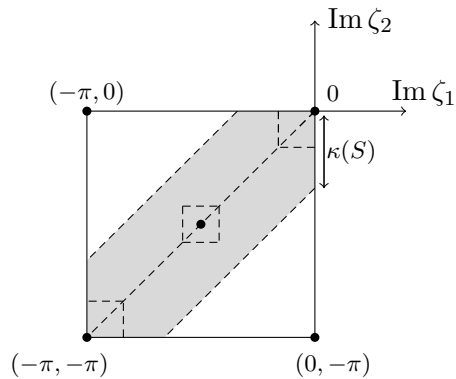\\

\noindent By modular theory we have that $\|(U_sA\Omega)_{n,\boldsymbol{\lambda}}\|\leq\|A\|$ for all $\boldsymbol{\lambda}\in l$. Consequently, the bound (\ref{eq:corner-bound}) also holds for $\boldsymbol{\lambda}$ in the convex closure of $\mathcal{C}_n^0\cup\mathcal{C}_n^{-\pi}\cup l$. But the cube $\boldsymbol{\lambda}_{\pi/2}+(-\tfrac{\kappa}{2n},\tfrac{\kappa}{2n})^{\times n}$ is contained in this region, yielding what is claimed.
\end{proof}

Having verified property (H) (for regular $S$), we may insert the values of $c_n$ and $\upsilon(s,n)$ we found into \eqref{eq:nuclear-norm-X} and \eqref{eq:general-Xi_n-nuclear-norm-bound}. Simplifying the resulting expressions (in particular, with $\cos\frac{\kappa}{2n}\geq\frac{1}{\sqrt{2}}$), we obtain the following corollary. 

\begin{corollary}\label{corollary:small-hardy}
	Let $S\in\SF_0$, $n\in\Nl$, $s>0$. Then $\Xi_n(s)$ is nuclear, and there exist constants $C_1(s)>0$ (depending on $s,\kappa,m_\circ,\dim\K,S$, but not on $n$) and $C_2,C_3>0$ (depending only on $\kappa,m_\circ,\dim\K,S$, but not on $n$ or $s$), such that	
	\begin{align}\label{eq:Xi-n-bound-on-small-tube}
		\|\Xi_n(s)\|_1
		\leq
		C_1(s)\,\left(\frac{C_2\,e^{-C_3\,s}}{s^{1/2}}\right)^n\cdot n^n
		<
		\infty\,.
	\end{align}
\end{corollary}

Explicit forms of $C_1(s),C_2,C_3$ can be obtained easily from our previous estimates. It is however more relevant to consider the dependence of this bound on $n$ (and $s$). In particular, in view of the factor $n^n$ (which originates from the $c_n^{-n}$ in \eqref{eq:nuclear-norm-X}), Thm.~\ref{thm:general-nuclearity} does not yield nuclearity of $\Xi(s)$ for any $s>0$ because \eqref{eq:Xi-n-bound-on-small-tube} is not summable in $n$. We will therefore discuss in the next section how the bound \eqref{eq:Xi-n-bound-on-small-tube} can be improved.

\section{The Intertwiner Property and Examples}\label{section:intertwiners}

Our basic strategy to estimate the nuclear norm of $\Xi_n(s)$ by $\|\Xi_n(s)\|_1\leq\|\Upsilon_{\C_n,s}\|\cdot\|X_{\C_n,s}\|_1$ has led to the bound \eqref{eq:Xi-n-bound-on-small-tube} for regular S-matrices. To improve it in such a way that it becomes summable over $n$ (which would imply nuclearity of $\Xi(s)$ via Thm.~\ref{thm:general-nuclearity}~$ii)$), we will in this section discuss a method for obtaining sharper bounds on $\|X_{\C_n,s}\|_1$. 

The basic idea to be used below is adopted from the scalar case \cite{L08} and consists in a more stringent exploitation of the $S$-symmetry of our Hilbert space. As a motivating example, and as a tool to be used later on, we first state a lemma which shows the effect of $S$-symmetry in the most drastic case of total antisymmetry (corresponding to the case $S=-F$, $F$ the tensor flip), where a Pauli principle becomes effective.

\begin{lemma}\label{lemma:pauli}
     Let $s>0$ and $\C_n\subset\Rl^n$ an open polyhedron as in Def.~\ref{definition:property-h}, with associated operator $X_{\C_n,s}$ \eqref{eq:def-X}. Let $P_n^-$ denote the projection onto the completely antisymmetric subspace of $L^2(\Rl^n,\K\tp{n})\cong L^2(\Rl,\K)\tp{n}$. Then $P_n^-X_{\C_n,s}$ is nuclear, with nuclear norm 
     \begin{align}\label{eq:nuclear-norm-X-fermi}
		\|P_n^-X_{\C_n,s}\|_1
		\leq
		\frac{2^n}{n!}\cdot
		\left(
			\frac{\dim\K}{(\frac{1}{2}\pi m_\circ)^{1/2}}
			\cdot \frac{e^{-\frac{sm_\circ}{2}\cos c_n}}{c_n(s\cos c_n)^{1/2}}\right)^n
			\,,
     \end{align}
     where $m_\circ=\min\{m_{[\alpha]}\,:\,\alpha=1,\dots,\dim\K\}\in(0,\infty)$ is the mass gap \eqref{eq:mass-gap}.	
\end{lemma}
The bound \eqref{eq:nuclear-norm-X-fermi} should be compared with \eqref{eq:nuclear-norm-X}, from which it differs by a factor $2^n/n!$. 
\begin{proof}
	Since $P_n^-$ is bounded, and $X_{\C_n,s}$ is nuclear by Prop.~\ref{lemma:Hardy->Nuclearity}, it is clear that $P_n^-X_{\C_n,s}$ is nuclear as well. To obtain the claimed bound on its nuclear norm, we represent $X_{\C_n,s}$ as in \eqref{x}
	 \begin{eqnarray}
		P_n^-X_{\C,s}
		&=& 
		P_n^-A_\delta\sum_{\boldsymbol{\varepsilon}} 
		\left(
			\hat R_{\delta,r,\varepsilon_1}\otimes\cdots\otimes \hat R_{\delta,r,\varepsilon_n}
		\right)
		\,B_{\delta,r,\boldsymbol{\varepsilon}}\, E_{\boldsymbol{\lambda}-r\boldsymbol{\varepsilon}}
		\,.
     \end{eqnarray}
     The definitions of the operators and parameters $r,\delta$ appearing here can be looked up in the proof of Prop.~\ref{lemma:Hardy->Nuclearity}, but will not be relevant for the following argument. We only recall that the $\hat R_{\delta,r,\eps}$ are positive trace class operators on $L^2(\Rl,\K)$, the operators $A_\delta$, $B_{\delta,r,\eps}$, and $E_{\boldsymbol{\lambda}-r\boldsymbol{\varepsilon}}$ are all bounded with norm at most 1, and $A_\delta$ is an operator that multiplies with a symmetric function, and hence commutes with $P_n^-$. The sum runs over the $2^n$ terms indexed by $\eps_1,...,\eps_n=\pm1$, and therefore commutes with $P_n^-$ as well. 
	
	We are thus left with the task to estimate the trace norm of the ``compressed'' sum $P_n^-\sum_{\boldsymbol{\varepsilon}}(\hat R_{\delta,r,\varepsilon_1}\otimes\cdots\otimes \hat R_{\delta,r,\varepsilon_n})P_n^-$. To this end, we consider the positive trace class operator $Z_{\delta,r}:=(\hat R_{\delta,r,+1}^2+\hat R_{\delta,r,-1}^2)^{1/2}$. By the operator monotonicity of the square root, we have $Z_{\delta,r}\geq \hat R_{\delta,r,\pm1}$, and therefore $\hat R_{\delta,r,\varepsilon_1}\otimes\cdots\otimes \hat R_{\delta,r,\varepsilon_n}\leq (Z_{\delta,r})\tp{n}$ for any $\eps_1,...,\eps_n$. Denoting the trace over the antisymmetric subspace ${P_n^- L^2(\Rl^n,\K\tp{n})}$ by $\Tr_-$, this yields by the monotonicity of the trace 
	\begin{align*}
		\Big\|P_n^-\sum_{\boldsymbol{\varepsilon}}(\hat R_{\delta,r,\varepsilon_1}\otimes\cdots\otimes \hat R_{\delta,r,\varepsilon_n})P_n^-\Big\|_1
		&=
		\Tr_-
		\Big(
			\sum_{\boldsymbol{\varepsilon}}\hat R_{\delta,r,\varepsilon_1}\otimes\cdots\otimes \hat R_{\delta,r,\varepsilon_n}
		\Big)
		\leq
		2^n\,\Tr_-
		\big(Z_{\delta,r}\tp{n}
		\big)
		\,.
	\end{align*}
	The crucial effect of the antisymmetrization is now the estimate $\Tr_-\big(Z_{\delta,r}\tp{n}\big)\leq\|Z_{\delta,r}\|_1^n/n!$ (see, for example, \cite[Lemma~3.3]{Simon:2005}), to be compared with $\Tr(Z_{\delta,r}\tp{n})=\Tr (Z_{\delta,r})^n$.
	
	To complete the proof, we estimate the trace norm of $Z_{\delta,r}$ according to $\|Z_{\delta,r}\|_1\leq  \|\hat R_{\delta,r,+1}\|_1+\|\hat R_{\delta,r,-1}\|_1$ \cite{kosaki1984continuity}. As $ \|\hat R_{\delta,r,-1}\|_1=\|\hat R_{\delta,r,+1}\|_1$ by Lemma~\ref{traceclassneu} and Prop.~\ref{lemma:Hardy->Nuclearity}, this yields the bound $\|P_n^-X_{\C_n,s}\|_1\leq 4^n\|R_{\delta,r,+1}\|_1^n/n!$, which differs from the bound underlying Prop.~\ref{lemma:Hardy->Nuclearity} by a factor $2^n/n!$. Estimating $\|R_{\delta,r,+1}\|_1$ as before therefore gives the claimed bound \eqref{eq:nuclear-norm-X-fermi}.
\end{proof}

For general $S\in\SF_0$, the effect of the projection onto an $S$-symmetric subspace $\Hil_n\subset L^2(\Rl,\K)\tp{n}$ is not as drastic as the total antisymmetrization. In case $S(0)=-F$, the exchange relations \eqref{exchange} are however ``close'' to the CAR relations for small rapidities. In connection with the fact that all the relevant functions in our analysis drop off quickly for large arguments, we can therefore expect some kind of ``effective Pauli principle'' for $S(0)=-F$.

\medskip

To describe this, we consider the representation $D_n^S$ of $\mathfrak{S}_n$ (see Lemma~\ref{lemma:Dn}), which we denote here $D_n^S$ instead of $D_n$ to emphasize the dependence on $S\in\SF_0$. By its definition, the tensors $S_n^\pi(\bte)$, $\pi\in{\mathfrak S}_n$, $\bte\in\Rl^n$, satisfy the cocycle equation $S_n^{\pi\sigma}(\bte)=S_n^\pi(\bte)S_n^\sigma(\bte_\pi)$, where $\pi,\sigma\in{\mathfrak S}_n$, and $\bte_\pi=(\te_{\pi(1)},...,\te_{\pi(n)})$.

Given two S-matrices $S,\tilde S\in\SF$, it is then straightforward to verify that the piecewise defined tensor-valued functions
\begin{align}
	\hat\I^{S,\tilde S}_n(\bte):=\tilde S^\pi_n(\bte)S^\pi_n(\bte)^{-1}\,,\qquad \te_{\pi(1)}<...<\te_{\pi(n)}\,,
\end{align}
induce unitary multiplication operators 
\begin{align}\label{eq:discontinuous-intertwiner}
	\hat\I^{S,\tilde S}_n:L^2(\Rl^n,\K\tp{n})\to L^2(\Rl^n,\K\tp{n})
	\,,\qquad
	(\hat\I^{S,\tilde S}_n\Psi)(\bte)
	:=
	\hat\I^{S,\tilde S}_n(\bte)\Psi(\bte)\,,
\end{align}
that {\em intertwine} the representations $D_n^S$ and $D_n^{\tilde S}$, i.e. 
\begin{align}
	\hat\I^{S,\tilde S}_n D_n^S(\sigma)=
	D_n^{\tilde S}(\sigma)\hat\I^{S,\tilde S}_n\,,\qquad \sigma\in{\mathfrak S}_n\,.
\end{align}
By means of these intertwiners, one can therefore pass between different $S$-symmetrizations; in particular, between the symmetry given by some $S\in\SF$ and its value at zero, which is also an $S$-matrix $S(0)\in\SF$. This procedure is efficient for our nuclearity estimates if $S(0)=-F$, so that a Pauli principle can be used, and if the intertwiners also preserve the Hardy space structure established so far.

We isolate the relevant properties in the following definition.

\begin{definition}\label{definition:property-I}
     Let $S\in\SF_0$ be regular, and choose polyhedra $\C_n$ such that property (H) holds (see Def.~\ref{definition:property-h}). Then $S$ is said to have the intertwining property if for any $n\in\Nl$, there exists an analytic tensor-valued function
     \begin{align}
	  \I_n:\Rl^n+i\C_n\to\B(\K^{\otimes n})
     \end{align}
	such that:
	\begin{enumerate}
		\item There exists a constant $\gamma>0$ such that for any $n\in\Nl $ and any $\bzeta\in\Rl^n+i\C_n$,
		\begin{align}
			\|\I_n(\bzeta)\|_{\B(\K^{\otimes n})} \leq \gamma^n\,.
		\end{align}
		\item There exists a constant $\tilde\gamma>0$ such that for any $n\in\Nl$ and any $\bte\in\Rl^n$, the tensor $\I_n(\bte+i\bla_{\pi/2})$ is invertible, and 
		\begin{align}
			\|\I_n(\bte+i\bla_{\pi/2})^{-1}\|_{\B(\K^{\otimes n})} \leq \tilde\gamma^n\,.
		\end{align}
		\item Multiplication by the tensor-valued function $\bte\mapsto\mathcal{I}_n(\bte+i\bla_{\pi/2})$ maps the $S$-symmetric subspace $\Hil_n^S\subset L^2(\Rl^n,\K\tp{n})$ onto the $S(0)$-symmetric subspace $\Hil_n^{S(0)}\subset L^2(\Rl^n,\K\tp{n})$.
	\end{enumerate}
\end{definition}

The actual intertwining property is contained in item $iii)$, which is, in particular, satisfied if multiplication by $\I_n(\,\cdot\,+i\bla_{\pi/2})$ intertwines $D_n^S$ and $D_n^{S(0)}$. 

It has to be noted that the simple intertwiner \eqref{eq:discontinuous-intertwiner} does in general {\em not} have the required analyticity properties. If, however, an ``analytic intertwiner'' as specified in Def.~\ref{definition:property-I} can be found, modular nuclearity follows, as we show next.

\begin{theorem}\label{mainTheorem}
	Let $S\in\mathcal{S}_0$ be regular, and suppose furthermore that $S$ has the intertwining property, and $S(0)=-F$. Then, there exists $s_{\rm{min}}<\infty$ such that $\Xi(s):\mathcal{F}(W_R)\rightarrow\mathscr{H}$ is nuclear for all $s>s_{\rm{min}}$.% Thus, in these models the corresponding local algebras $\mathcal{F}(\mathcal{O}_{x,y})=\mathcal{F}(W_R+x)\cap\mathcal{F}(W_L+y)$ have $\Omega$ as a cyclic vector, for each double cone $\mathcal{O}_{x,y}=(W_R+x)\cap(W_L+y)$ with $y-x\in W_R$ and $-(y-x)^2>s_{\text{min}}^2$.
\end{theorem}
\begin{proof}
	We decompose $\Xi_n(s)=X_{\C_n,s}\Upsilon_{\C_n,s}$ as in Thm.~\ref{thm:general-nuclearity}, and aim at improving the estimate on the nuclear norm of $X_{\C_n,s}$ \eqref{eq:def-X} by using the intertwining tensors $\I_n$ from Def.~\ref{definition:property-I}. We write the same symbol $\I_n$ for the operator on $H^2(\Rl^n+i\C_n,\K\tp{n})$ which multiplies with this tensor, and $\I_{n,\bla_{\pi/2}}$ for the operator on $L^2(\Rl^n,\K\tp{n})$ which multiplies with $\bte\mapsto\I_n(\bte+i\bla_{\pi/2})$. 
	
	As before, we may replace all masses $m_{[\alpha_k]}$ in \eqref{eq:def-X} by their minimum $m_\circ$ to obtain upper bounds on nuclear norms, and denote  $\tilde X_{\C_n,s}$ the resulting operator. Since $\I_n$ is analytic and bounded on the tube based on $\C_n$, we then have $\I_{n,\bla_{\pi/2}}\tilde X_{\C_n,s}=\tilde X_{\C_n,s}\I_n$ (cf. definition \eqref{eq:def-X}). 
	
	In view of Def.~\ref{definition:property-I}~$iii)$, we furthermore have that $\I_{n,\bla_{\pi/2}}$ satisfies $\I_{n,\bla_{\pi/2}}\Hil_n^S=\Hil_n^{S(0)}=\Hil_n^{-F}$, the $n$-fold totally antisymmetric tensor power of $\Hil_1$. This implies $\I_{n,\bla_{\pi/2}}\tilde X_{\C_n,s}=P_n^-\I_{n,\bla_{\pi/2}}\tilde X_{\C_n,s}=P_n^-\tilde X_{\C_n,s}\I_n$, with $P_n^-$ the orthogonal projection $\Hil_1\tp{n}\to\Hil_n^{-F}$ as in Lemma~\ref{lemma:pauli}.
	
	Using the bounds from Def.~\ref{definition:property-I}~$i),ii)$, we find
	\begin{align*}
		\|X_{\C_n,s}\|_1
		\leq
		\|\tilde X_{\C_n,s}\|_1
		=
		\|\I_{n,\bla_{\pi/2}}^{-1}\I_{n,\bla_{\pi/2}}\tilde X_{\C_n,s}\|_1
		\leq
		\tilde\gamma^n\,\|P_n^-\tilde X_{\C_n,s}\I_n\|_1
		\leq
		\tilde\gamma^n\,\gamma^n\,\|P_n^-\tilde X_{\C_n,s}\|_1
		\,.
	\end{align*}
	The last trace norm can now be estimated with Lemma~\ref{lemma:pauli}. We keep the same bounds as before on $\|\Upsilon_{\C_n,s}\|$ in $\|\Xi_n(s)\|_1\leq\|X_{\C_n,s}\|_1\|\Upsilon_{\C_n,s}\|$. Since Lemma~\ref{lemma:pauli} gives an improvement by a factor $2^n/n!$ in comparison to \eqref{eq:nuclear-norm-X}, our new bound differs from the one in Corollary \eqref{corollary:small-hardy} by a factor $2^n\gamma^n\tilde\gamma^n/n!$.
	
	Thus we have 
	\begin{align}\label{eq:Xi-n-bound-on-small-tube-with-pauli}
		\|\Xi_n(s)\|_1
		\leq
		C_1(s)\,\left(\frac{C_2\,e^{-C_3\,s}}{s^{1/2}}\right)^n\cdot \frac{n^n}{n!}
		<
		\infty\,,
	\end{align}
	with constants $C_1(s)>0$ (depending on $s,\kappa,m_\circ,\dim\K,S$, but not on $n$) and $C_2,C_3>0$ (depending only on $\kappa,m_\circ,\dim\K,S$, but not on $n$ or $s$).
	
	By Stirling's formula, $n^n/n!\leq e^n/\sqrt{2\pi}$. Inserting this estimate, it is clear that \eqref{eq:Xi-n-bound-on-small-tube-with-pauli} is summable over $n$ for all $s$ such that $\frac{e\,C_2\,e^{-C_3s}}{s^{1/2}}<1$, i.e. for $s$ larger than some minimal value~$0<s_{\rm min}<\infty$. By Thm.~\ref{thm:general-nuclearity}~$ii)$, this implies nuclearity of $\Xi(s)$.
\end{proof}

The intertwiner property is known to hold for all regular {\em scalar} S-matrices (i.e., with $\K=\Cl$) \cite{L08}. A complete analysis of the intertwining properties of the cocycles $S_n^\pi$ will be presented elsewhere. For the purposes of the present article, we restrict ourselves to giving some examples of non-scalar S-matrices which have the intertwiner property.

These examples will be the so-called (regular) {\em diagonal} S-matrices (see, for example, \cite{jimbo1986quantumr, liguoriLetters}). In our terminology, they are given by massive, neutral particles with an $N$-fold internal degree of freedom, i.e. $m_{[\alpha]}=m>0$ and and $\overline{\alpha}=\alpha$, $\alpha=1,\dots,N:=\dim\K$. The S-matrix is in this case of the form 
\begin{equation}\label{diagS}
	S(\theta)^{\alpha\beta}_{\gamma\eta}:=\omega_{\alpha\beta}(\theta)\delta^\alpha_\eta\delta^\beta_\gamma,\qquad\text{(no summation over $\alpha,\beta$)},
\end{equation}
which corresponds to a diagonal $(N^2\times N^2)$-matrix when multiplied by the flip $F$, and therefore solves the Yang-Baxter equation.

When the coefficients $\omega_{\alpha\beta}$ are taken to be analytic bounded functions $\Strip(-\kappa,\pi+\kappa)$ for some $0<\kappa<\tfrac{\pi}{2}$, and required to satisfy
\begin{equation}\label{omega}
	\overline{\omega_{\alpha\beta}(\theta)}=\omega_{\alpha\beta}(\theta)^{-1}=\omega_{\beta\alpha}(-\theta)=\omega_{\alpha\beta}(i\pi+\theta)
	\,,
\end{equation}
then \eqref{diagS} defines a regular S-matrix $S\in\SF_0$. All these requirements on the coefficient function can be satisfied, in particular, if we take $\omega_{\alpha\beta}=\omega_{\beta\alpha}$ to be {\em scattering functions}, i.e. regular S-matrices for $\K=\Cl$. 

\begin{proposition}\label{lemDiag}
	Let $S\in\SF_0$ be a regular diagonal S-matrix with symmetric coefficients, i.e. $\omega_{\alpha\beta}=\omega_{\beta\alpha}$ in \eqref{diagS}. Then $S$ has the intertwining property (Def.~\ref{definition:property-I}).
\end{proposition}
\begin{proof}
	Let us define the intertwiners $\I_n$. We set $\I_0=1$, $\I_1(\zeta)=1_1=\id_\K$). For $n\geq2$, we note that in view of \eqref{omega}, we have $\omega_{\alpha\beta}(0)=\pm1$. We therefore have $\omega_{\alpha\beta}(\te)=\eps_{\alpha\beta}\rho_{\alpha\beta}(\te)$ with $\eps_{\alpha\beta}=\pm1$, and $\rho_{\alpha\beta}(0)=1$. We define	
	\begin{align}\label{eq:In}
		\I_n(\zeta_1,...,\zeta_n)^\balpha_\bbeta
		:=
		\prod_{1\leq l<r\leq n}\left(\eps_{\alpha_l\alpha_r}
		\sqrt{\rho_{\alpha_l\alpha_r}(\zeta_l-\zeta_r)}\right)\,(1_n)^\balpha_\bbeta
		\,,
	\end{align}
	no sums over indices are implied here. As explained in the scalar case \cite[Lemma~5.7]{L08}, these $\I_n$ are analytic on $\Rl^n+i(-\frac{\kappa}{2},\frac{\kappa}{2})^{\times n}$. By the assumed regularity of $S$, we have $|\omega_{\alpha\beta}(\zeta)|\leq\gamma$ for some $\gamma>0$, and all $\zeta\in\Strip(-\kappa,0)$, whereas $|\omega_{\alpha\beta}(\zeta)|\leq1$ for all $\zeta\in\Strip(0,\kappa)$. In the same manner as in Prop.~\ref{Corrhardystructure}, we therefore obtain $\|\I_n(\bzeta)\|\leq\gamma^n$ for $\bzeta\in\Rl^n+i(-\frac{\kappa}{2n},\frac{\kappa}{2n})^{\times n}$. This establishes property $i)$ of Def.~\ref{definition:property-I}.
	
	For property $ii)$, we note that $\I_n(\bte+i\bla_{\pi/2})$ is unitary for all $\bte\in\Rl^n$ because \eqref{eq:In} depends only on differences of rapidities, and for real arguments, the $\omega_{\alpha\beta}$ are phase factors. So property~$ii)$ holds with $\tilde\gamma=1$.
	
	$iii)$ We claim that $\I_{n,\bla_{\pi/2}}$ has the intertwining property $\I_{n,\bla_{\pi/2}}D_n^S(\sigma)=D_n^{S(0)}(\sigma)\I_{n,\bla_{\pi/2}}$ for any $\sigma\in\mathfrak{S}_n$. Since the projection $P_n^S$ is the mean over the representation $D_n^S$, and similarly for $P_n^{S(0)}$, the intertwining property implies 
	\begin{align*}
		\I_{n,\bla_{\pi/2}}\Hil_n^S
		=
		\I_{n,\bla_{\pi/2}}P_n^S\Hil_1\tp{n}
		=
		P_n^{S(0)}\I_{n,\bla_{\pi/2}}\Hil_1\tp{n}
		=
		P_n^{S(0)}\Hil_1\tp{n}
		\,,
	\end{align*}
	and thus item $iii)$ in Def.~\ref{definition:property-I}.
	
	To prove the intertwining property, we first note that since the transpositions generate the symmetric group, it is sufficient to demonstrate $\I_{n,\bla_{\pi/2}}D_n^S(\tau_k)=D_n^{S(0)}(\tau_k)\I_{n,\bla_{\pi/2}}$, $k=1,...,n-1$. In view of the definition \eqref{reprpermugroup}, this in turn is equivalent to 
	\begin{align}
		\I_{n,\bla_{\pi/2}}(\te_1,...,\te_n)\, S(\te_{k+1}-\te_k)_{n,k}
		=
		S(0)_{n,k}\,\I_{n,\bla_{\pi/2}}(\te_1,...,\te_{k+1},\te_k,...,\te_n)
		\,,
	\end{align}
	where we have used the shorthand notation \eqref{kurznotation}.
	
	Using $\omega_{\alpha\beta}(-\zeta)=\omega_{\alpha\beta}(\zeta)^{-1}$ and $\omega_{\alpha\beta}=\omega_{\beta\alpha}$ (the same properties hold for the $\rho_{\alpha\beta}$), we compute (no summation convention used here)
	\begin{align*}
		\big(\I_{n,\bla_{\pi/2}}(\bte)&S(\te_{k+1}-\te_k)_{n,k}\big)^{\balpha}_\bbeta
		= 
		\prod_{1\leq l<r\leq n}\left(\eps_{\alpha_l\alpha_r}\sqrt{\rho_{\alpha_l\alpha_r}(\theta_l-\theta_r)}\right)S^{\alpha_k\alpha_{k+1}}_{\beta_k\beta_{k+1}}(\theta_{k+1}-\theta_k)
		\,
		\delta^\balpha_\bbeta
		\\
		&= 
		\prod_{1\leq l<r\leq n}\left(\eps_{\alpha_l\alpha_r}
		\sqrt{\rho_{\alpha_l\alpha_r}(\theta_l-\theta_r)}\right)
		\eps_{\alpha_k\alpha_{k+1}}\rho_{\alpha_k\alpha_{k+1}}(\theta_{k+1}-\theta_k)
		\,
		(F_{n,k})^\balpha_\bbeta
		\\
		&=\prod_{\stackrel{1\leq l<r\leq n}{(l,r)\neq (k,k+1)}}
		\left(\eps_{\alpha_l\alpha_r}\sqrt{\rho_{\alpha_l\alpha_r}(\theta_l-\theta_r)}\right)\sqrt{\rho_{\alpha_k\alpha_{k+1}}(\theta_{k+1}-\theta_k)}(F_{n,k})^{\boldsymbol{\alpha}}_{\boldsymbol{\beta}}
		\\
		&=\eps_{\alpha_k\alpha_{k+1}}(F_{n,k})^{\boldsymbol{\alpha}}_{\boldsymbol{\beta}}\, \I_{n,\bla_{\pi/2}}(\te_1,...,\te_{k+1},\te_k,...,\te_n)^\balpha_\bbeta
		\\
		&=
		\big(S(0)_{n,k}\big)^{\boldsymbol{\alpha}}_{\boldsymbol{\beta}}\, \I_{n,\bla_{\pi/2}}(\te_1,...,\te_{k+1},\te_k,...,\te_n)^\balpha_\bbeta
		\,.
	\end{align*}
	This finishes the proof.
\end{proof}

By choosing arbitrary scattering functions $\omega_{\alpha\beta}$ with $\omega_{\alpha\beta}(0)=-1$, we have therefore found a large family of S-matrices to which Thm.~\ref{mainTheorem} applies.

The above construction of intertwiners does however not carry over to more general $S\in\SF_0$ in a straightforward manner, mainly due to the fact that various S-factors do not commute. However, there are indications supporting the conjecture that for more general regular scattering functions $S\in\SF_0$, the intertwining property holds on the tubes based on $\bla_{\pi/2}+(-\frac{\kappa}{2n},\frac{\kappa}{2n})^{\times n}$.

\bigskip

As a particularly prominent non-diagonal example, we mention the S-matrices of the $O(N)$-invariant nonlinear $\sigma$-models. In this case, the gauge group is $G=O(N)$ in its defining self-conjugate irreducible representation on $\mathcal{K}=\mathbb{C}^N$, $N\geq 3$. Hence, we have, in particular, $\overline{\alpha}=\alpha$ and, moreover, $m_{[\alpha]}=m$, $\alpha=1,\dots,N$.\par
The derivation of the $O(N)$ nonlinear $\sigma$-model S-matrix relies on the existence of a stable $O(N)$-vector multiplet of massive particles with equal masses $m$. As shown by the Zamolodchikov brothers \cite{Zam78}, by exploiting the $O(N)$-symmetry, the corresponding S-matrices can been determined up to CDD ambiguities and the maximal analytic solutions are of the form
\begin{equation}\label{s-matrixsigma}
\begin{aligned}
     S_N(\theta)^{\alpha\beta}_{\gamma\eta}&=\sigma_1(\theta)\delta^{\alpha}_\beta\delta^\gamma_\eta+\sigma_2(\theta)\delta^\alpha_\gamma\delta^\beta_\eta+\sigma_3(\theta)\delta^\alpha_\eta\delta^\beta_\gamma,
\end{aligned}
\end{equation}
with functions $\sigma_k:\mathbb{R}\rightarrow\mathbb{C}$, $k=1,2,3$, given by
\begin{equation}
\begin{aligned}
\sigma_2(\theta)&=\frac{\Gamma\left(\tfrac{1}{N-2}-i\tfrac{\theta}{2\pi}\right)\Gamma\left(\tfrac{1}{2}-i\tfrac{\theta}{2\pi}\right)\Gamma\left(\tfrac{1}{2}+\tfrac{1}{N-2}+i\tfrac{\theta}{2\pi}\right)\Gamma\left(1+i\tfrac{\theta}{2\pi}\right)}{\Gamma\left(\tfrac{1}{2}+\tfrac{1}{N-2}-i\tfrac{\theta}{2\pi}\right)\Gamma\left(-i\frac{\theta}{2\pi}\right)\Gamma\left(1+\tfrac{1}{N-2}+i\tfrac{\theta}{2\pi}\right)\Gamma\left(\tfrac{1}{2}+i\tfrac{\theta}{2\pi}\right)},\\
\sigma_1(\theta)&=-\frac{2\pi i}{N-2}\cdot\frac{\sigma_2(\theta)}{i\pi-\theta},\\
\sigma_3(\theta)&=\sigma_1(i\pi-\theta).
\end{aligned}
\end{equation}

As has been observed in \cite{LS}, this S-matrix is regular and ``fermionic'' in the sense that $S_N(0)=-F$. Furthermore, as the $\sigma$-model S-matrix $S_N$ \eqref{s-matrixsigma} is a $\te$-dependent linear combination of three $O(N)$-invariant operators, which commute among themselves, one can show that intertwiners $\I_n$ as in Def.~\ref{definition:property-I} exist for $S_N$, at least up to tensor level $n=2$ \cite{Alazzawi:2014}. The analysis of the intertwiner problem for $n\geq3$ requires new input, and we must leave this question for another investigation.

\footnotesize
\bibliography{Literatur}
\bibliographystyle{plain}

\end{document}